\documentclass[journal]{IEEEtran}

\usepackage{graphicx}
\usepackage{amsmath}
\usepackage{amssymb}
\usepackage{amsfonts}
\usepackage{amsthm}
\usepackage{cite}
\usepackage{array}
\usepackage{subfigure}
\usepackage{arydshln}
\usepackage{color}
\usepackage{multirow}
\usepackage{threeparttable}

\newtheorem{defi}{Definition}
\newtheorem{theo}{Theorem}

\usepackage{bm,bbm}

\newcommand{\la}{\lambda}
\newcommand{\li}{\lambda_i}

\begin{document}
\title{Two-Channel Critically-Sampled Graph Filter Banks\\With Spectral Domain Sampling}
\author{Akie~Sakiyama, Kana~Watanabe, Yuichi~Tanaka, and~Antonio~Ortega
\thanks{This work was partially funded by JST PRESTO under grant JPMJPR1656, JSPS KAKENHI under Grant JP16J04362, NSF under grants CCF-1410009 and CCF-1527874, and Institute of Global Innovation Research in Tokyo University of Agriculture and Technology.}
\thanks{\emph{A. Sakiyama and K. Watanabe contributed equally to this work.}}
\thanks{A. Sakiyama, K. Watanabe, and Y. Tanaka are with the Graduate School of BASE, Tokyo University of Agriculture and Technology, Koganei, Tokyo, 184-8588 Japan. Y. Tanaka is also with PRESTO, Japan Science and Technology Agency, Kawaguchi, Saitama, 332-0012, Japan (email: nbkn@msp-lab.org; sakiyama@msp-lab.org; ytnk@cc.tuat.ac.jp).}
\thanks{A. Ortega is with the University of Southern California, Los Angeles, CA 90089 USA (email: antonio.ortega@sipi.usc.edu)  and also a visiting researcher at the Institute of Global Innovation Research, Tokyo University of Agriculture of Technology.}
}

\markboth{}{}
\maketitle

\begin{abstract}
We propose two-channel critically-sampled filter banks for signals on undirected graphs that utilize spectral domain sampling. Unlike conventional approaches based on vertex domain sampling, our transforms have the following desirable properties: 1) perfect reconstruction regardless of the characteristics of the underlying graphs and graph variation operators and 2) a symmetric structure; i.e., both analysis and synthesis filter banks are built using similar building blocks. Along with the structure of the filter banks, this paper also proves the general criterion for perfect reconstruction and theoretically shows that the vertex and spectral domain sampling coincide for a special case.
The effectiveness of our approach is evaluated by comparing its performance in nonlinear approximation and denoising with various conventional graph transforms.
\end{abstract}


\IEEEpeerreviewmaketitle

\section{Introduction}
\subsection{Motivation}
Graph signal processing focuses on \textit{graph signals}, discrete signals defined on the vertices of a graph \cite{Shuman2013, Sandry2013, Ortega2018}. Graph signals can represent a broad range of irregularly structured data, such as signals on brain, sensor, social, and traffic networks, point cloud attributes, and images/videos. Developing sparse representations for these signals by using appropriate bases or frames is important, because these signals are often high-dimensional. Promising applications for such sparse representations of graph signals include feature extraction \cite{Leonar2013,Miller2015,Shahid2016}, denoising \cite{Onuki2016,Pang2017,Yamamo2016, Cheung2018}, compression \cite{Shuman2016,Hu2015,Liu2017,Cheung2018}, and others in many different areas \cite{Zhang2014,Ono2015,Higash2016,Bronst2017,Segarr2016,Rustam2013}.

As is the case with classical signal processing, multiscale transforms or dictionaries are important tools for achieving sparse representations of graph signals. Sampling strategies are crucial for controlling the level of redundancy in the graph signal in a multiscale signal representation. Undecimated transforms \cite{Hammon2011,Sakiya2016a,Shuman2015} require a significant storage overhead for the transformed coefficients. Other approaches, which can achieve different trade-offs in terms of redundancy, performance, computation cost, and storage, are oversampled \cite{Tay2015,Sakiya2014a,Tanaka2014a,Sakiya2016a}, critically sampled (CS) \cite{Narang2012,Narang2013,Tay2017,Tay2017a,Trembl2016,Jin2017,Sakiya2016a,Ekamba2015,Teke2016b} and undersampled transforms \cite{Sakiya2016b}.

For time domain signals, downsampling by a factor of two followed by upsampling by a factor of two corresponds to replacing every other sample by zero. In the frequency domain, it is well known that the resulting signal has two components; the original frequency content of the signal and an aliasing term (a modulated version of the original spectrum) \cite{Vaidya1993, Oppenh2009, Vetter2014}. In other words, we can perform the sampling in the frequency (DFT) domain that yields the same signal as the downsampled-then-upsampled signal in the time domain. This can be done by making the shifted replicas of the original spectrum with the period $2\pi$.

Sampling of graph signals can also be intuitively defined in the vertex domain \cite{Pesens2008, Chen2015, Wang2015, Anis2016, Tsitsv2016}. Downsampling-then-upsampling in the vertex domain sampling means replacing some of the values on the graph vertices by zero. However, in contrast to the classical case, the resulting signal in the graph frequency domain generally has a spectrum that cannot be separated into main and aliasing components even when the signal is bandlimited. This is the main difference between sampling in the time domain and that in the vertex domain. Sampling in the graph frequency domain \cite{Tanaka2018} has  been recently proposed as an extension to graph signal sampling of frequency domain sampling developed in classical signal processing.

The main contribution of this paper is the design of CS graph filter banks (GFBs) using sampling in the graph frequency domain \cite{Tanaka2018}. This paper is a significantly extended version of our preliminary study \cite{Watana2018}, which first introduced this idea. With respect to \cite{Watana2018} we have added theoretical results with proofs as well as 
much more comprehensive experimental results. The main advantages of our proposed approach are\footnote{See Section \ref{sec:related} for a more detailed discussion of differences between the proposed CS GFBs and existing designs.}:
\begin{itemize}
 \item Perfect reconstruction is guaranteed for {\em any graph and for any variation operator} as long as the operator is diagonalizable and has real eigenvalues.
 \item The (frequency domain) sampling that leads to perfect reconstruction is {\em unique}, while the analysis and synthesis operations have the same complexity and a {\em matrix inversion is not required to compute the reconstruction operator}.
\end{itemize}
Moreover, we show that the GFBs obtained in the vertex domain and those obtained using spectral domain sampling are identical in some special cases. We also assess their performance through experiments on denoising and nonlinear approximation.

The rest of the paper is organized as follows. We review related work in Section~\ref{sec:related}. Sampling methods in the vertex and spectral domains are introduced in Section \ref{sec:II}. Section \ref{sec:III} reviews the conventional CS GFBs and graph wavelet transforms (GWTs). The proposed CS GFBs are presented in Section \ref{sec:IV} along with the octave-band structure and polyphase representation. The relationship between the vertex and spectral domain sampling approaches is studied in Section \ref{sec:V}. Section \ref{sec:VI} presents a few potential applications of the proposed CS GFBs, together with comparisons with the conventional methods.
Finally, Section \ref{sec:VII} is the conclusion.

\subsection{Notation}
A graph $\mathcal{G}=( \mathcal{V},\mathcal{E} )$ consists of a set of edges $\mathcal{E}$ and vertices $\mathcal{V}$, where the number of vertices is $N = |\mathcal{V}|$. We consider undirected graphs without self-loops and nonnegative edge weights. A graph signal is a function $f: \mathcal{V} \rightarrow \mathbb{R}$, and it can be represented in vector form $\mathbf{f} \in \mathbb{R}^N$, whose $n$th sample $f[n]$ is regarded as a signal value on the $n$th vertex of the graph.

$\mathbf{A} \in \mathbb{R}^{N \times N}$ is an adjacency matrix of the graph whose $(i,j)$th-element $a_{ij}$ represents the weight of the edge between the $i$th and $j$th vertices. $\mathbf{D} \in \mathbb{R}^{N \times N}$ is a diagonal degree matrix whose elements are defined as $d_{ii} = \sum_{j} a_{ij}$. The combinatorial and symmetric normalized graph Laplacians are defined as $\mathbf{L} = \mathbf{D}- \mathbf{A}$ and $\bm{\mathcal{L}} = \mathbf{D}^{-1/2}\mathbf{L}\mathbf{D}^{-1/2}$, respectively. Since a graph Laplacian is a real symmetric matrix, the eigendecomposition of $\mathbf{L}$ (or $\bm{\mathcal{L}}$) can always be represented as $\mathbf{L} = \mathbf{U \Lambda U}^{\top}$, where $\mathbf{U} = \left[ \mathbf{u}_{0},\mathbf{u}_{1},\ldots ,\mathbf{u}_{N-1} \right]$ is an eigenvector matrix, $\mathbf{\Lambda} = \text{diag} \left( \lambda_0, \lambda_1, \ldots , \lambda_{N-1} \right)$ is an eigenvalue matrix having eigenvalues $\lambda_{i} $ ($i = 0,1, \cdots , N-1$) of $\mathbf{L}$ as diagonal elements, and $\cdot^{\top}$ represents the transpose of a matrix.

For a symmetric normalized graph Laplacian, its eigenvalues are bounded in $\lambda_i \in [0,2]$. In addition, the maximum eigenvalue becomes $\lambda_{\max} = 2$ and the eigenvalues are distributed symmetrically with respect to $\lambda = 1$ only for the bipartite case \cite{Chung1997}.

The graph Fourier transform (GFT) is defined as 
\begin{equation}
\widetilde{f}[i] = \langle \mathbf{u}_i, \mathbf{f}  \rangle = \sum^{N-1}_{n=0} u_i[n]f[n],
\end{equation}
while the other definitions of the GFT, such as those in \cite{Deri2017, Giraul2018}, can be used as long as the GFT matrix is nonsingular.

\begin{table*}[tp]
\caption{List of Critically-Sampled Graph Wavelets and Filter Banks. Properties Are Described at the Bottom of the Table.}\label{tb:list_gwt}
\centering
\begin{threeparttable}
\begin{tabular}{l||c|c|c|c|c|c|c|c}
\hline
GFBs / Properties& Analysis\tnote{1} & Synthesis\tnote{2} & Filter\tnote{3} & Graphs\tnote{4}  &VO\tnote{5} &Orth.\tnote{6} & Comp.\tnote{7} & PR\tnote{8} \\\hline
GraphQMF \cite{Narang2012}& Filt. $\rightarrow$ VS  & VS $\rightarrow$ Filt.& S & Bipartite&SNL& O& & \checkmark\\
GraphBior \cite{Narang2013}&Filt. $\rightarrow$ VS & VS $\rightarrow$ Filt.&S & Bipartite&SNL& B& \checkmark &\checkmark\\
GraphFC \cite{Sakiya2016a}&Filt. $\rightarrow$ VS & VS $\rightarrow$ Filt.&S &Bipartite&SNL& O, B& & \checkmark\\
Nearorth \cite{Tay2017a}&Filt. $\rightarrow$ VS & VS $\rightarrow$ Filt. &S &Bipartite&SNL& B& \checkmark & \checkmark\\
$\Omega$-structure \cite{Teke2016b}&Filt. $\rightarrow$ VS & VS $\rightarrow$ Filt.& S &$\Omega$-structure\tnote{9} & Adjacency &B& & \checkmark\\
Generalized Spline \cite{Ekamba2015} & Filt. $\rightarrow$ VS & Interpolation & S & Any & NAD & N/A & \checkmark & \checkmark\\
Lifting \cite{Narang2009} & Filt. $\rightarrow$ VS & VS $\rightarrow$ Filt. & V  & Bipartite & Any & B & \checkmark & \checkmark\\
Wavelets on balanced tree \cite{Gavish2010} & Filt. $\rightarrow$ VS & VS $\rightarrow$ Filt. & V & Tree & Any & B & \checkmark & \checkmark\\
SubGFB \cite{Trembl2016}&Filt. $\rightarrow$ VS & VS $\rightarrow$ Filt.&V & Any& LoS & B &\checkmark & \checkmark\\
Qualified sampling \cite{Chen2015} & Filt. $\rightarrow$ VS & Interpolation & S & Any & Any & N/A & & \checkmark\\
Uniqueness set \cite{Jin2017}&Filt. $\rightarrow$ VS & Interpolation &S &Any&Any &N/A& & \checkmark\\
Sampling set selection \cite{Anis2017} & Filt. $\rightarrow$ VS & VS $\rightarrow$ Filt. & S & Any & SNL & O, B & \checkmark & *\tnote{10}\\
\textbf{GraphSS (Proposed)}& Filt. $\rightarrow$ SS & SS $\rightarrow$ Filt. &S & Any&Any  &O, B& & \checkmark\\\hline
\end{tabular}
\begin{tablenotes}
\item[1] Building blocks for the analysis transform. Filt. $\rightarrow$ VS: Filtering then vertex domain downsampling. Filt. $\rightarrow$ SS: Filtering then spectral domain downsampling.
\item[2] Building blocks for the synthesis transform. VS $\rightarrow$ Filt.: Vertex domain upsampling then filtering. SS $\rightarrow$ Filt.: Spectral domain upsampling then filtering. Interpolation: Interpolation operator that cannot be separated into VS $\rightarrow$ Filt.
\item[3] Domain for filter design. S: Graph frequency domain. V: Vertex domain.
\item[4] Applicable graphs.
\item[5] Applicable variation operators. SNL: symmetric normalized graph Laplacian. NAD: normalized adjacency matrix. LoS: Laplacian on subgraphs.
\item[6] Orthogonality. O: Orthogonal. B: Biorthogonal.
\item[7] Compact support.
\item[8] Perfect reconstruction property.
\item[9] The constraint on the $\Omega$-structure is described in \cite{Teke2016b} that includes $M$-block cyclic graphs. Precisely, this constraint can be relaxed by using a \emph{similarity transformation} \cite{Teke2016b}. However, it changes the graph Fourier bases and requires additional computations.
\item[10] Perfect reconstruction is possible only for bipartite graphs.
\end{tablenotes}
\end{threeparttable}
\end{table*}%

\section{Related Work}
\label{sec:related}
Several CS GFBs/GWTs using vertex domain sampling have been proposed for signals on bipartite graphs. They can be used on non-bipartite graphs by dividing the original graph into several bipartite graphs and then using a ``multidimensional" decomposition. Filter design methods for this class of CS GWT include: GraphQMF \cite{Narang2012}, which utilizes quadrature mirror filters; GraphBior \cite{Narang2013} a biorthogonal and polynomial filter solution with spectral factorization; a frequency conversion method (GraphFC) \cite{Sakiya2016a} that transforms time domain filters into graph spectral filters; near-orthogonal polynomial filter design methods (Nearorth) proposed in \cite{Tay2017,Tay2017a}. Oversampled graph filter banks were introduced in \cite{Tanaka2014a,Sakiya2014a} as an extension of CS GWTs for bipartite graphs.

The above methods are for designing filters in the graph frequency domain. There are also CS graph filter banks whose filters are designed in the vertex domain. For example, a lifting-based transform \cite{Narang2009} divides the original graph into even and odd-indexed vertices and performs vertex domain filtering. The subgraph-based biorthogonal filter bank (SubGFB) \cite{Trembl2016} decomposes the original graph into several partitions. Wavelets on a balanced tree \cite{Gavish2010} provide CS perfect reconstruction transforms using vertex domain filtering. There is a CS graph filter bank for a specific class of graphs, called $\Omega$-structures \cite{Teke2016b}\footnote{For more general graphs, we need to redesign bases for a new graph Fourier transform.}. However, all of these methods require simplifying the graph, i.e., eliminating some of the edges in the original graph, in order to ensure critical sampling and invertibility.

CS graph filter banks can also be designed with careful vertex domain sampling. An $M$-channel CS graph filter bank \cite{Jin2017} was designed that selects sampled vertices for each subband in order to satisfy the \textit{uniqueness set} condition. In the context of sampling theory of graph signals, bandlimiting the input graph signal followed by vertex domain sampling has been proposed as a graph filter bank \cite{Chen2015}. However, such approaches have several limitations. First, they have to select an appropriate sampling set for perfect reconstruction. In other words, arbitrarily selected sampling sets do not generally lead to a perfect reconstruction transform. Second, the sampling set is not unique; different sampling sets significantly affect the overall performance of the graph transforms in applications. Third, many approaches are perfect reconstruction only if ideal filters are used in the analysis transform. That means there is no flexibility in the design of the filter. However, non-ideal filters are sometimes preferred when the eigenvalue distribution of the variation operator is irregular (described in Section \ref{sec:idealvsnonideal}). Fourth, they need to calculate the reconstruction operator for the synthesis side that requires a matrix inversion \cite{Jin2017, Chen2015}. While the spline-based graph wavelet \cite{Ekamba2015} is CS and guarantees perfect reconstruction, with a relatively flexible downsampling pattern, it still requires a matrix inversion for the synthesis transform. In fact, a perfect reconstruction transform with \emph{polynomial} analysis filters can only be obtained if the synthesis transform is the matrix inverse of the analysis transform matrix, with the sole exception of the bipartite case \cite{Anis2017}. This leads to two complex computations: 1) computing the matrix inverse and 2) multiplying the frequency domain representation by this (dense) matrix to obtain the reconstructed signal (instead of using polynomial filters for reconstruction).

The performance of CS GWTs with vertex domain sampling varies according to the graph reduction method used. Graph coloring \cite{Narang2012,Aspval1984,Narang2013,Harary1977}, Kron reduction \cite{Dorfle2013,Shuman2016}, maximum spanning trees \cite{Nguyen2015}, weighted max-cut \cite{Narang2010}, and graph coarsening using algebraic distance \cite{Ron2011} are examples of the various graph reduction methods.

The properties of the existing and proposed GFBs are summarized in Table \ref{tb:list_gwt}. It should be emphasized that all of the existing approaches have limitations on their design, e.g., eligible graphs/variation operators, sampling set guaranteeing perfect reconstruction, or filter design. Our approach overcomes the limitations by employing a novel sampling in the graph frequency domain, and it is the only approach that has all of the following features: i) spectral domain filtering, ii) orthogonality, iii) perfect reconstruction, and iv) applicability to any graph.

\begin{figure*}[tp]
\centering
\subfigure[][(I) Original graph signal. 
  (II) Vertex domain downsampling (GD1). 
  (III) Spectral domain downsampling (GD2).]
  {\includegraphics[width = 0.45\linewidth]{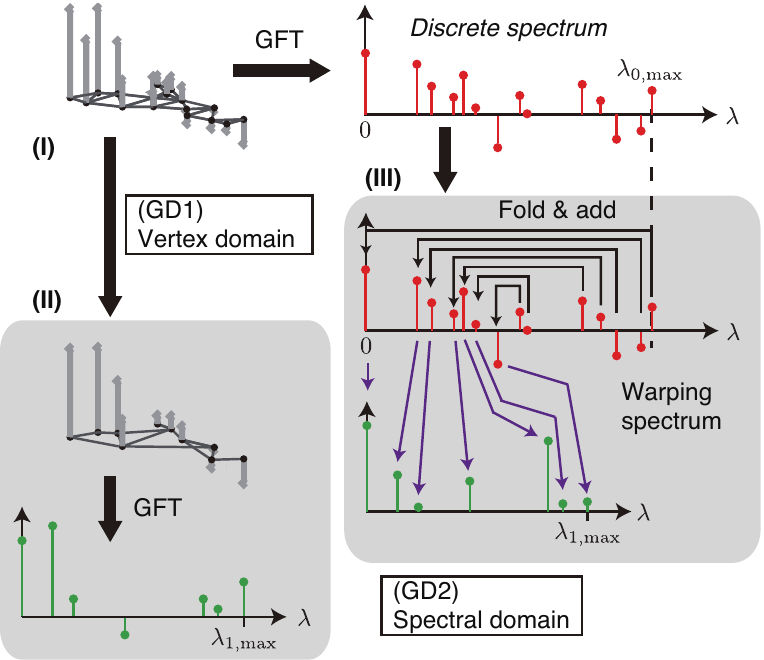}}\quad
\subfigure[][(I) Original graph signal. 
  (II) Vertex domain upsampling (GU1). 
  (III) Spectral domain upsampling (GU2).]
 {\includegraphics[width = 0.45\linewidth]{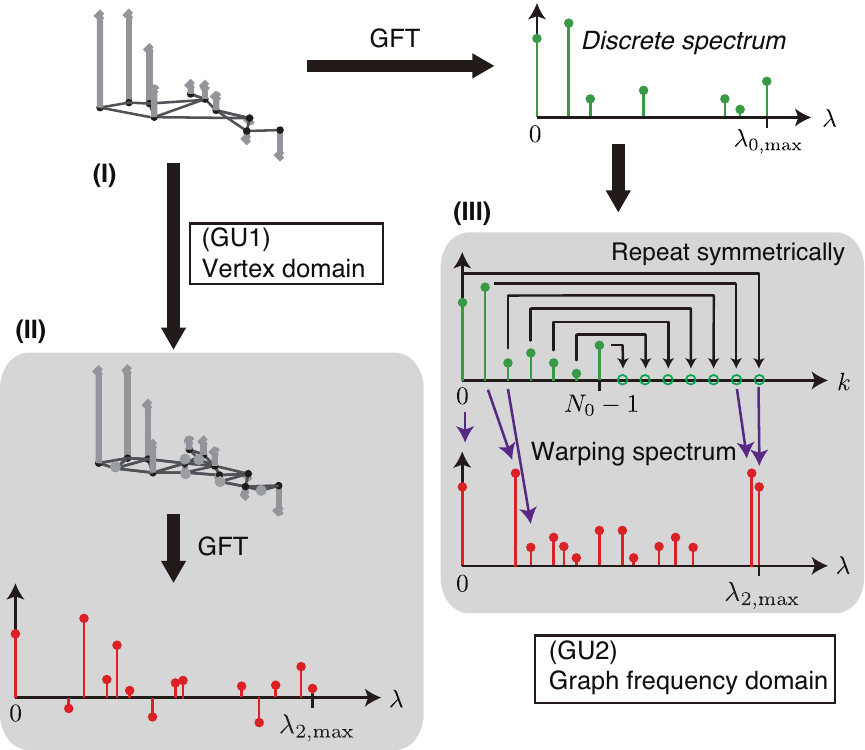}}
\caption{Sampling of graph signals. (a) Downsampling. (b) Upsampling.}
\label{fig:GD2}
\end{figure*}

\section{Sampling of Graph Signals}\label{sec:II}
This section describes the sampling methods of graph signals in the vertex and graph frequency domains.

\subsection{Sampling in Vertex Domain}
The conventional and widely used method for sampling graph signals in the vertex domain \cite{Pesens2008, Wang2015, Anis2016, Tsitsv2016} is defined as follows:
\begin{defi}\label{def:GD1}
(Downsampling of graph signals in vertex domain). Let $\mathcal{G}_0=(\mathcal{V}_0, \mathcal{E}_0)$ and $\mathcal{G}_1=(\mathcal{V}_1, \mathcal{E}_1)$ be the original and reduced-size graphs, respectively, where every vertex in $\mathcal{G}_1$ has a one-to-one correspondence to one of the vertices in $\mathcal{G}_0$. The original signal is $\mathbf{f} \in \mathbb{R}^{|\mathcal{V}_0|}$. In the vertex domain, downsampling of $\mathbf{f}$ to $\mathbf{f}_d \in \mathbb{R}^{|\mathcal{V}_1|}$ is defined as follows.
\end{defi}
\begin{itemize}
 \setlength{\leftskip}{2.7ex}
 \item[(GD1)] Keeping samples in $\mathcal{V}_1$.
 \begin{equation}
 \label{eqn:gd1}
 f_d [n] = f[n'] \quad \text{if} \; v_{0,n'} \in \mathcal{V}_0 \text{ corresponds to } v_{1,n} \in \mathcal{V}_1.
 \end{equation}
\end{itemize}
This is illustrated in Fig. \ref{fig:GD2}(II).

\begin{defi}
(Upsampling of graph signals in vertex domain). $\mathcal{G}_0$ and $\mathcal{G}_1$ are the same as in Definition \ref{def:GD1}. The original signal at this time is $\mathbf{f} \in \mathbb{R}^{|\mathcal{V}_1|}$, and its sample is associated with $\mathcal{G}_1$. Upsampling in the vertex domain, i.e., mapping from $\mathbf{f}$ to $\mathbf{f}_u \in \mathbb{R}^{|\mathcal{V}_0|}$, is defined as follows.
\end{defi}
\begin{itemize}
 \setlength{\leftskip}{2.7ex}
 \item[(GU1)] Placing samples on $\mathcal{V}_1$ into the corresponding vertices in $\mathcal{G}_0$.
 \begin{equation}\label{eqn:gu1}
 f_u [n] =
 \begin{cases}
 f[n'] & \text{if} \; v_{n'} \in \mathcal{V}_1 \text{ corresponds to } v_n \in \mathcal{V}_0 \\
 0  & \text{otherwise.}
 \end{cases}
\end{equation}
\end{itemize}
This is illustrated in Fig. \ref{fig:GD2}(II).

\subsection{Sampling in Graph Frequency Domain}\label{subsec:spectraldomainsamp}
Next, we describe sampling of graph signals defined in the graph frequency domain \cite{Tanaka2018}. Note that there are a number of slightly different definitions in the literature. Please refer to \cite{Tanaka2018} for other definitions besides the ones used here.

\begin{defi}
(Downsampling of graph signals in graph frequency domain). Let $\mathbf{L}_0 \in \mathbb{R}^{N \times N}$ and $\mathbf{L}_1 \in \mathbb{R}^{N/2 \times N/2}$ respectively be graph Laplacians for the original and reduced-size graphs, respectively, and assume that their eigendecompositions are given as $\mathbf{L}_0 = \mathbf{U}_0 \mathbf{\Lambda}_0 \mathbf{U}^{\top}_{0}$ and $\mathbf{L}_1 = \mathbf{U}_1 \mathbf{\Lambda} _1 \mathbf{U}^{\top}_{1}$, where $\mathbf{\Lambda}_\ell = \text{\emph{diag}}(\lambda_{\ell,0}, \lambda_{\ell,1}, \ldots , \lambda_{\ell, \max})$. The downsampled graph signal in the graph frequency domain $\widetilde{\mathbf{f}}_d \in \mathbb{R}^{N/2}$ is defined as follows.
\end{defi}
\begin{itemize}
 \setlength{\leftskip}{2.7ex}
 \item[(GD2)] $\widetilde{\mathbf{f}}$, i.e., the signal in the graph frequency domain, is evenly divided by $2$. Then the second portion is flipped and summed with the first one.
 \begin{equation}
\widetilde{f}_d[i] = \widetilde{f}[i]+\widetilde{f}[N-i-1],
 \end{equation}
where $i=0, \ldots , N/2-1$. The above equation is easily represented in matrix form:
 \begin{equation}
 \mathbf{f}_d = \mathbf{U}_1 \widetilde{\mathbf{S}}_d \mathbf{U}^{\top}_{0} \mathbf{f},
 \end{equation}
where $\widetilde{\mathbf{S}}_d = \left[ \mathbf{I}_{N/2}\ \mathbf{J}_{N/2} \right]$, in which $\mathbf{I}$ and $\mathbf{J}$ are the identity and counter-identity matrices, respectively.
This downsampling strategy is illustrated in Fig. \ref{fig:GD2}(III).
\end{itemize}

\begin{defi}
(Upsampling of graph signals in graph frequency domain). Let $\mathbf{L}_{0} \in \mathbb{R}^{N \times N}$ and $\mathbf{L}_2 \in \mathbb{R}^{2N \times 2N}$ be the graph Laplacians for the original and increased-size graphs, respectively. The upsampled graph signal in the graph frequency domain $\widetilde{\mathbf{f}}_u \in \mathbb{R}^{2N}$ is defined as follows.
\end{defi}
\begin{itemize}
 \setlength{\leftskip}{2.7ex}
 \item[(GU2)] Repeating the original and flipped spectra alternatively.
 \begin{equation}
 \widetilde{f}_u [i] =
 \begin{cases}
 \widetilde{f}[i] & i= 0,\ldots , N-1 \\
 \widetilde{f}[2N-i-1] & i=N, \ldots , 2N-1.
 \end{cases}
 \end{equation}
The above equation is easily represented in matrix form:
 \begin{equation}
 \mathbf{f}_u = \mathbf{U}_2 \widetilde{\mathbf{S}}_u \mathbf{U}^{\top}_{0} \mathbf{f},
 \end{equation}
where $\widetilde{\mathbf{S}}_u = \left[ \mathbf{I}_{N}\ \mathbf{J}_{N} \right]^{\top}$ and $\mathbf{U}_2$ is the eigenvector matrix of $\mathbf{L}_2$.
 \end{itemize}
This upsampling strategy is illustrated in Fig. \ref{fig:GD2}(III).

It is worth noting that, when sampling a signal in the graph frequency domain, in general we do not have a simple expression for the corresponding signal in the vertex domain, with the exception of bipartite graphs (see Section \ref{sec:V}).

\begin{figure}[tp]
\centering
\subfigure[][CS GFB with vertex domain sampling.]
  {\includegraphics[width = \linewidth]{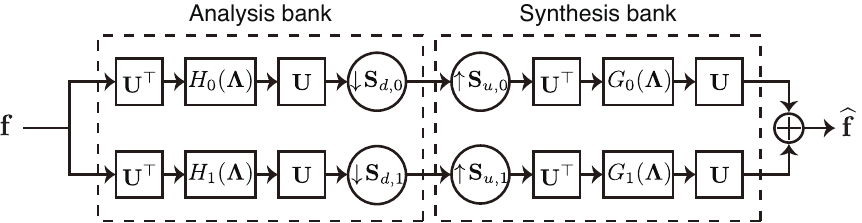}}\\
\subfigure[][CS GFB with spectral domain sampling.]
 {\includegraphics[width = \linewidth]{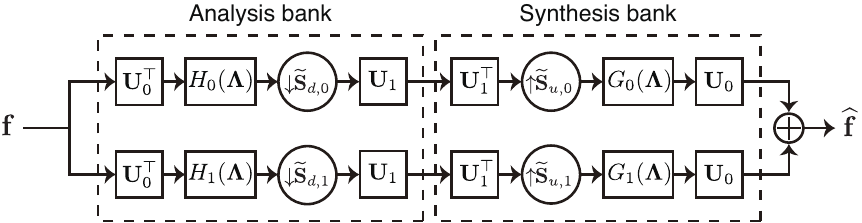}}
\caption{Two-channel CS GFBs.}
\label{fig:CSSGWT}
\end{figure}

\section{Two-Channel CS GFBs with Vertex Domain Sampling}\label{sec:III}

\subsection{Framework and Perfect Reconstruction Condition}
The most popular CS GFBs are designed for bipartite graphs \cite{Narang2012,Narang2013}. They are perfect reconstruction if the underlying graph is bipartite and the variation operator is a symmetric normalized graph Laplacian or normalized random walk graph Laplacian. Non-bipartite graphs should be simplified to bipartite ones before transformation by the CS GFBs to guarantee the perfect reconstruction condition.

Fig. \ref{fig:CSSGWT}(a) illustrates the entire transformation for one bipartite graph. In the figure, $\mathbf{H}_k := \mathbf{U}H_k(\mathbf{\Lambda})\mathbf{U}^{\top}$ is the $k$th filter in the analysis filter bank and $\mathbf{G}_k := \mathbf{U}G_k(\mathbf{\Lambda})\mathbf{U}^{\top}$ is the $k$th filter in the synthesis filter bank, in which
\begin{equation}\label{eqn:graphspectralfilters}
\begin{split}
H_k(\mathbf{\Lambda}) &= \text{diag} (H_k(\lambda_0), H_k(\lambda_1) , \ldots , H_k(\lambda_{N-1})) \\
G_k(\mathbf{\Lambda}) &= \text{diag} (G_k(\lambda_0), G_k(\lambda_1) , \ldots , G_k(\lambda_{N-1})).
\end{split}
\end{equation} 

Let $\mathcal{B}=( \mathcal{L}, \mathcal{H},\mathcal{E})$ be a bipartite graph only having edges between vertex sets $\mathcal{L}$ and $\mathcal{H}$. The number of samples in each channel is determined on the basis of the graph-coloring result. Down- and upsampling in the vertex domain for $\mathcal{B}$, represented in \eqref{eqn:gd1} and \eqref{eqn:gu1}, is defined in matrix notation as follows:
\begin{equation}\label{eq:samp2}
\begin{split}
&\mathbf{S}_{d,0} =\mathbf{I}_{\mathcal{L}}\in \{ 0,1 \}^{|\mathcal{L}| \times N}, \quad 
\mathbf{S}_{u,0} = \mathbf{S}^{\top}_{d,0} \\
&\mathbf{S}_{d,1} = \mathbf{I}_{\mathcal{H}} \in \{ 0,1 \}^{|\mathcal{H}| \times N}, \quad 
\mathbf{S}_{u,1} = \mathbf{S}^{\top}_{d,1},
\end{split}
\end{equation}
where $\mathbf{I}_{\mathcal{L}}$ and $\mathbf{I}_{\mathcal{H}}$ are submatrices of $\mathbf{I}_N$ whose rows correspond to the indices of $\mathcal{L}$ and $\mathcal{H}$, respectively. 
That is, the sampled signal can be represented as $\mathbf{f}_{d,0} = \mathbf{S}_{d,0}\mathbf{f}$ and so on.
The two-channel GFB shown in Fig.  \ref{fig:CSSGWT}(a) is designed to satisfy the following perfect reconstruction condition.
\begin{equation}\label{eq:t}
\mathbf{T}_v= \mathbf{G}_0 \mathbf{S}_{u,0}\mathbf{S}_{d,0} \mathbf{H}_0 
+ \mathbf{G}_1 \mathbf{S}_{u,1}\mathbf{S}_{d,1} \mathbf{H}_1 
= c^{2}\mathbf{I}_N,
\end{equation}
where $c \in \mathbb{R}$. \eqref{eq:t} is further represented as the condition for spectral graph filters as follows.
\begin{align}
G_0(\lambda)H_0(\lambda) + G_1(\lambda)H_1(\lambda) &= c^{2}\label{eq:pr2-1}\\
G_0(\lambda)H_0(2-\lambda) - G_1(\lambda)H_1(2-\lambda) &= 0.\label{eq:pr2-2}
\end{align}

\subsection{Filter Design}
Several CS GFBs that satisfy the above perfect reconstruction condition have been proposed, together with some filter designs \cite{Sakiya2016a,Narang2012,Narang2013,Sakiya2014a,Tanaka2014a}.

\subsubsection{GraphQMF \cite{Narang2012}}
GraphQMF is an orthogonal solution designed from one spectral kernel $H_0(\lambda)$. The remaining filters are defined as follows.
\begin{equation}
\begin{split}
&H_1(\lambda) = H_0(2-\lambda) \\
&G_0(\lambda) = H_0(\lambda) \\
&G_1(\lambda) = H_1(\lambda) = H_0(2-\lambda).
\end{split}
\end{equation}
$H_{0}(\lambda)$ has to satisfy the following condition to ensure perfect reconstruction:
\begin{equation}
H^{2}_{0}(\lambda) + H^{2}_{0}(2-\lambda) = c^2.
\end{equation}
Fig. \ref{fig:CSSGWTs_existing}(a) shows its graph spectral characteristics with the Meyer wavelet kernel \cite{Narang2012}.

\subsubsection{GraphBior \cite{Narang2013}}
GraphBior, which is a biorthogonal CS GWT, is designed to satisfy
\begin{equation}
H_1(\lambda) = G_0 (2-\lambda), \; G_1(\lambda) = H_0(2-\lambda).
\end{equation}
This leads to
\begin{equation}
H_0(\lambda)G_0(\lambda) + H_0(2-\lambda)G_0(2-\lambda) = 2.
\end{equation}
A low-pass half-band product filter $P(\lambda) = H_0(\lambda)G_0(\lambda)$ is designed first; then $H_0(\lambda)$ and $G_0(\lambda)$ are obtained via spectral factorization similar to the Cohen-Daubechies-Feauveau (CDF) biorthogonal wavelet transform in classical signal processing \cite{Cohen1992}. The analysis filter characteristics are shown in Fig. \ref{fig:CSSGWTs_existing}(b). 

\subsubsection{GraphFC \cite{Sakiya2016a}}
A method has been proposed for converting time domain filters $H(\omega)$ into graph spectral filters $H(\la)$ through a frequency mapping from $\omega \in [ 0, \pi ]$ to $\lambda \in [0, \lambda_{\max}]$ \cite{Sakiya2016a}. In this approach, the perfect reconstruction condition \eqref{eq:pr2-1} and \eqref{eq:pr2-2} is always satisfied as long as the set of time domain filters are perfect reconstruction (in the time domain). The analysis filter characteristics based on the CDF 9/7 filters are shown in Fig. \ref{fig:CSSGWTs_existing}(c).


\begin{figure}[tp]
\centering
\subfigure[][GraphQMF]{\includegraphics[width = 0.32\linewidth]{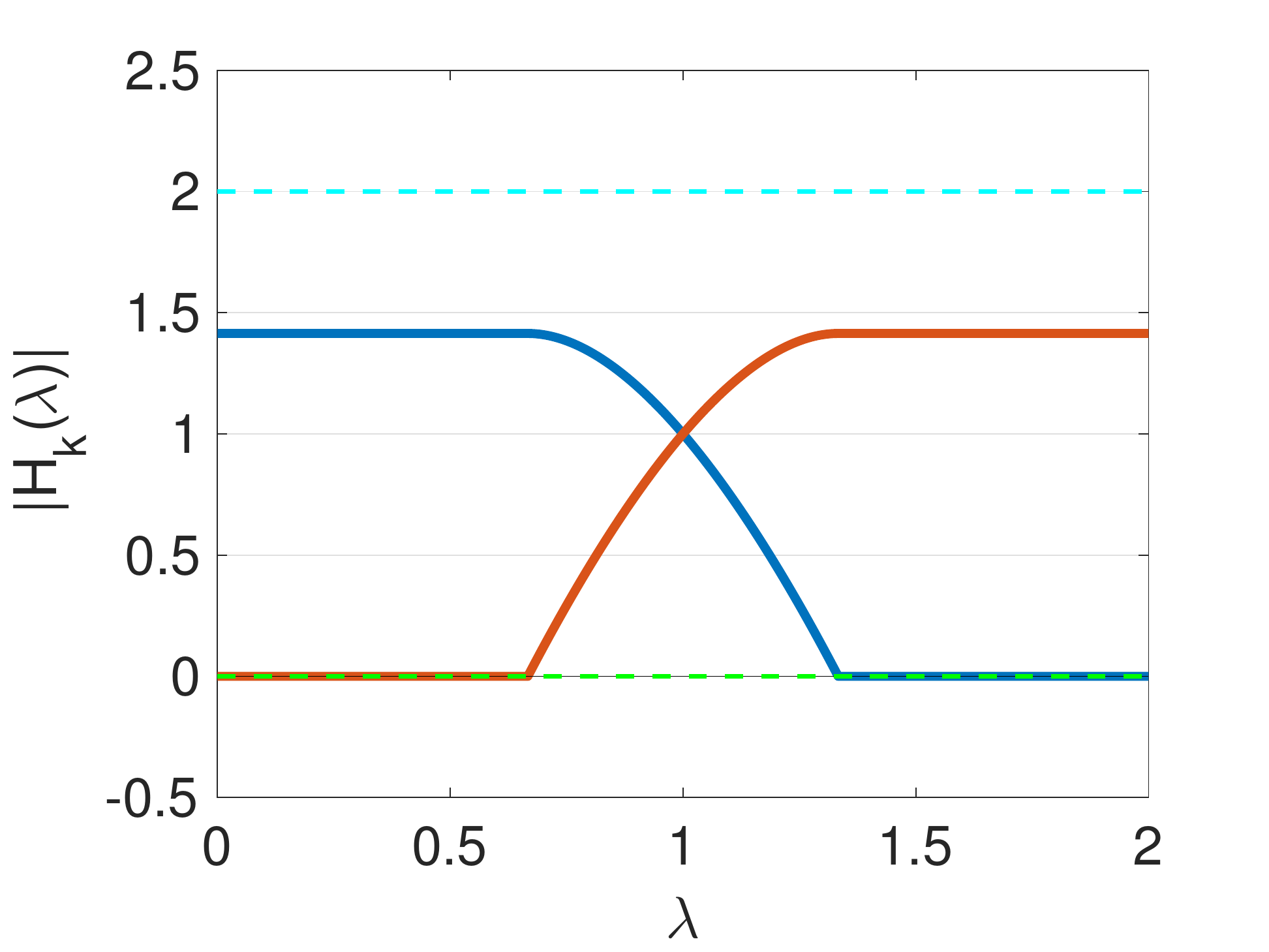}}\ 
\subfigure[][GraphBior]{\includegraphics[width = 0.32\linewidth]{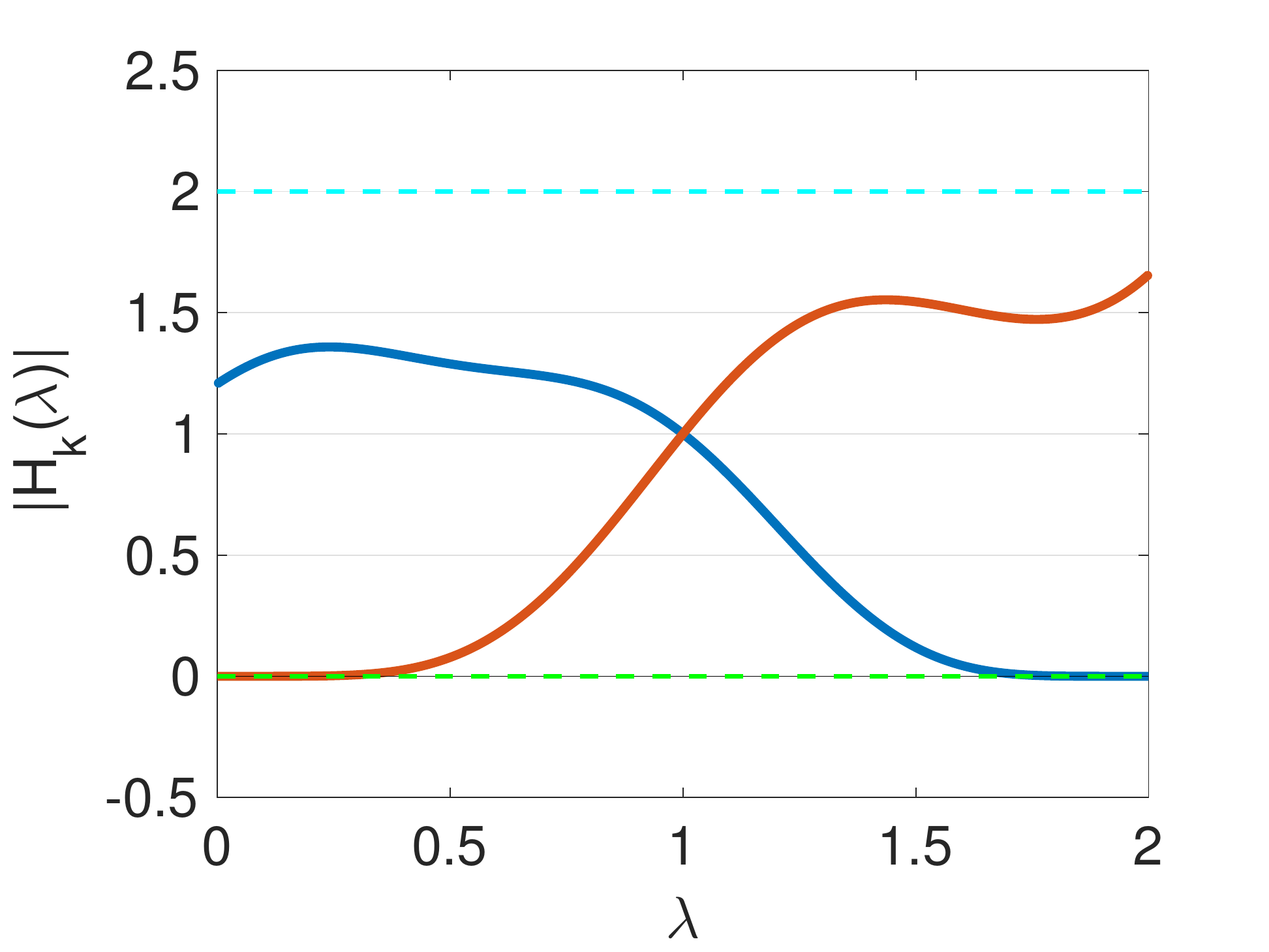}}\ 
\subfigure[][GraphFC]{\includegraphics[width = 0.32\linewidth]{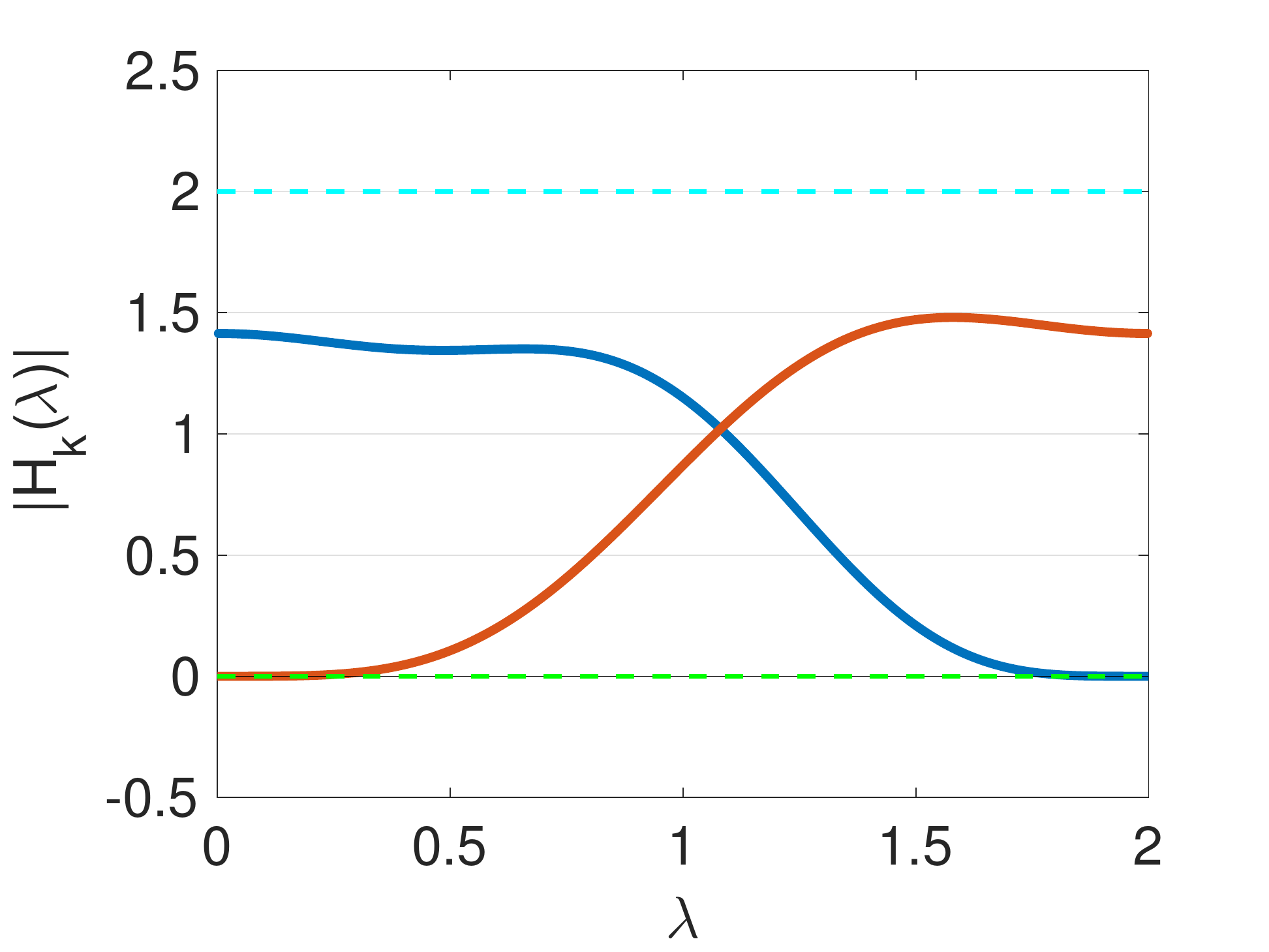}}
\caption{Existing CS GWTs. Light blue and yellow-green dashed lines represent \eqref{eq:pr2-1} and \eqref{eq:pr2-2}, respectively.}
\label{fig:CSSGWTs_existing}
\end{figure}


\section{Two-Channel CS GFBs with Spectral Domain Sampling}\label{sec:IV}
Here, we describe the framework and perfect reconstruction condition of the proposed CS GFBs using spectral domain sampling. We also present their octave-band decomposition and polyphase representation, along with filter design methods.

\subsection{Framework and Perfect Reconstruction Condition}
The framework of the proposed transform is shown in Fig. \ref{fig:CSSGWT}(b), where $H_k(\bm{\Lambda})$ and $G_k(\bm{\Lambda})$ are the same as \eqref{eqn:graphspectralfilters}. It seems to be similar to the existing CS GFBs, but it differs in that all operations are performed in the graph frequency domain.

The sampling matrices are defined as follows:
\begin{equation}\label{eq:samp}
\begin{split}
&\widetilde{\mathbf{S}}_{d,0} = \left[ \mathbf{I}_{N/2} \;\; \mathbf{J}_{N/2} \right], \;\;\; 
\widetilde{\mathbf{S}}_{u,0} = \widetilde{\mathbf{S}}^{\top}_{d,0} \\
&\widetilde{\mathbf{S}}_{d,1} = \left[ \mathbf{I}_{N/2} \;\; -\mathbf{J}_{N/2} \right], \;\;\; 
\widetilde{\mathbf{S}}_{u,1} = \widetilde{\mathbf{S}}^{\top}_{d,1}.
\end{split}
\end{equation}
The down- and upsampling operations for the low-pass branch are, respectively, the same as (GD2) and (GU2) introduced in Section \ref{subsec:spectraldomainsamp}, whereas those in the highpass branch are modulated versions of (GD2) and (GU2).

By using this structure, the following theorem gives the perfect reconstruction condition.

\begin{theo}\label{thm:prcond}
For any variation operator having orthogonal GFT matrices\footnote{Its extension to the invertible or unitary case is trivial.}, the two-channel CS GFB with spectral domain sampling shown in Fig. \ref{fig:CSSGWT}(b) is a perfect reconstruction transform if the graph spectral responses of the filters satisfy the following relationship for all $i$:
\begin{align}
G_0(\lambda_i)H_0(\lambda_i) + G_1(\lambda_i)H_1(\lambda_i) &= c^{2}\label{eq:pr1-1}\\
G_0(\lambda_i)H_0(\lambda_{N-i-1}) - G_1(\lambda_i)H_1(\lambda_{N-i-1}) &= 0.\label{eq:pr1-2}
\end{align}
\end{theo}
\begin{proof}
The output signal $\widehat{\mathbf{f}}$ in Fig.  \ref{fig:CSSGWT}(b) is represented as
\begin{equation}
\begin{split}
\widehat{\mathbf{f}} = &\mathbf{U}_0 G_0(\mathbf{\Lambda}) \widetilde{\mathbf{S}}_{u,0}\widetilde{\mathbf{S}}_{d,0} H_0(\mathbf{\Lambda})\mathbf{U}^{\top}_0 \mathbf{f} \\
& + \mathbf{U}_0 G_1(\mathbf{\Lambda}) \widetilde{\mathbf{S}}_{u,1}\widetilde{\mathbf{S}}_{d,1} H_1(\mathbf{\Lambda})\mathbf{U}^{\top}_0 \mathbf{f}.
\end{split}
\end{equation}
Since $\mathbf{U}_0$ is an orthogonal matrix, if the transfer matrix
\begin{equation}\label{eq:ts}
\widetilde{\mathbf{T}}_s = G_0(\mathbf{\Lambda}) \widetilde{\mathbf{S}}_{u,0}\widetilde{\mathbf{S}}_{d,0} H_0(\mathbf{\Lambda}) + G_1(\mathbf{\Lambda}) \widetilde{\mathbf{S}}_{u,1}\widetilde{\mathbf{S}}_{d,1} H_1(\mathbf{\Lambda})
\end{equation}
is the identity matrix, the whole transform becomes a perfect reconstruction system.
By substituting \eqref{eq:samp} into \eqref{eq:ts}, $\widetilde{\mathbf{T}}_s$ becomes
\begin{equation}
\begin{split}
\widetilde{\mathbf{T}}_s &= G_0(\mathbf{\Lambda})
\begin{bmatrix}
\mathbf{I}_{N/2}\\
\mathbf{J}_{N/2}
\end{bmatrix}
\begin{bmatrix}
\mathbf{I}_{N/2}&\mathbf{J}_{N/2}
\end{bmatrix}
H_0(\mathbf{\Lambda}) \\
&\quad+
G_1(\mathbf{\Lambda})
\begin{bmatrix}
\mathbf{I}_{N/2}\\
-\mathbf{J}_{N/2}
\end{bmatrix}
\begin{bmatrix}
\mathbf{I}_{N/2}&-\mathbf{J}_{N/2}
\end{bmatrix}
H_1(\mathbf{\Lambda}) \\
&=
G_0(\mathbf{\Lambda}) \left( \mathbf{I}_N + \mathbf{J}_N \right)H_0(\mathbf{\Lambda})+G_1(\mathbf{\Lambda}) \left( \mathbf{I}_N - \mathbf{J}_N \right)H_1(\mathbf{\Lambda})\\
&=
\left( G_0(\mathbf{\Lambda})H_0(\mathbf{\Lambda}) + G_1(\mathbf{\Lambda}) H_1(\mathbf{\Lambda}) \right)\\
&\quad+\left(G_0(\mathbf{\Lambda}){H}_0(\mathbf{\Lambda}') - G_1(\mathbf{\Lambda}) {H}_1(\mathbf{\Lambda}') \right)\mathbf{J}_N,
\end{split}
\end{equation}
where $\mathbf{\Lambda}' = \text{diag}(\lambda_{N-1} , \ldots , \lambda_0)$.
If the filters satisfy \eqref{eq:pr1-1} and \eqref{eq:pr1-2}, we have
\begin{align}
G_0(\mathbf{\Lambda})H_0(\mathbf{\Lambda}) + G_1(\mathbf{\Lambda}) H_1(\mathbf{\Lambda}) &= c^{2}\mathbf{I}_N\\
G_0(\mathbf{\Lambda}){H}_0(\mathbf{\Lambda}') - G_1(\mathbf{\Lambda}) {H}_1(\mathbf{\Lambda}') &=\mathbf{0}_N.
\end{align}
This leads to $\widetilde{\mathbf{T}}_s = c^{2} \mathbf{I}_N$.
\end{proof}
In \eqref{eq:pr1-1} and \eqref{eq:pr1-2}, we have $2N$ constraints for a perfect reconstruction transform. They can be satisfied when we use filter coefficients from two-channel classical wavelet transforms. We describe the filter design method in Section \ref{subsec:designmethods}.

\begin{figure*}[tp]
\centering
\subfigure[][]{\includegraphics[scale=0.65]{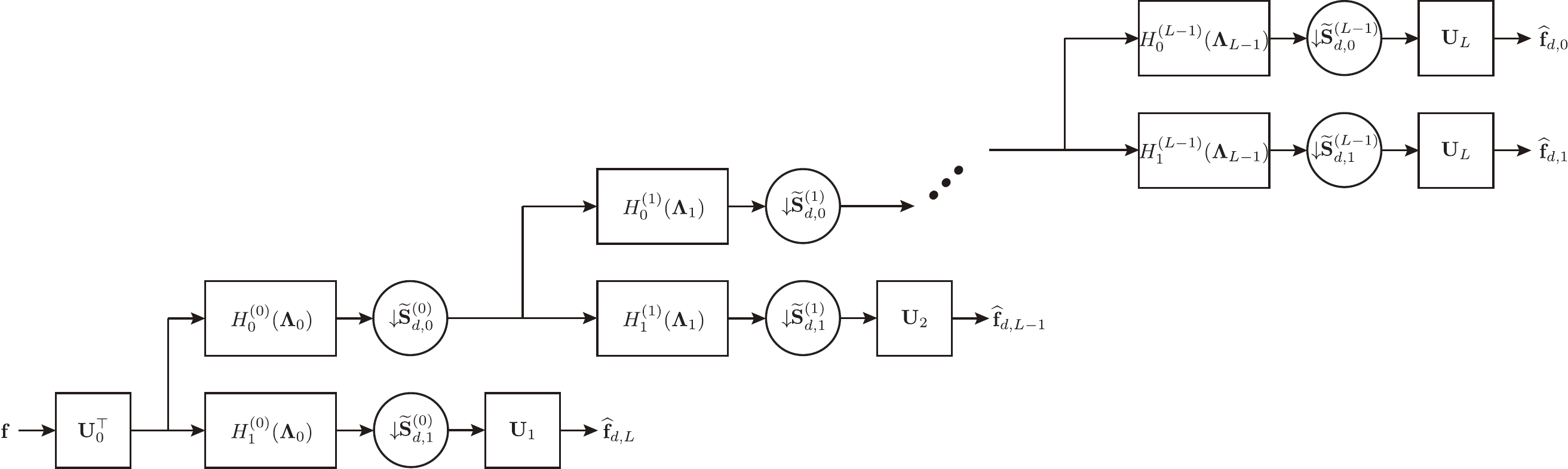}}\\
\subfigure[][]{\includegraphics[scale=0.65]{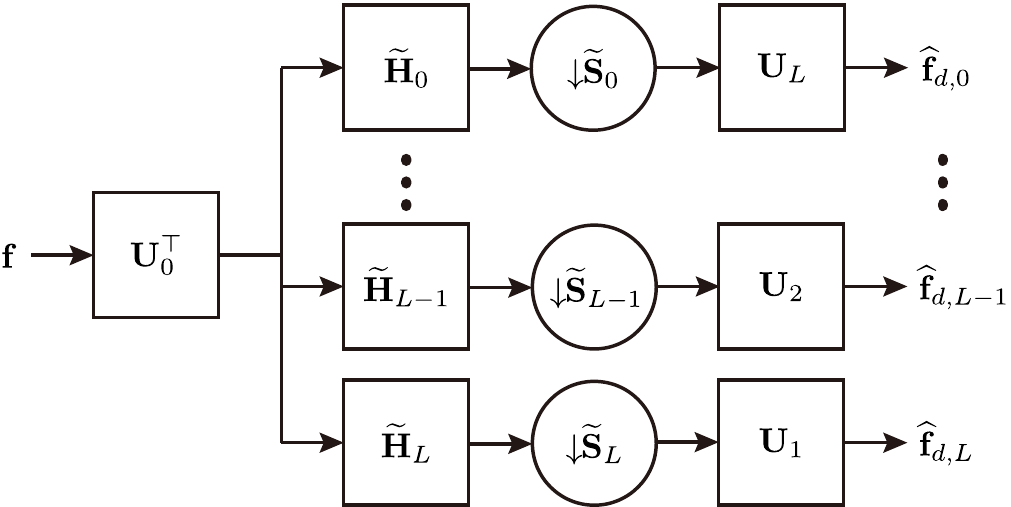}}
\caption{$L$-level octave-band analysis filter bank with spectral domain sampling. (a) The direct structure. For the output in the low-pass channel, i.e., the input for the transform at the next level, we do not need to perform the inverse GFT because the forward GFT in the next level cancels it.} (b) The equivalent structure after merging the filters and sampling operators. $\widetilde{\mathbf{H}}_k$ and $\widetilde{\mathbf{S}}_k$ are the $k$th merged filters and sampling matrices, respectively.
\label{fig:L-level}
\end{figure*}


\subsection{Octave-Band Decomposition}
The above-mentioned transform is one-level, but we can cascade it to realize an octave-band decomposition. Here, we consider the $L$-level octave-band analysis transform shown in Fig. \ref{fig:L-level} and merge the filters and sampling operations. Although there are many methods to obtain multiscale graphs as described in Section \ref{sec:related} (desired properties for graph reduction are summarized in \cite{Shuman2016}), it is worth noting that we can design a perfect reconstruction with any orthogonal $\mathbf{U}_i$'s according to Theorem \ref{thm:prcond}.

For simplicity, only the lowest frequency band on the analysis side is considered. Our discussion is easily generalized to the other bands and the synthesis transform.

Let $\widetilde{\mathbf{S}}^{(k)}_{d, 0}$ be the $k$th level downsampling matrix for the spectral domain sampling. After the $L$-level transform, the transform matrix in the lowest band can be written as follows.
\begin{equation}
\begin{split}
\mathbf{T}_{a}^{(L)} =\ & \mathbf{U}_{L} \left(\widetilde{\mathbf{S}}^{(L-1)}_{d, 0} H^{(L-1)}_{0}(\mathbf{\Lambda}_{L-1})\right)\\
&\times \left(\widetilde{\mathbf{S}}^{(L-2)}_{d, 0} H^{(L-2)}_{0}(\mathbf{\Lambda}_{L-2})\right) \cdots\left(\widetilde{\mathbf{S}}^{(0)}_{d, 0} H^{(0)}_{0}(\mathbf{\Lambda}_{0})\right)
\mathbf{U}^\top_0.
\end{split}
\end{equation}
Furthermore, for all $k$, $H^{(k)}_{0}(\mathbf{\Lambda}_{k})\widetilde{\mathbf{S}}_{d,0}^{(k-1)}H^{(k-1)}_{0}(\mathbf{\Lambda}_{k-1})$ can be rewritten as
\begin{equation}
\label{eq:SH}
\begin{split}
&H^{(k)}_{0}(\mathbf{\Lambda}_{k})\widetilde{\mathbf{S}}_{d,0}^{(k-1)}H^{(k-1)}_{0}(\mathbf{\Lambda}_{k-1})\\
& = \widetilde{\mathbf{S}}_{d,0}^{(k-1)}
\begin{bmatrix}
H^{(k)}_{0}(\mathbf{\Lambda}_{k}) & \\
 & \mathbf{J}_{\frac{N}{2^{k}}}H^{(k)}_{0}(\mathbf{\Lambda}_{k})\mathbf{J}_{\frac{N}{2^{k}}}
\end{bmatrix}H^{(k-1)}_{0}(\mathbf{\Lambda}_{k-1}),
\end{split}
\end{equation}
and the product of the second and last terms in the above equation is still diagonal. 

Consequently, the spectral domain sampling can be merged in the last part of the analysis transform. This means the filtering and sampling has to be performed only once even when using our framework for octave-band decomposition as shown in Fig. \ref{fig:L-level}(b).

\subsection{Polyphase Representation}
\subsubsection{General Form}
A polyphase representation is an efficient implementation of filter banks for classical signal processing \cite{Vaidya1993, Strang1996, Bellan1976}, and it has also been studied in graph signal processing \cite{Tay2017b, Tanaka2017a}. It is efficient because it allows us to move the downsampling and upsampling operations, respectively, to the very first and last parts of the transform via the \emph{noble identity} \cite{Vaidya1993, Strang1996}. Here, we show that such a structure is possible for our framework and describe a special case for bipartite graphs.

Let us define $\mathbf{\Lambda}_u = \text{diag}(\lambda_0,\ldots, \lambda_{\frac{N}{2}-1})$ and $\mathbf{\Lambda}'_l = \text{diag}(\lambda_{N-1}, \ldots, \lambda_{\frac{N}{2}})$. Note that the order of the eigenvalues in $\mathbf{\Lambda}'_l$ is flipped. The matrix of the analysis transform $\mathbf{H} \in \mathbb{R}^{N \times N}$ is represented as
\begin{equation}
\label{eqn:H}
\mathbf{H} = 
\begin{bmatrix}
\widetilde{\mathbf{S}}_{d,0} & \\
 & \widetilde{\mathbf{S}}_{d,1}
\end{bmatrix}
\begin{bmatrix}
H_0(\mathbf{\Lambda}) &\\
& H_1(\mathbf{\Lambda})
\end{bmatrix}
\begin{bmatrix}
\mathbf{U}_0^\top\\
\mathbf{U}_0^\top
\end{bmatrix}.
\end{equation}
By looking at \eqref{eq:SH} and doing some elementary calculations, \eqref{eqn:H} can be rewritten as follows.
\begin{equation}
\label{ }
\mathbf{H} = \begin{bmatrix}
H_{0}(\mathbf{\Lambda}_u) & H_{0}(\mathbf{\Lambda}'_l)\\
H_{1}(\mathbf{\Lambda}_u) & -H_{1}(\mathbf{\Lambda}'_l)
\end{bmatrix}
\begin{bmatrix}
\mathbf{I}_{N/2} & \\
 & \mathbf{J}_{N/2}
\end{bmatrix}
\mathbf{U}_0^\top.
\end{equation}
Hence, the polyphase matrix in the graph frequency domain $\mathbf{H}_{\text{poly}}$ is
\begin{equation}
\label{ }
\mathbf{H}_{\text{poly}} = \begin{bmatrix}
H_{0}(\mathbf{\Lambda}_u) & H_{0}(\mathbf{\Lambda}'_l)\\
H_{1}(\mathbf{\Lambda}_u) & -H_{1}(\mathbf{\Lambda}'_l)
\end{bmatrix}.
\end{equation}

Similarly, the synthesis transform matrix $\mathbf{G} \in \mathbb{R}^{N\times N}$ is represented as
\begin{equation}
\label{ }
\mathbf{G} = \begin{bmatrix}
\mathbf{U}_0 & \mathbf{U}_0
\end{bmatrix}
\begin{bmatrix}
G_0(\mathbf{\Lambda}) &\\
& G_1(\mathbf{\Lambda})
\end{bmatrix}
\begin{bmatrix}
\widetilde{\mathbf{S}}_{u,0} & \\
 & \widetilde{\mathbf{S}}_{u,1}
\end{bmatrix}.
\end{equation}
After a calculation similar to that for the analysis side, we obtain the following equivalent expression,
\begin{equation}
\label{ }
\mathbf{G} = \mathbf{U}_0 \text{diag}(\mathbf{I}_{N/2}, \mathbf{J}_{N/2})\mathbf{G}_{\text{poly}},
\end{equation}
where
\begin{equation}
\label{ }
\mathbf{G}_{\text{poly}} = \begin{bmatrix}
G_{0}(\mathbf{\Lambda}_u) & G_{1}(\mathbf{\Lambda}_u)\\
G_{0}(\mathbf{\Lambda}'_l) & -G_{1}(\mathbf{\Lambda}'_l)
\end{bmatrix}.
\end{equation}
The polyphase representation in the graph frequency domain is illustrated in Fig. \ref{fig:SH}.

As in the traditional polyphase structure, it is clear that the graph signal is perfectly recovered when the product of the analysis and synthesis polyphase matrices becomes the identity matrix, i.e., $\mathbf{G}_{\text{poly}}\mathbf{H}_{\text{poly}} = \mathbf{I}_N$.

\subsubsection{Bipartite Case}
Here, suppose the underlying graph is a bipartite one with $| \mathcal{L}| =|\mathcal{H}|$ and the variation operator is a symmetric normalized graph Laplacian. Without loss of generality, we can assume that the vertices of the first half correspond to $\mathcal{L}$ and those of the second half correspond to $\mathcal{H}$. In this case, the normalized graph Laplacian can be written as
\begin{equation}
\begin{split}
{\bm{\mathcal{ L}}}&=
\begin{bmatrix}
{\bm{\mathcal{ L}}}_{\mathcal{LL}} & {\bm{\mathcal{ L}}}_{\mathcal{L}\mathcal{H}}\\
{\bm{\mathcal{ L}}}_{\mathcal{H}\mathcal{L}} & {\bm{\mathcal{ L}}}_{\mathcal{H}\mathcal{H}}
\end{bmatrix}\\
&=\begin{bmatrix}
\mathbf{U}_{\mathcal{LL}} &\mathbf{U}_{\mathcal{LL}} \\
\mathbf{U}_{\mathcal{HL}} & -\mathbf{U}_{\mathcal{HL}}
\end{bmatrix}
\begin{bmatrix}
\mathbf{\Lambda}_\mathcal{L}& \mathbf{0}\\
\mathbf{0} & 2\mathbf{I}-\mathbf{\Lambda}_\mathcal{L}
\end{bmatrix}
\begin{bmatrix}
\mathbf{U}_{\mathcal{LL}}^\top &\mathbf{U}_{\mathcal{HL}}^\top \\
\mathbf{U}_{\mathcal{LL}}^\top & -\mathbf{U}_{\mathcal{H}\mathcal{L}}^\top
\end{bmatrix}\\
&=\begin{bmatrix}
\mathbf{I} & 2\mathbf{U}_{\mathcal{LL}}(\mathbf{\Lambda}_\mathcal{L}- \mathbf{I})\mathbf{U}_{\mathcal{HL}}^\top\\
 2\mathbf{U}_{\mathcal{HL}}(\mathbf{\Lambda}_\mathcal{L} - \mathbf{I})\mathbf{U}_{\mathcal{LL}}^\top&\mathbf{I}
\end{bmatrix}.
\end{split}\label{eqn:normalizedGL}
\end{equation}

We can rewrite the signal before it is transformed using $\mathbf{H}_{\text{poly}}$ as follows:
\begin{equation}
\label{ }
\begin{split}
\begin{bmatrix}
\mathbf{I}_{N/2} & \\
 & \mathbf{J}_{N/2}
\end{bmatrix} & 
\begin{bmatrix}
\mathbf{U}_{\mathcal{LL}}^\top &\mathbf{U}_{\mathcal{HL}}^\top \\
\mathbf{U}_{\mathcal{LL}}^\top & -\mathbf{U}_{\mathcal{H}\mathcal{L}}^\top
\end{bmatrix}
\begin{bmatrix}
\mathbf{f}_u\\
\mathbf{f}_l
\end{bmatrix}\\
& = \begin{bmatrix}
\mathbf{U}_{\mathcal{LL}}^\top \mathbf{f}_u + \mathbf{U}_{\mathcal{HL}}^\top \mathbf{f}_l\\
\mathbf{J} (\mathbf{U}_{\mathcal{LL}}^\top \mathbf{f}_u - \mathbf{U}_{\mathcal{HL}}^\top \mathbf{f}_l)
\end{bmatrix},
\end{split}
\end{equation}
where $\mathbf{f}_u$ and $\mathbf{f}_l$ are the upper and lower halves of $\mathbf{f}$, respectively. Note that we can compute $\mathbf{U}_{\mathcal{LL}}^\top \mathbf{f}_u$ and $\mathbf{U}_{\mathcal{HL}}^\top \mathbf{f}_l$ separately; the downsampling operation can be performed before the graph Fourier transform as in the classical case.

\begin{figure}[tp]
\begin{center}
   \includegraphics[width = 1.0\linewidth]{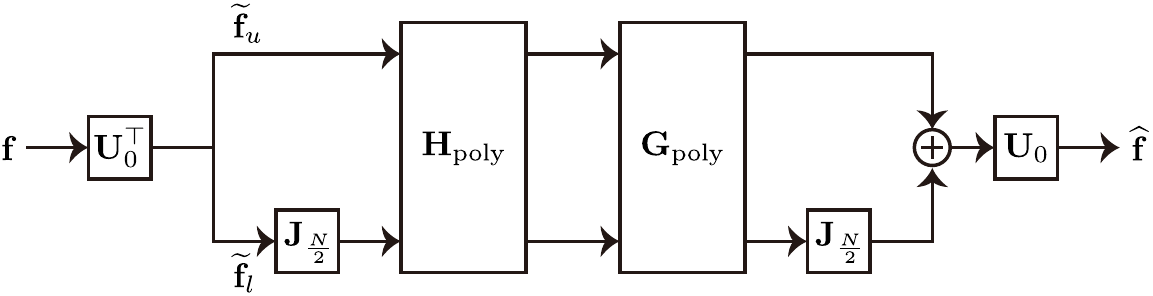}
  \caption{Polyphase representation of analysis filter bank with spectral domain sampling, where $\widetilde{\mathbf{f}}_u$ and $\widetilde{\mathbf{f}}_l$ are the upper and lower halves of $\widetilde{\mathbf{f}}$, respectively.}
  \label{fig:SH}
  \end{center}
\end{figure}

\begin{figure}[t]
\centering
\subfigure[][GraphSS-I]{\includegraphics[width = 0.32\linewidth]{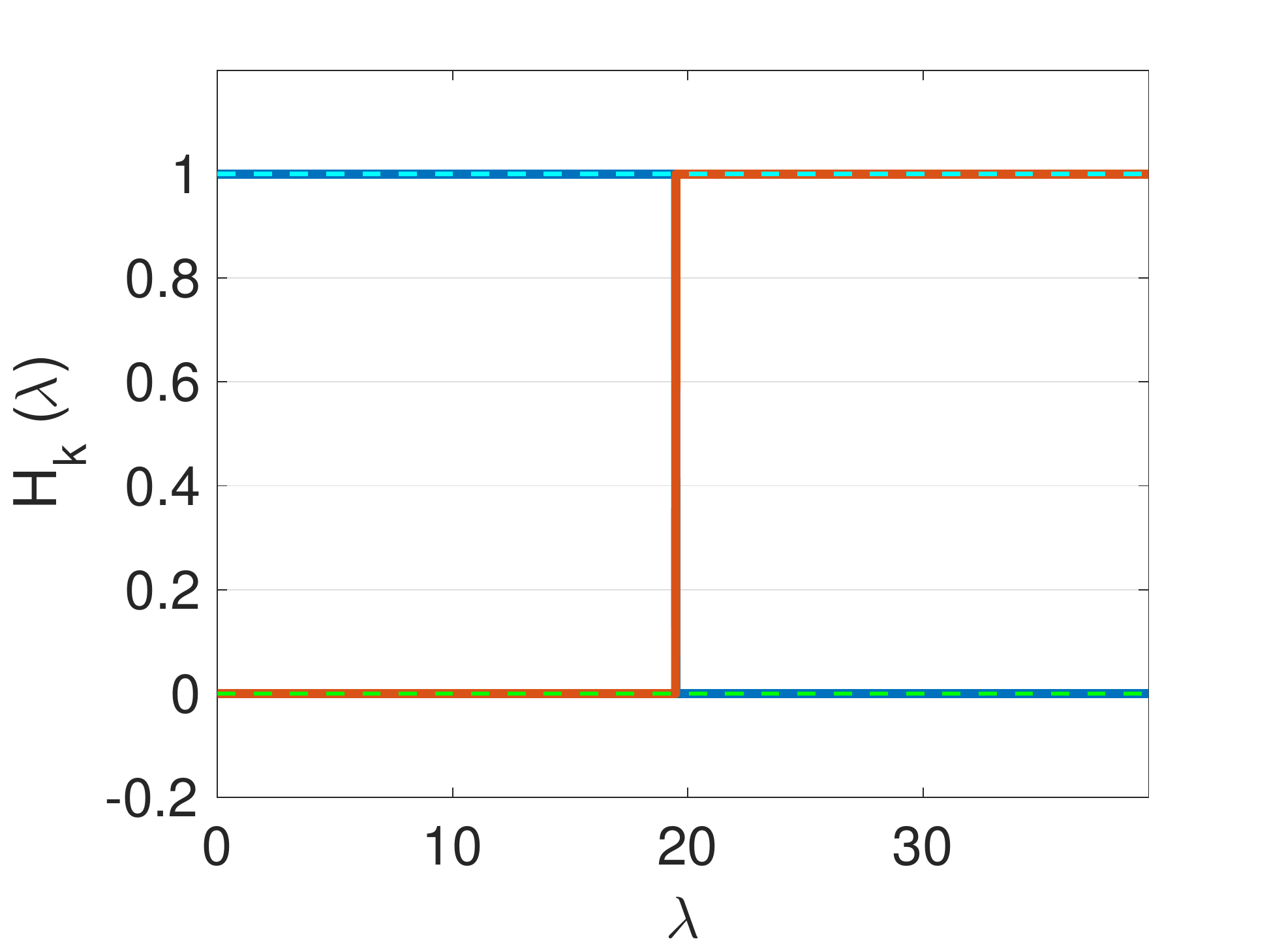}}\ 
\subfigure[][GraphSS-O]{\includegraphics[width = 0.32\linewidth]{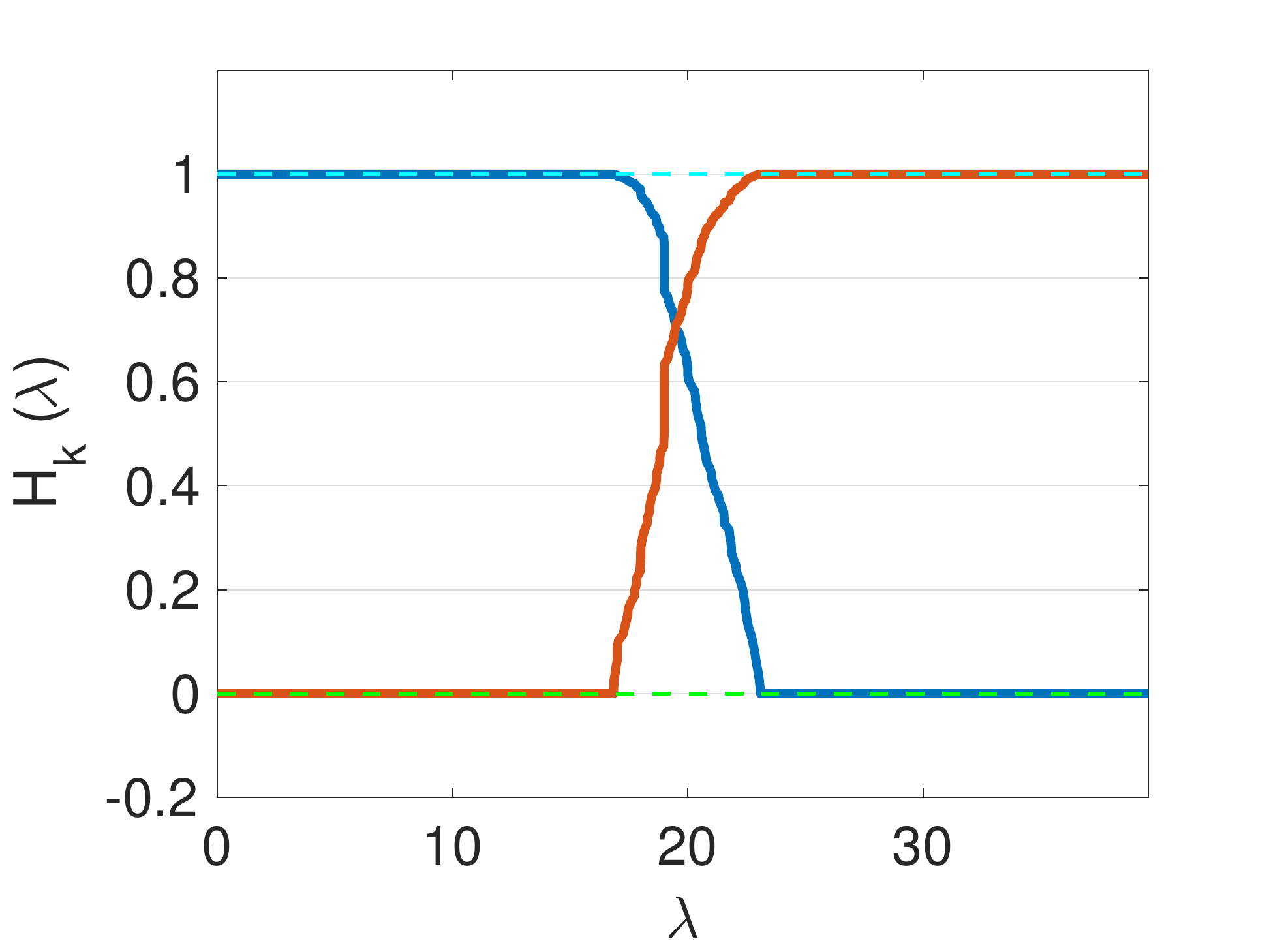}}\ 
\subfigure[][GraphSS-B]{\includegraphics[width = 0.32\linewidth]{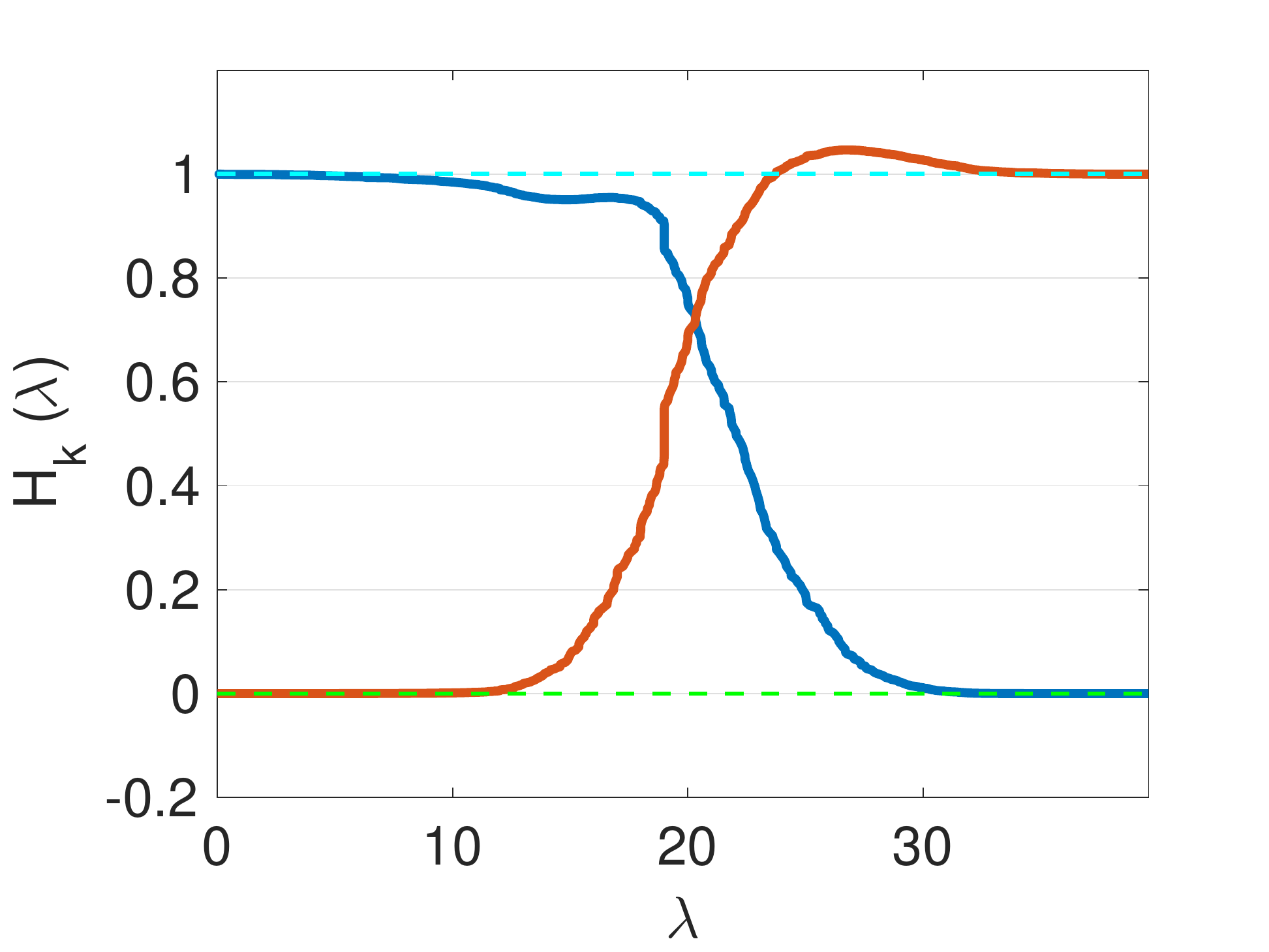}}
\caption{Proposed CS GFBs with spectral domain sampling. Light blue and yellow dashed lines represent \eqref{eq:pr1-1} and \eqref{eq:pr1-2}, respectively.}
\label{fig:CSSGWTs_proposed}
\end{figure}

\subsection{Design Methods}
\label{subsec:designmethods}
By using a method similar to \cite{Sakiya2016a}, both orthogonal and biorthogonal filters can be designed on the basis of those used in classical signal processing.

\subsubsection{Orthogonal Solution}
Similarly to the design of GraphQMF \cite{Narang2012}, the proposed orthogonal CS GFB uses one prototype filter $H_0(\li)$. The remaining $G_0(\lambda_i)$, $H_1(\lambda_i)$, and $G_1(\lambda_i)$ are calculated from $H_0(\li)$, i.e., $H_k(\lambda_i) = G_k(\lambda_i)$ and $H_1(\lambda_i)=H_0(\lambda_{N-i-1})$. $H_0(\lambda_i)$ has to satisfy the following condition to ensure perfect reconstruction:
\begin{equation}
H^{2}_{0}(\lambda_i) + H^{2}_{0}(\lambda_{N-i-1}) = c^2.
\end{equation}
We utilize the frequency responses of the time domain filters to design $H_0(\lambda_i)$. First, a real-valued function $H^{\text{freq}}(\omega)$, where $\omega \in [0,\pi]$, is obtained from the time-domain filter; then $H_0(\lambda_i)$ is calculated according to the eigenvalue distribution of the graph Laplacian. That is,
\begin{equation}
H_0(\lambda_i) = H^{\text{freq}}\left(\frac{\pi i}{N}\right).
\end{equation}

\subsubsection{Biorthogonal Solution}
Similar to graphBior \cite{Narang2013} and wavelets in the time domain, the high-pass filters used in the proposed biorthogonal CS GFB are defined from the low-pass filters as
\begin{equation}
H_1(\lambda_i) = G_0(\lambda_{N-i-1}), \quad G_1(\lambda_i) = H_0(\lambda_{N-i-1}).
\end{equation}
This leads to the following condition:
\begin{equation}
H_0(\lambda_i)G_0(\lambda_i)+H_0(\lambda_{N-i-1})G_0(\lambda_{N-i-1}) =c^2.
\end{equation}

\subsection{Complexity}
The proposed filter bank operates in the GFT domain, thus requiring the eigendecomposition of the variation operator in order to obtain the GFT basis. Typically, this requires $\mathcal{O}(N^3)$ complexity. However, it is important to note that the same graph is often used many times, i.e., many signals will be transformed with one underlying graph. In such a case, we only need to calculate the GFT basis once.

Calculating the GFT coefficients, i.e., $\widetilde{\mathbf{f}} = \mathbf{U}^\top\mathbf{f}$, needs a matrix-vector multiplication and its complexity is $\mathcal{O}(N^2)$. This is required for every input signal. Alternatively, ``fast GFT'' approaches such as \cite{LeMag2018, Lu2017} can be used to reduce the computation cost.

Since the filtering and sampling in the graph frequency domain have much less complexities ($\mathcal{O}(N)$), the entire complexity is dominated by the GFT. One exception is the bipartite case in Theorem \ref{th:bipartite} (in the next section); in this case, the vertex domain sampling also coincides with the spectral domain sampling, i.e., transforms with the vertex domain sampling inherit the properties of those with the spectral domain sampling. The study for more general case to avoid eigendecomposition is one of future works.

\section{Relationship between Vertex and Spectral Domain Sampling-Based Approaches}\label{sec:V}
This section examines the relationship between CS GFBs with vertex domain sampling and those with spectral domain sampling. 

\subsection{Reduced-Size Bipartite Graphs with Kron Reduction}
Kron reduction is a widely used method to reduce the size of the graph Laplacian, especially in multiscale transforms of graph signals \cite{Dorfle2013}. The following theorem describes the condition under which vertex domain sampling is identical to spectral domain sampling.

\begin{theo}
\label{th:bipartite}
If the original graph is bipartite with $|\mathcal{L}|=|\mathcal{H}| = N/2$ and the reduced graph is obtained by Kron reduction \cite{Dorfle2013}, the signal downsampled by (GD2) is equivalent to the one downsampled by (GD1) up to scaling factor when the symmetric normalized graph Laplacian is used as the variation operator whose eigenvalues are arranged in ascending order.
\end{theo}

\begin{proof}
Without loss of generality, we can assume the vertices in the first half correspond to $\mathcal{L}$ and in the second half to $\mathcal{H}$. According to \eqref{eqn:normalizedGL}, the symmetric normalized graph Laplacian, whose eigenvalues are arranged in ascending order, is represented as
\begin{equation}
{\bm{\mathcal{ L}}}=\mathbf{V}
\begin{bmatrix}
\mathbf{\Lambda}_\mathcal{L}& \mathbf{0}\\
\mathbf{0} & \mathbf{J}(2\mathbf{I}-\mathbf{\Lambda}_\mathcal{L})\mathbf{J}
\end{bmatrix}
\mathbf{V}^\top,\label{eqn:L_ascending}
\end{equation}
where
\begin{equation}
\mathbf{V} := 
\begin{bmatrix}
\mathbf{U}_{\mathcal{LL}} &\mathbf{U}_{\mathcal{LL}} \\
\mathbf{U}_{\mathcal{HL}} & -\mathbf{U}_{\mathcal{HL}}
\end{bmatrix}
\begin{bmatrix}
\mathbf{I} &\mathbf{0} \\
\mathbf{0} & \mathbf{J}
\end{bmatrix},
\end{equation}
where we have used the same notation as in \eqref{eqn:normalizedGL}.

The Kron reduction of \eqref{eqn:L_ascending} is represented as
\begin{equation}
\begin{split}
\bm{\mathcal{L}}_{\text{reduced}}&={\bm{\mathcal{ L}}}_{\mathcal{LL}}-{\bm{\mathcal{ L}}}_{\mathcal{LH}}{\bm{\mathcal{ L}}}_{\mathcal{H}\mathcal{H}}^{-1}{\bm{\mathcal{ L}}}_{\mathcal{H}\mathcal{L}}\\
&=\mathbf{I}-\mathbf{U}_{\mathcal{LL}}(2\mathbf{\Lambda}_\mathcal{L}-2\mathbf{I})\mathbf{U}_{\mathcal{HL}}^\top\mathbf{I}\mathbf{U}_{\mathcal{HL}}(2\mathbf{\Lambda}_\mathcal{L}-2\mathbf{I})\mathbf{U}_{\mathcal{LL}}^\top\\
&=\mathbf{I}-2\mathbf{U}_{\mathcal{LL}}(\mathbf{\Lambda}_\mathcal{L}-\mathbf{I})^2\mathbf{U}_{\mathcal{LL}}^\top\\
&= \sqrt{2}\mathbf{U}_{\mathcal{LL}}(-\mathbf{\Lambda}_\mathcal{L}^2 + 2\mathbf{\Lambda}_\mathcal{L})(\sqrt{2}\mathbf{U}_{\mathcal{LL}})^\top.\\
\end{split}
\end{equation}
Therefore, the eigenvector and eigenvalue matrices of the reduced graph are $\sqrt{2}\mathbf{U}_{\mathcal{LL}}$ and $-\mathbf{\Lambda}_\mathcal{L}^2 + 2\mathbf{\Lambda}_\mathcal{L}$, respectively\footnote{Note that $\bm{\mathcal{ L}}_{\text{reduced}}$ is not guaranteed to be positive semidefinite, i.e., its eigenvalues are not always positive in general.}.

Recall that the downsampling matrix in the vertex domain can be represented as
\begin{equation}
\mathbf{S}_d=\begin{bmatrix}
\mathbf{I}_{|\mathcal{L}|} &\mathbf{0}_{|\mathcal{H}|}
\end{bmatrix},
\end{equation}
and the downsampling matrix created by spectral domain sampling in this case is
\begin{equation}
\begin{split}
\mathbf{U}_{\mathcal{LL}}\begin{bmatrix}
\mathbf{I}_{|\mathcal{L}|} &\mathbf{J}_{|\mathcal{H}|}
\end{bmatrix}
\mathbf{V}^\top & = \mathbf{U}_{\mathcal{LL}}\begin{bmatrix}
\mathbf{I}_{|\mathcal{L}|} &\mathbf{I}_{|\mathcal{H}|}
\end{bmatrix}
\begin{bmatrix}
\mathbf{U}_{\mathcal{LL}}^\top &\mathbf{U}_{\mathcal{HL}}^\top \\
\mathbf{U}_{\mathcal{LL}}^\top & -\mathbf{U}_{\mathcal{HL}}^\top
\end{bmatrix}\\
& =\begin{bmatrix}
\mathbf{I}_{|\mathcal{L}|} &\mathbf{0}_{|\mathcal{H}|}
\end{bmatrix}.
\end{split}
\end{equation}
This completes the proof.
\end{proof}

\begin{figure*}[tp]
\centering
\subfigure[][]{\includegraphics[width = 0.25\linewidth]{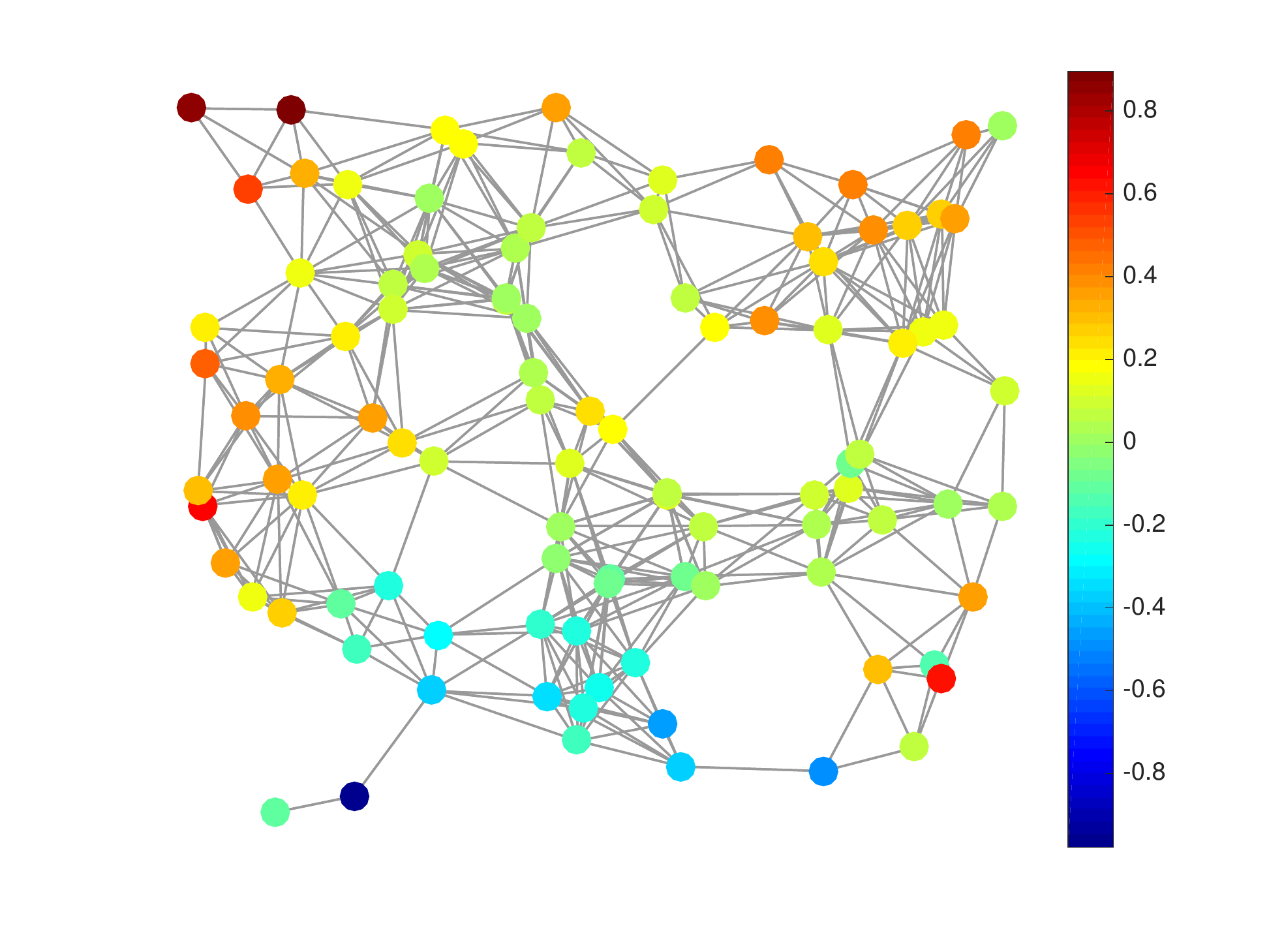}}
\subfigure[][]{\includegraphics[width = 0.25\linewidth]{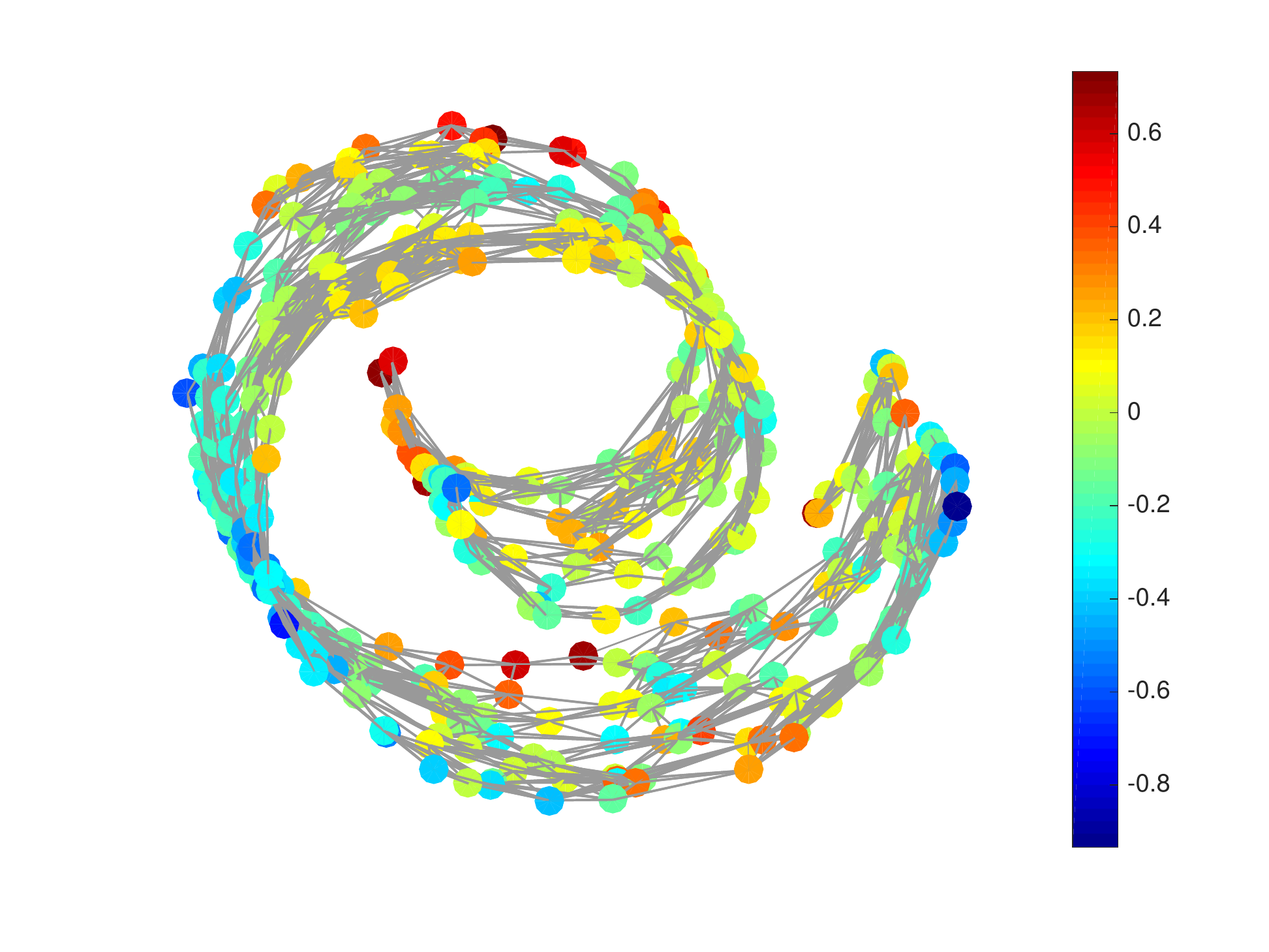}}
\subfigure[][]{\includegraphics[width = 0.25\linewidth]{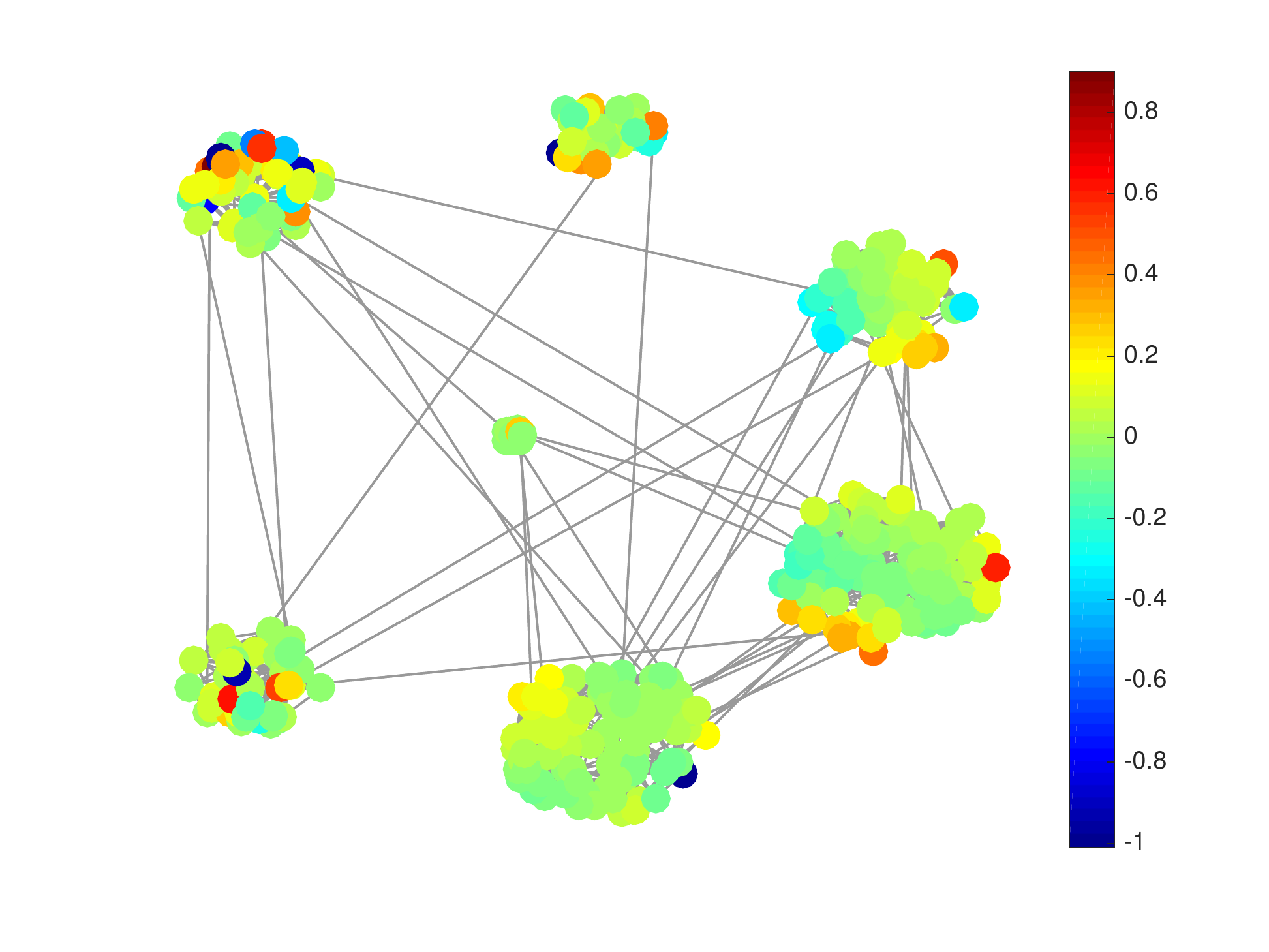}}\\
\subfigure[][]{\includegraphics[width = 0.25\linewidth]{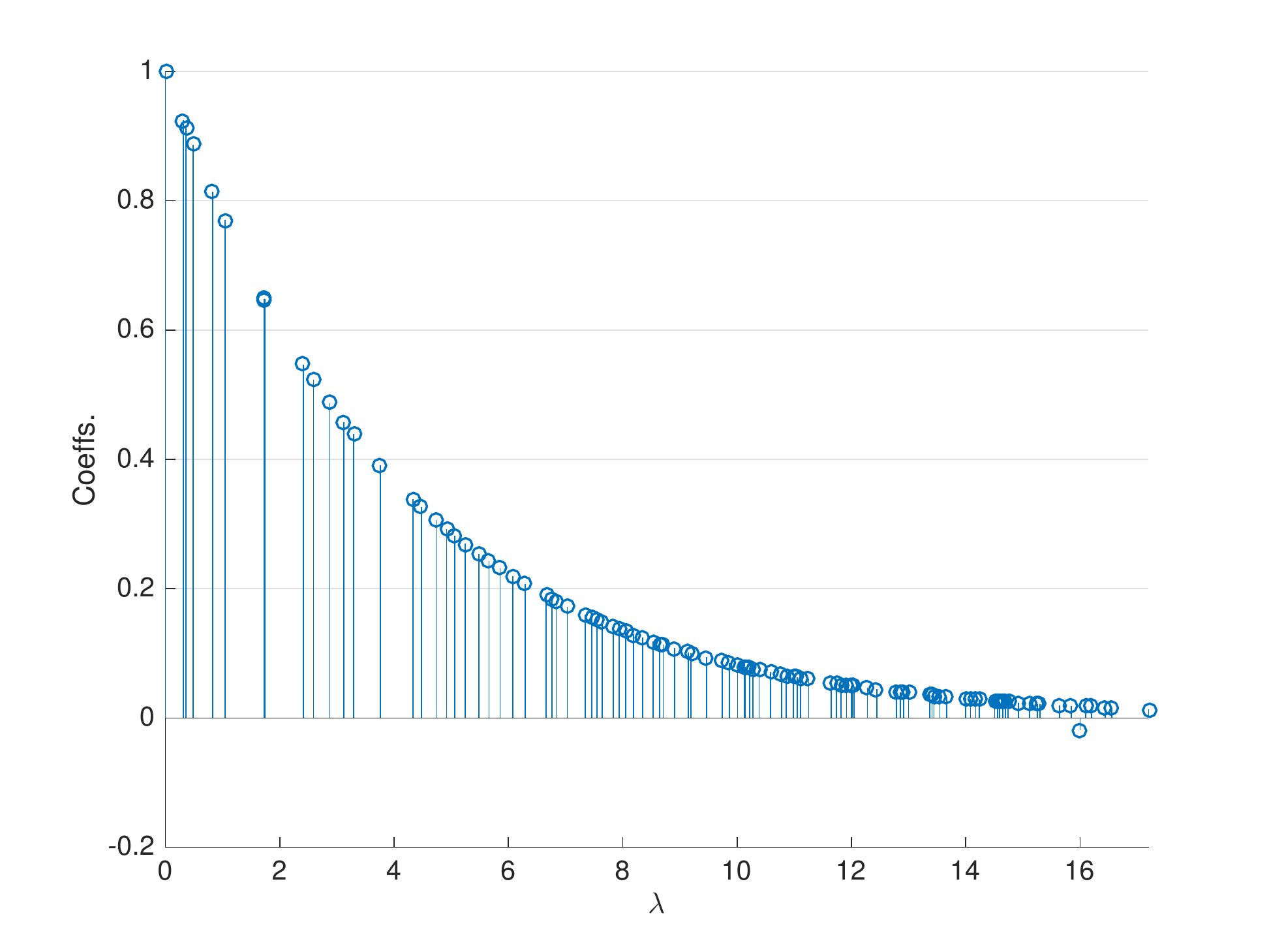}}
\subfigure[][]{\includegraphics[width = 0.25\linewidth]{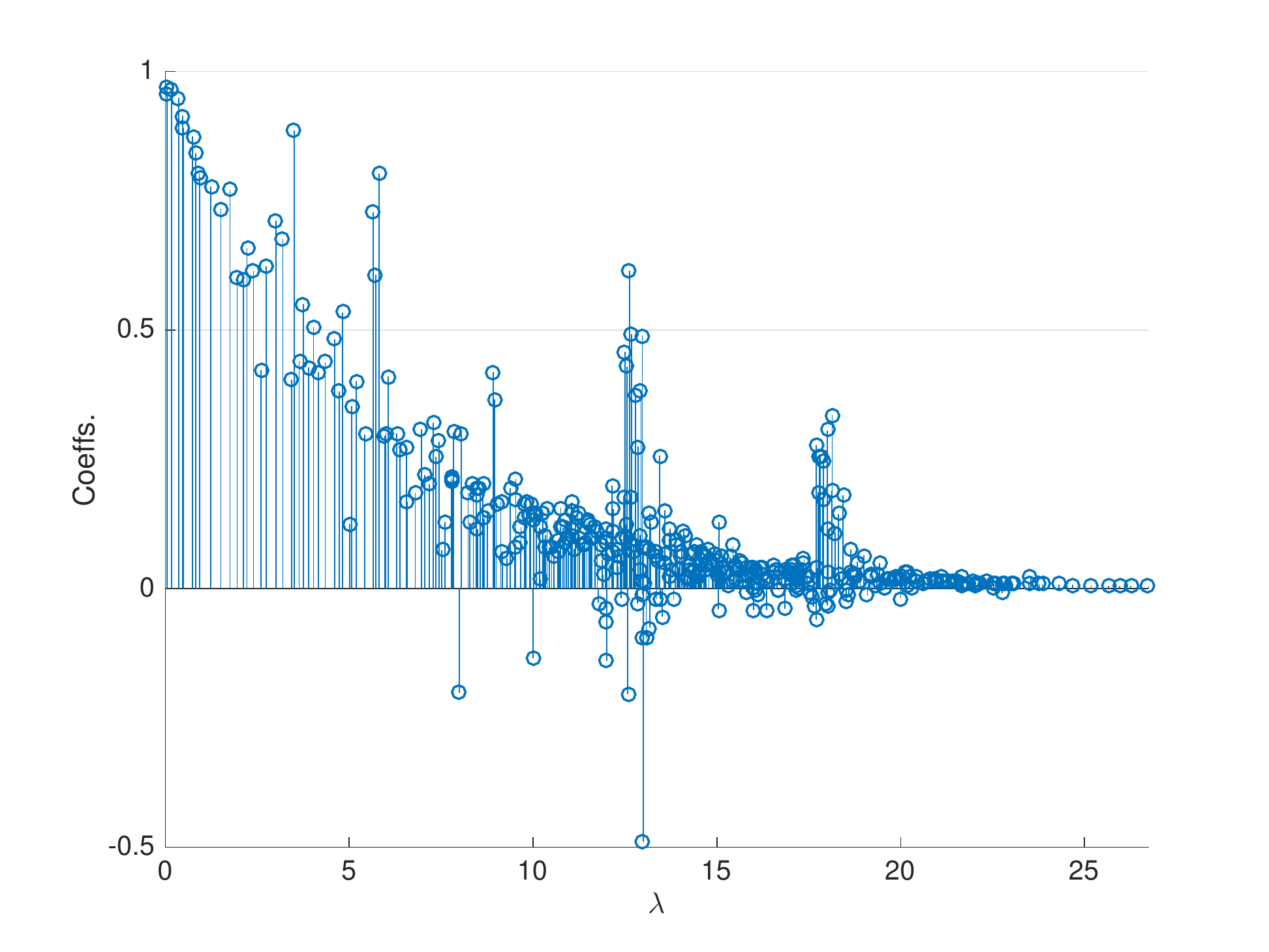}}
\subfigure[][]{\includegraphics[width = 0.25\linewidth]{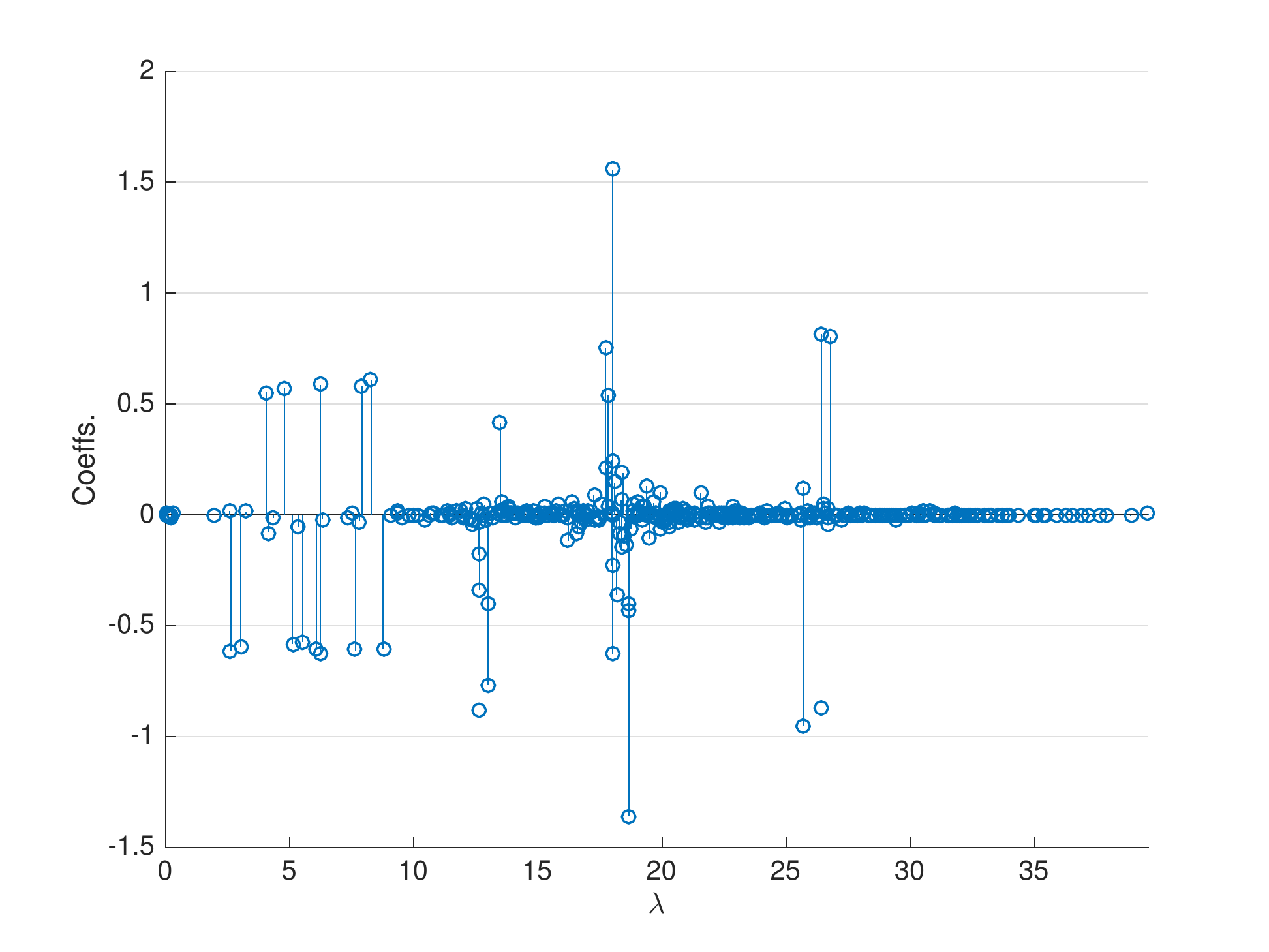}}
\caption{Original graph signals. Top row: Signals in the vertex domain. Bottom row: Signals in the graph frequency domain. From left to right: Community graph ($N=400$), sensor graph ($N=100$), and Swiss roll graph ($N=400$).
}
\label{fig:spectrum}
\end{figure*}

\subsection{Perfect Reconstruction Condition for Bipartite Graphs}
The following theorem reveals that the perfect reconstruction condition of the proposed method coincides with that of vertex domain sampling in a special case.
\begin{theo}
If the underlying graph is a bipartite graph and the symmetric normalized graph Laplacian is used as the variation operator, the perfect reconstruction condition for the two-channel CS GFB using the vertex domain sampling is identical to that using spectral domain sampling.
\end{theo}
\begin{proof}
The eigenvalue distribution of the normalized graph Laplacian of a bipartite graph is symmetric with respect to $\lambda = 1$, and the maximum eigenvalue is $2$.
Therefore, $\lambda_{N-1-i} = 2 - \lambda_i$, which implies that the perfect reconstruction condition \eqref{eq:pr1-1} and \eqref{eq:pr1-2} is identical to \eqref{eq:pr2-1} and \eqref{eq:pr2-2}.
\end{proof}

\section{Experiments}\label{sec:VI}
Nonlinear approximations and denoising for synthetic graph signals were selected as target applications of the numerical performance comparisons. The bases of the comparisons were existing GWTs/GFBs. Hereafter, we abbreviate the proposed filter bank as \emph{GraphSS}, where SS refers to spectral sampling. We used GSPBOX \cite{Perrau2014} for graph construction and visualizations.

\subsection{Setup}
The prototype filters used for GraphSSs were:
\begin{itemize}
   \item The ideal filters (denoted as GraphSS-I). The set of ideal filters clearly satisfies \eqref{eq:pr1-1} and \eqref{eq:pr1-2}.
  \item  The orthogonal filter set designed with the Meyer kernel as $H^{\text{freq}}(\omega)$ (denoted as GraphSS-O).
  \item The biorthogonal filter set based on 9/7-CDF filters (denoted as GraphSS-B).
\end{itemize}
Their spectral characteristics are shown in Figs. \ref{fig:CSSGWTs_proposed}(a)--(c). As described above, the proposed transforms can be applied to both combinatorial and symmetric normalized graph Laplacians, so we decided to examine both cases. In what follows, the combinatorial version is specified by (C), e.g., GraphSS-O(C), whereas the normalized one is specified by (N).

The compared methods were:
\begin{itemize}
  \item GraphQMF \cite{Narang2012}
  \item GraphBior \cite{Narang2013}
  \item GraphFC \cite{Sakiya2016a}
  \item Diffusion wavelet (abbreviated as DiffWav) \cite{Coifma2006}
  \item SubGFB \cite{Trembl2016}
  \item Graph Laplacian pyramid (abbreviated as GLP) \cite{Shuman2016}
\end{itemize}
We used MATLAB codes provided by the authors. Note that the existing GraphQMF, GraphBior, and GraphFC require bipartition of the underlying graph. We used the coloring-based bipartition \cite{Narang2013,Narang2012,Harary1977} as suggested by the authors. All transforms except SubGFB decomposed graph signals into two-level octave bands. SubGFB needs to decompose the original graph into several subgraphs for a multi-level transform. We used a one-level cascade suggested by the authors' code. Although DiffWav and GLP are not CS transforms, they are used for a comparison purpose.

Three synthetic graph signals shown in Fig. \ref{fig:spectrum} were used in the experiments. They have different characteristics; The first graph signal shown in Fig. \ref{fig:spectrum}(a) is smooth in the graph frequency domain, where $\widetilde{f}[i] = \exp (-\lambda_i/4)$, and its spectrum is shown in Fig. \ref{fig:spectrum}(d). The second one is the sum of the spectral localized signal and the exponential one, as shown in Figs. \ref{fig:spectrum}(b) and (e). The third signal shown in Fig. \ref{fig:spectrum}(c) is localized both in the vertex and spectral domains (see Figs. \ref{fig:spectrum}(f)). We designed the signal by using the method in \cite{Shuman2015}. Specifically, $\mathbf{f} = \sum^{4}_{j = 1} \mathbf{f}_j / \| \mathbf{f}_j \|_{\infty}$, where
\begin{equation}
f_j[i] = \mathbbm{1}_{\{\text{vertex $i$ is in cluster $j$}\}}\sum^{N-1}_{\ell =0}u_{\ell}[i] \mathbbm{1}_{\{ \underline{\tau}_j \le \lambda_{\ell} \le \overline{\tau}_j \}}.
\end{equation}
We took the sequence $\{ \underline{\tau}_j, \overline{\tau}_j \}_{j = 1,\ldots , 4}$ to be $[\lambda_{9}, \lambda_{29}], [\lambda_{59}, \lambda_{79}], [\lambda_{149}, \lambda_{169}], [\lambda_{299}, \lambda_{319}]$.


\begin{figure*}[tp]
\centering
\subfigure[][Random sensor]{\includegraphics[width = 0.32\linewidth]{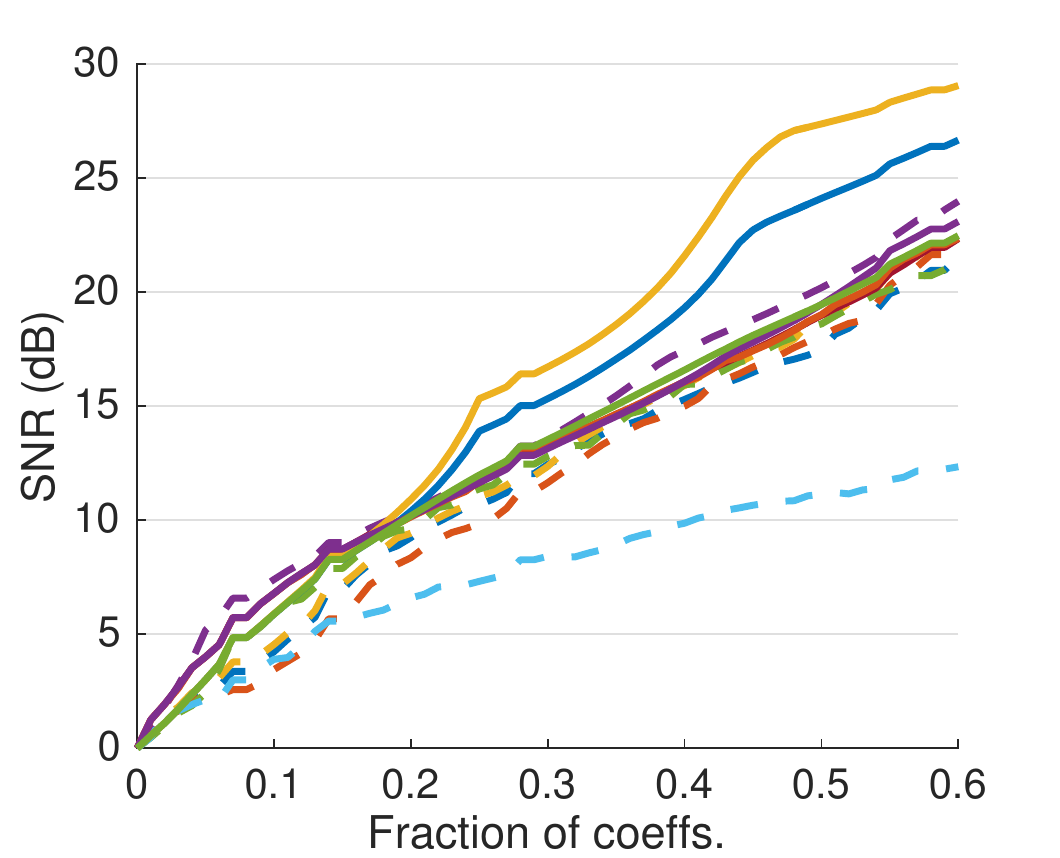}}
\subfigure[][Swiss roll]{\includegraphics[width = 0.32\linewidth]{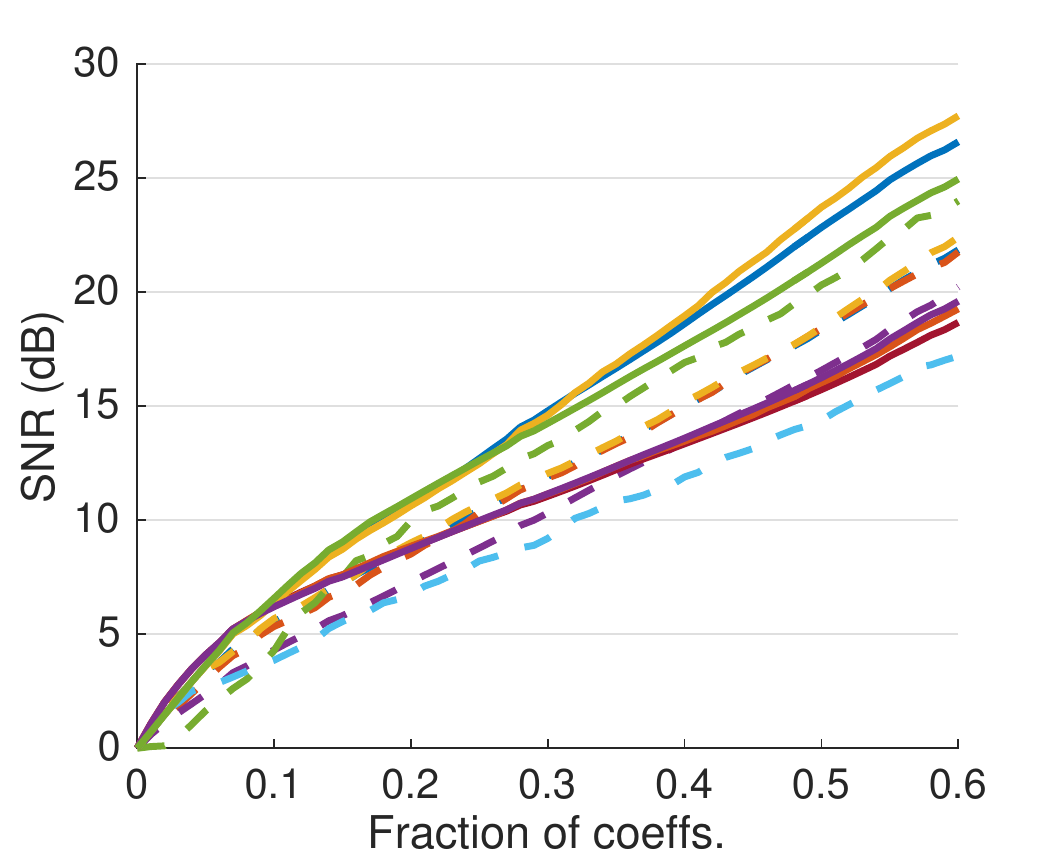}}
\subfigure[][Community]{\includegraphics[width = 0.32\linewidth]{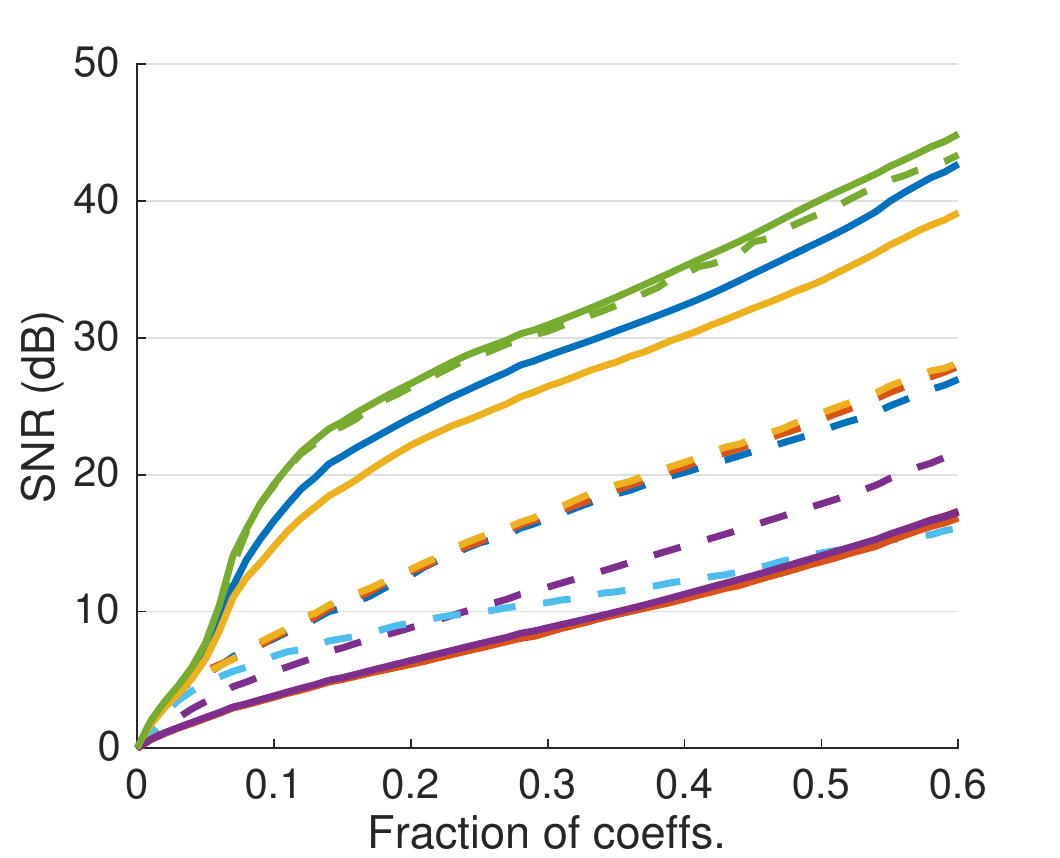}}\\
\subfigure[][Legend]{\includegraphics[width = 0.55\linewidth]{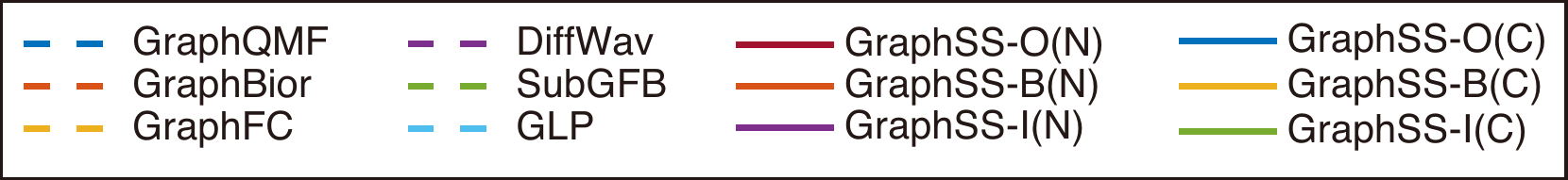}}
\caption{Results of nonlinear approximation.
When the signals contain spectrally-smooth components like (a) and (b), GraphSSs with the combinatorial graph Laplacian significantly outperform the existing methods. For the signal localized in the vertex domain (Fig. \ref{fig:spectrum}(c)), GraphSSs and SubGFB have comparable performances.
}
\label{fig:nla}
\end{figure*}


\subsection{Nonlinear Approximation}
In nonlinear approximation, we keep the fraction of the transformed coefficients having high absolute values and set the remaining coefficients to zero. Figs. \ref{fig:nla}(a)--(c) show the results. GraphSSs with the combinatorial Laplacian gave better SNRs than the conventional methods did and the ones with the normalized graph Laplacian.

When the spectrum was smooth, GraphSSs significantly outperformed the existing methods, as shown in Fig. \ref{fig:nla}(a). For the signals containing vertex-localized components, the gap between the SNRs was smaller, but the proposed method still had better reconstruction quality. Interestingly, GraphSS-I was not always the best among the proposed transforms (see Figs. \ref{fig:nla}(a) and (b)).

\begin{table*}[tp]
\caption{Denoising Results: SNR (dB). Average of $100$ Runs. The highest SNRs are shown in bold and the second highest ones are underlined.}
\label{table:snr}
\centering
\begin{tabular}{c|r|r|r|r|r|r|r|r|r}
\hline
Methods / Graphs & \multicolumn{3}{c|}{Random sensor} & \multicolumn{3}{c|}{Swiss roll} & \multicolumn{3}{c}{Community} \\\hline\hline
$\sigma$& $1/16$& $1/8$&$1/4$ &  $1/16$& $1/8$&$1/4$ &  $1/16$& $1/8$&$1/4$ \\\hline
GraphQMF&11.61&7.80&3.81&9.96&5.73&2.67&10.43&5.89&1.23\\
GraphBior&11.70&7.82&3.70&10.03&5.73&2.60&10.33&5.93&1.20\\
GraphFC&11.69&7.85&3.82&9.88&5.65&2.70&10.50&5.91&0.96\\
DiffWav&4.18&3.26&2.25&1.60&1.00&0.19&1.77&1.08&0.08\\
SubGFB&11.23&6.82&2.42&10.18&5.97&1.71&\textbf{14.71}&\textbf{9.83}&\textbf{2.65}\\
GLP&8.36&5.69&3.73&9.73&5.51&1.71&7.67&4.85&1.69\\\hdashline
GraphSS-O(N)&11.22&7.75&4.51&9.26&5.91&3.16&7.30&2.77&-0.42\\
GraphSS-O(C)&\underline{12.87}&\underline{10.32}&\underline{6.28}&\underline{11.98}&\underline{7.78}&\textbf{3.76}&12.39&7.58&1.81\\
GraphSS-B(N)&11.20&7.75&4.52&9.29&5.89&3.06&7.25&2.82&-0.48\\
GraphSS-B(C)&\textbf{13.93}&\textbf{10.78}&\textbf{6.32}&\textbf{12.08}&\textbf{7.87}&\underline{3.75}&12.06&7.33&1.69\\
GraphSS-I(N)&11.22&7.68&4.46&9.43&5.90&3.10&7.33&2.81&-0.43\\
GraphSS-I(C)&11.74&9.32&5.76&11.56&7.68&3.65&\underline{13.05}&\underline{8.02}&\underline{2.03}\\\hline\hline
Noisy&13.33&7.34&1.39&11.85&5.83&-0.15&10.10&4.09&-2.00\\\hline
\end{tabular}
\end{table*}%

\begin{figure*}[t]
\centering
\subfigure[][Original]{\includegraphics[width = 0.3\linewidth]{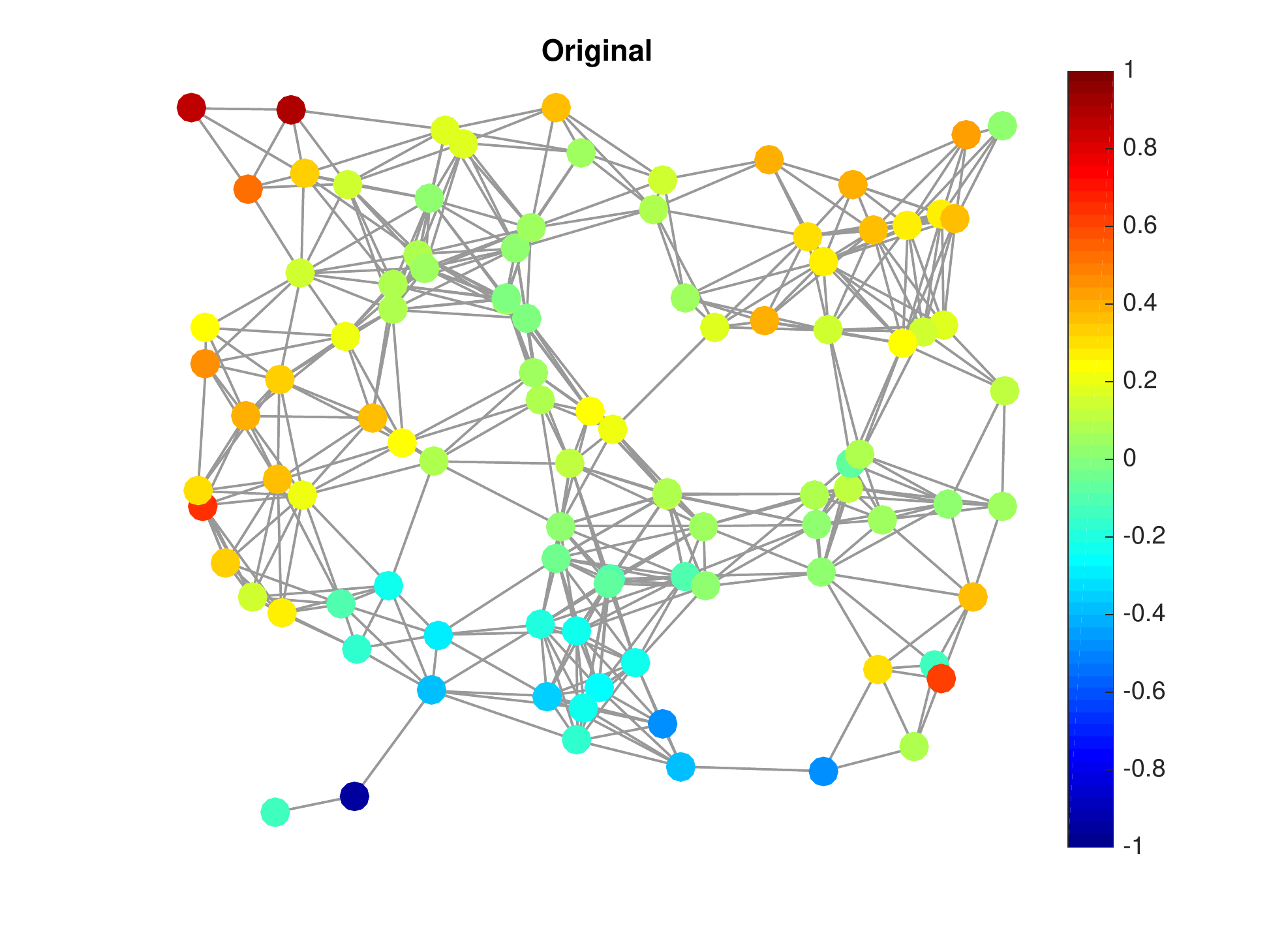}}
\subfigure[][Noisy (SNR: 1.42 dB)]{\includegraphics[width = 0.3\linewidth]{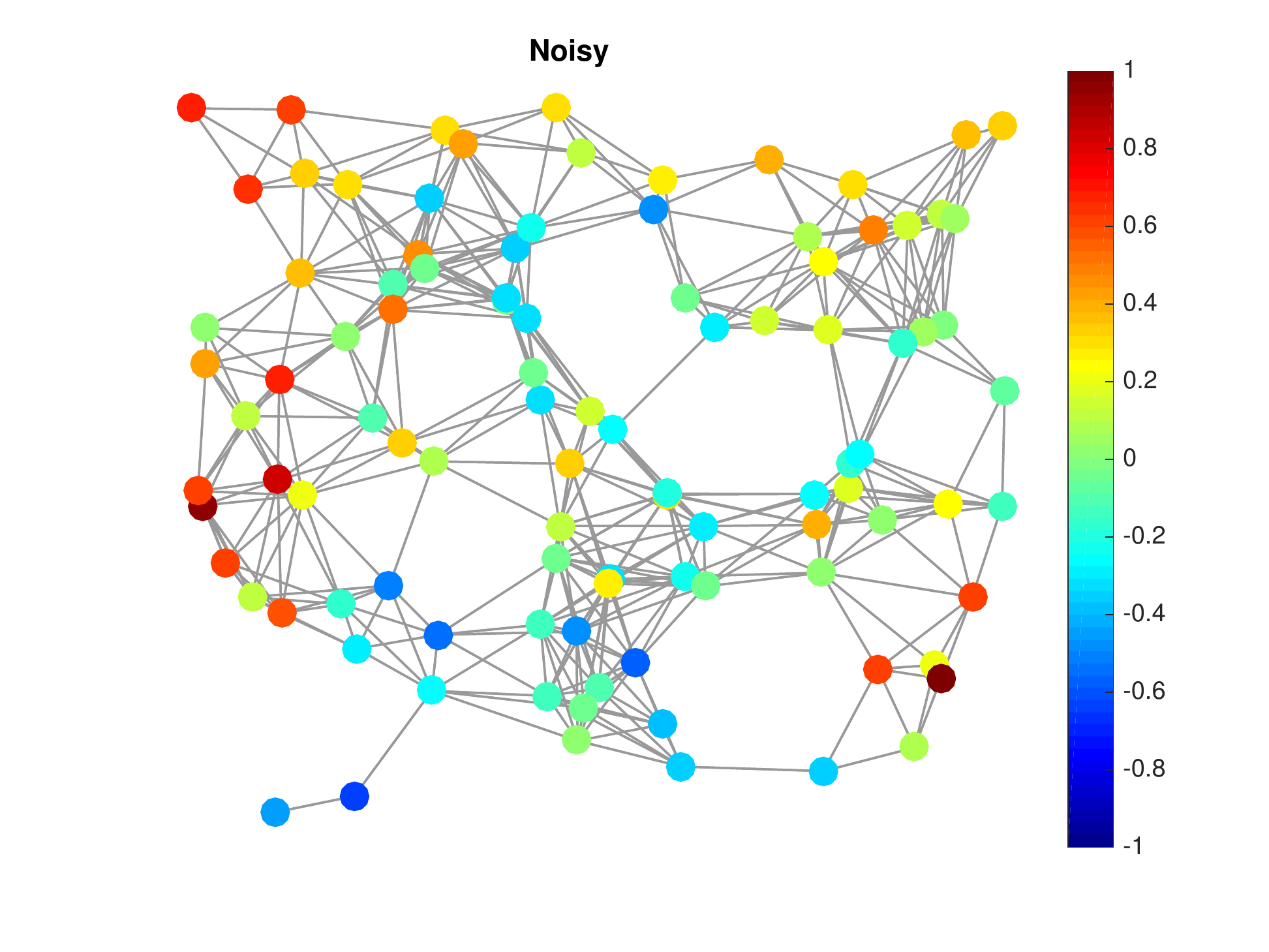}}
\subfigure[][GraphQMF (SNR: 3.72 dB)]{\includegraphics[width = 0.3\linewidth]{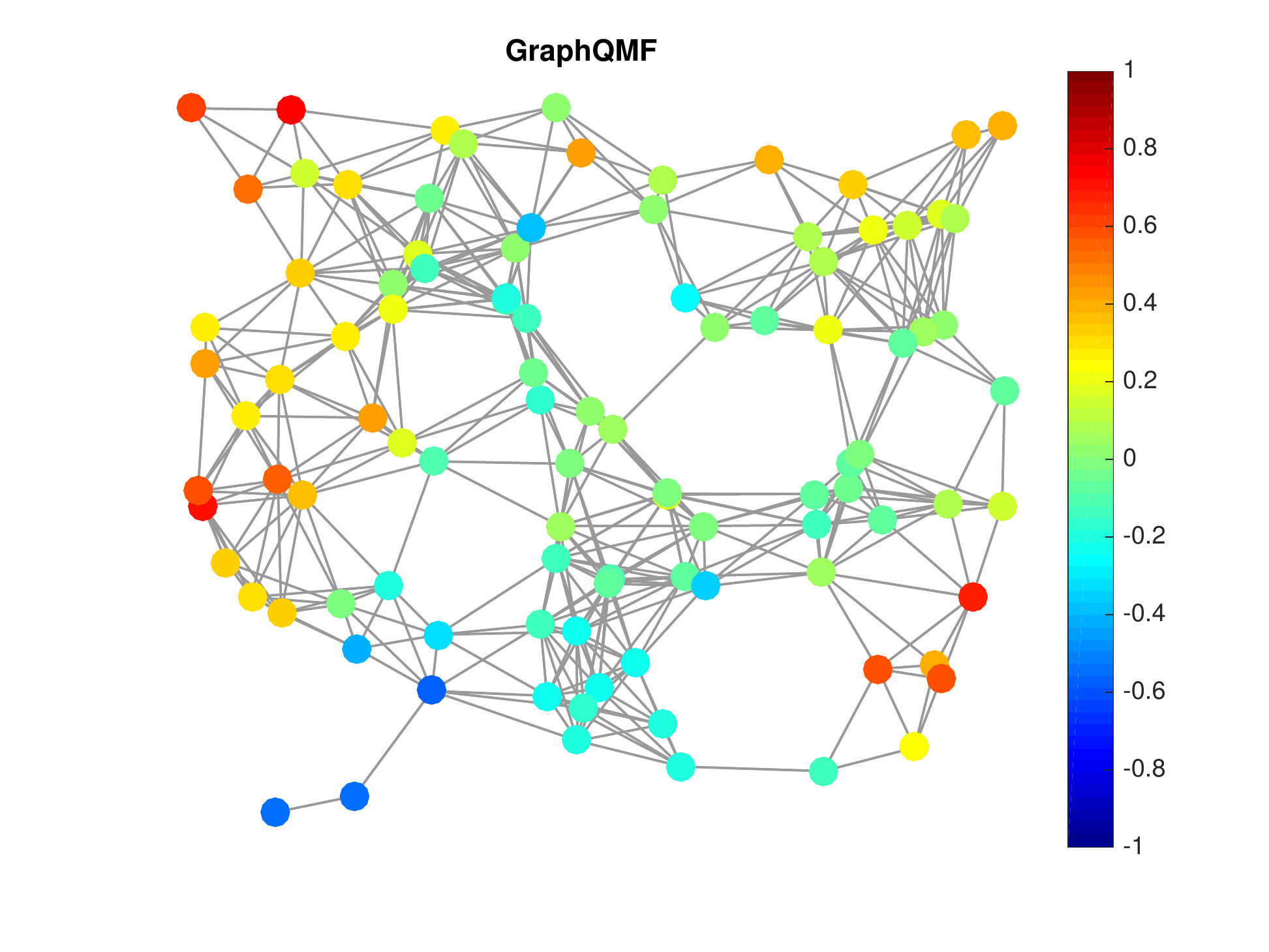}}\\
\subfigure[][GraphBior (SNR: 3.66 dB)]{\includegraphics[width = 0.3\linewidth]{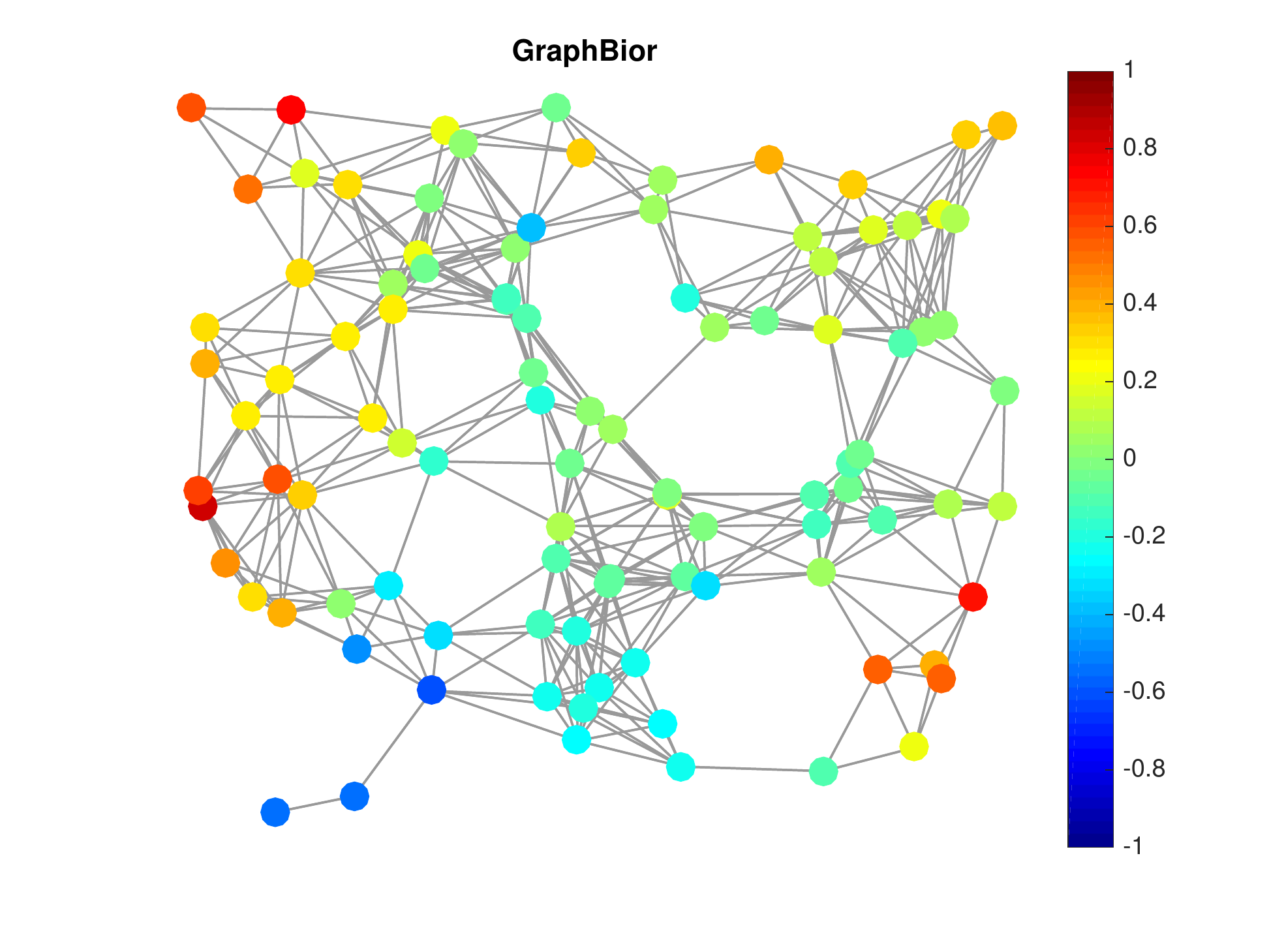}}
\subfigure[][GraphFC (SNR: 3.63 dB)]{\includegraphics[width = 0.3\linewidth]{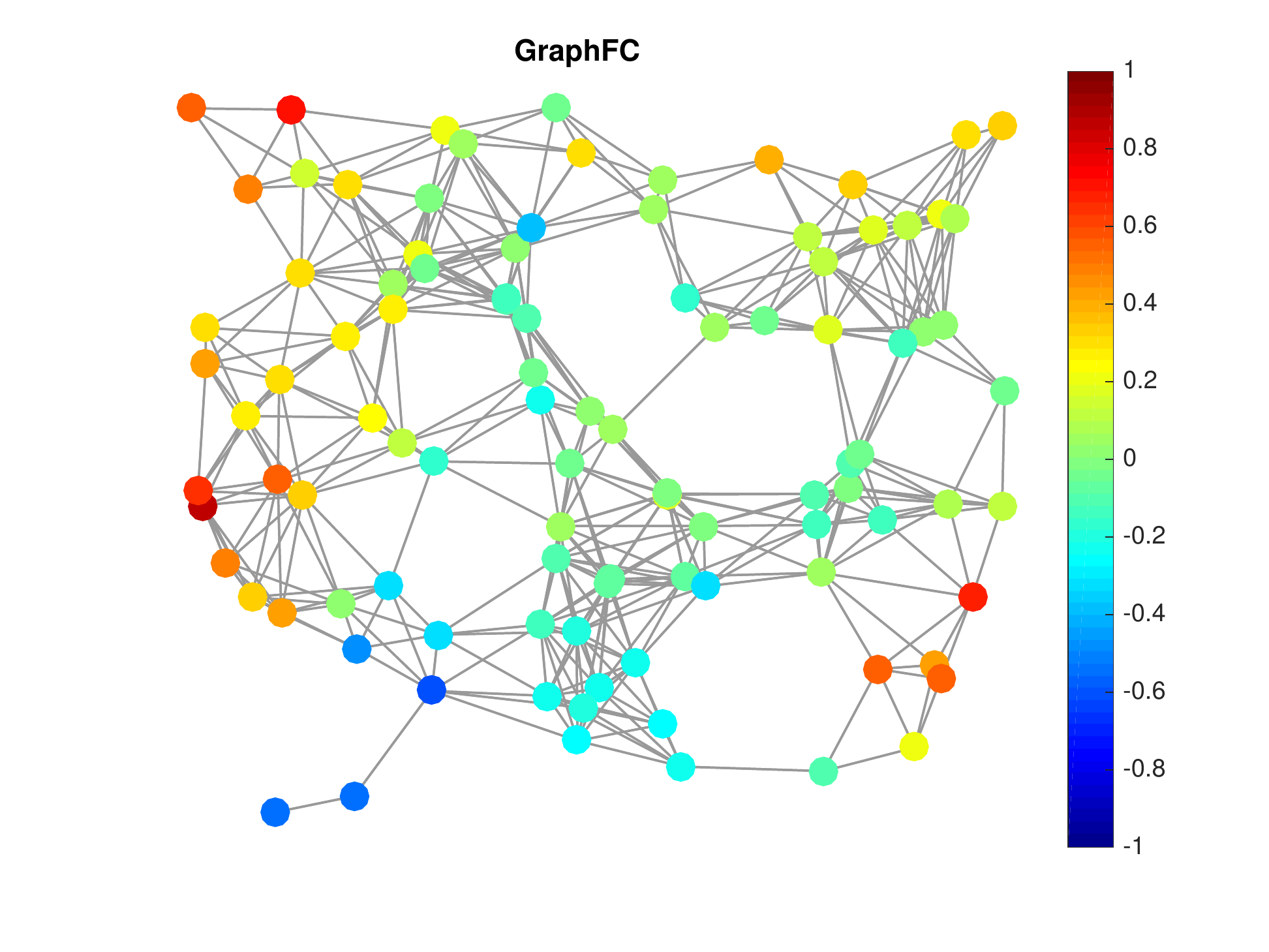}}
\subfigure[][DiffWav (SNR: 1.40 dB)]{\includegraphics[width = 0.3\linewidth]{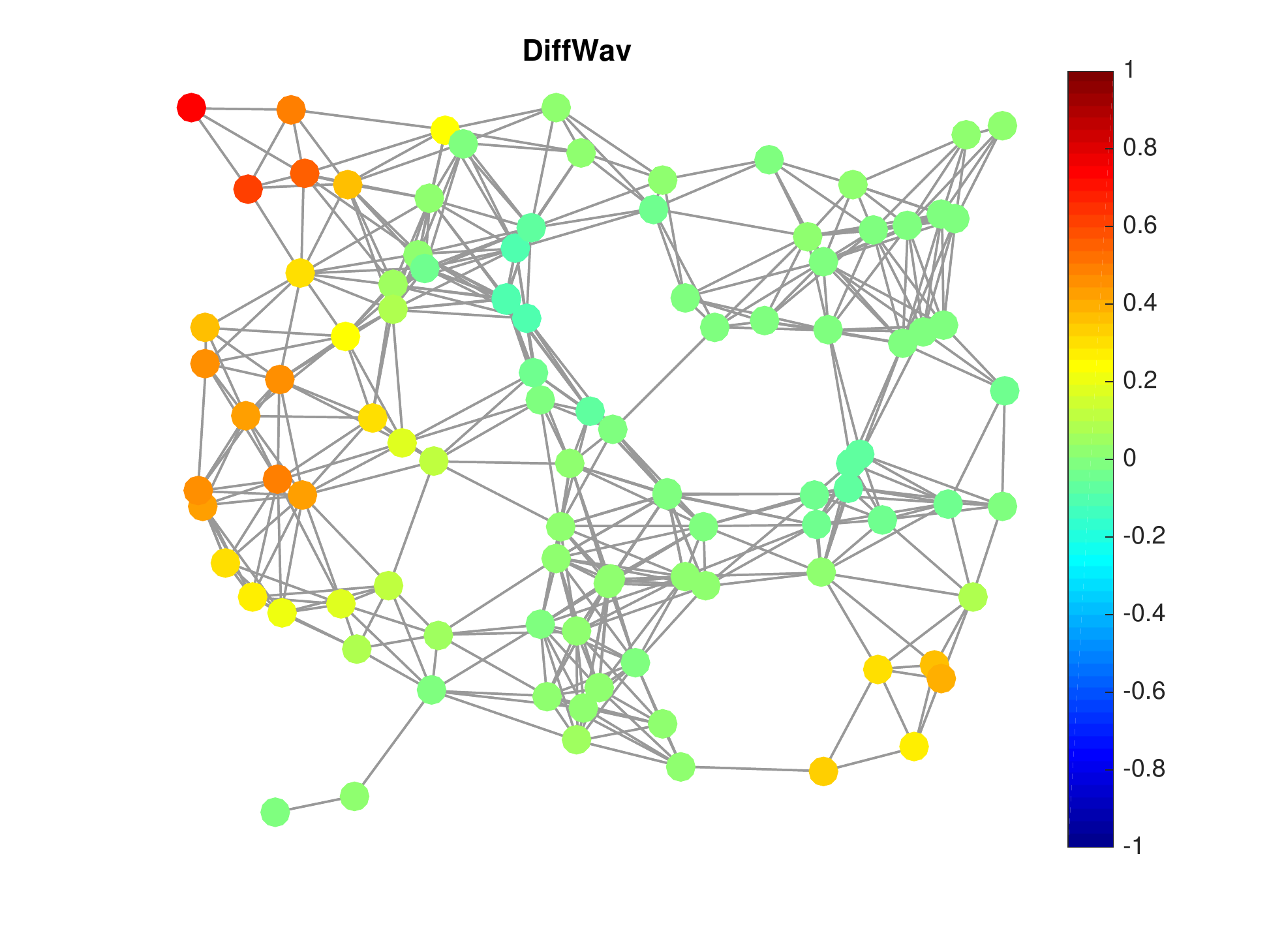}}\\
\subfigure[][SubGFB (SNR: 1.88 dB)]{\includegraphics[width = 0.3\linewidth]{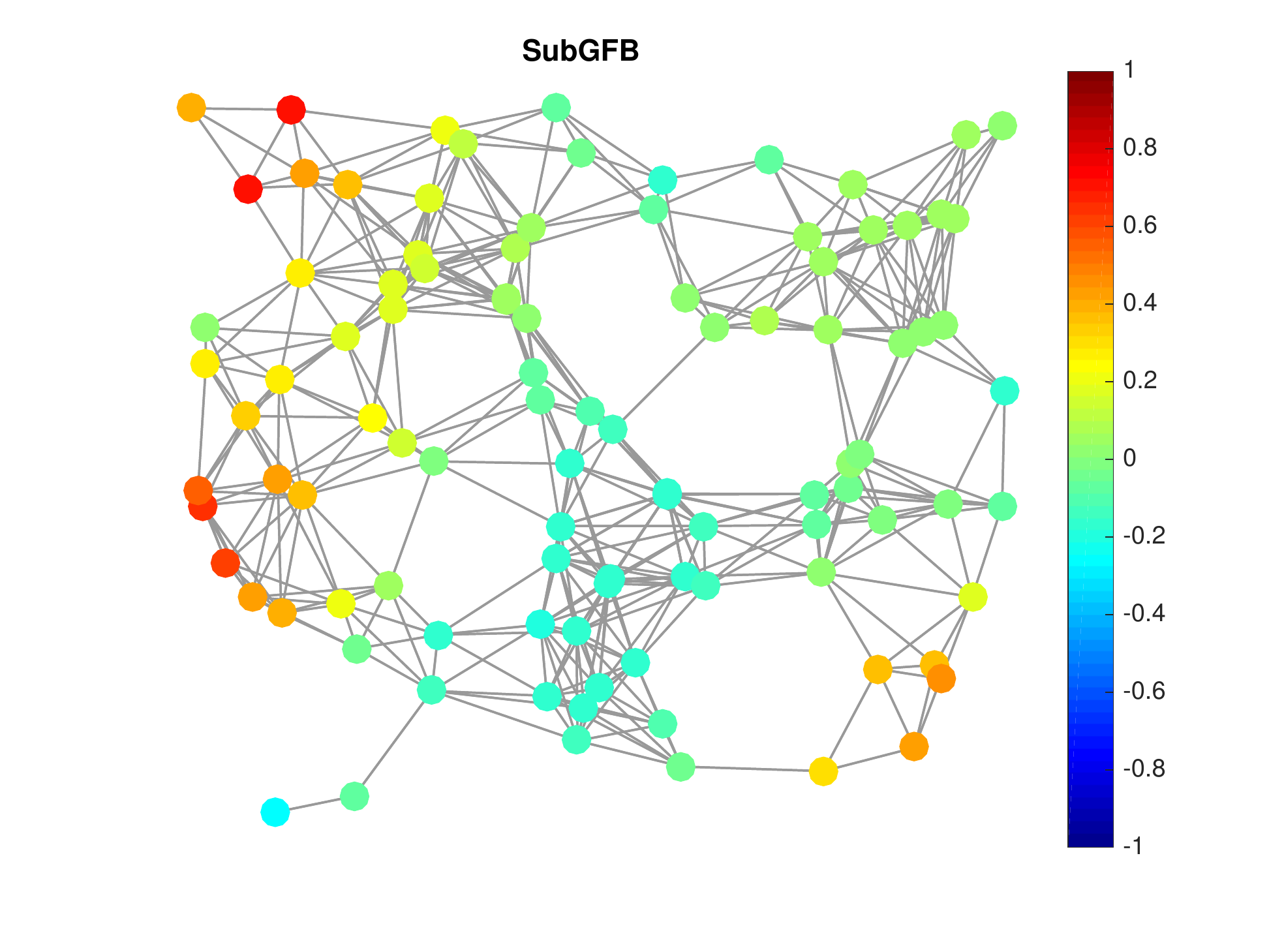}}
\subfigure[][GLP (SNR: 3.45 dB)]{\includegraphics[width = 0.3\linewidth]{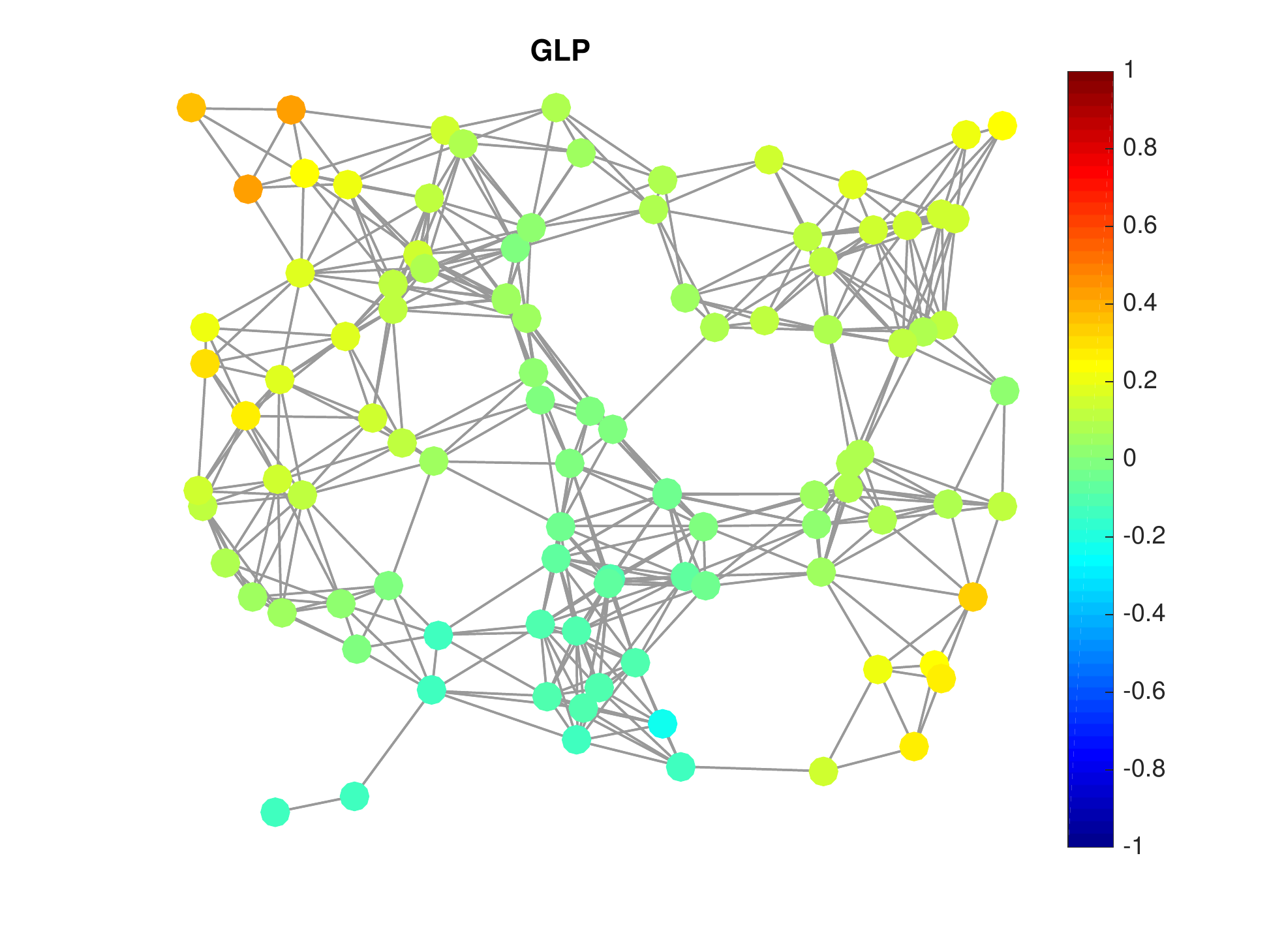}}
\subfigure[][GraphSS-B(C) (SNR: 5.08 dB)]{\includegraphics[width = 0.3\linewidth]{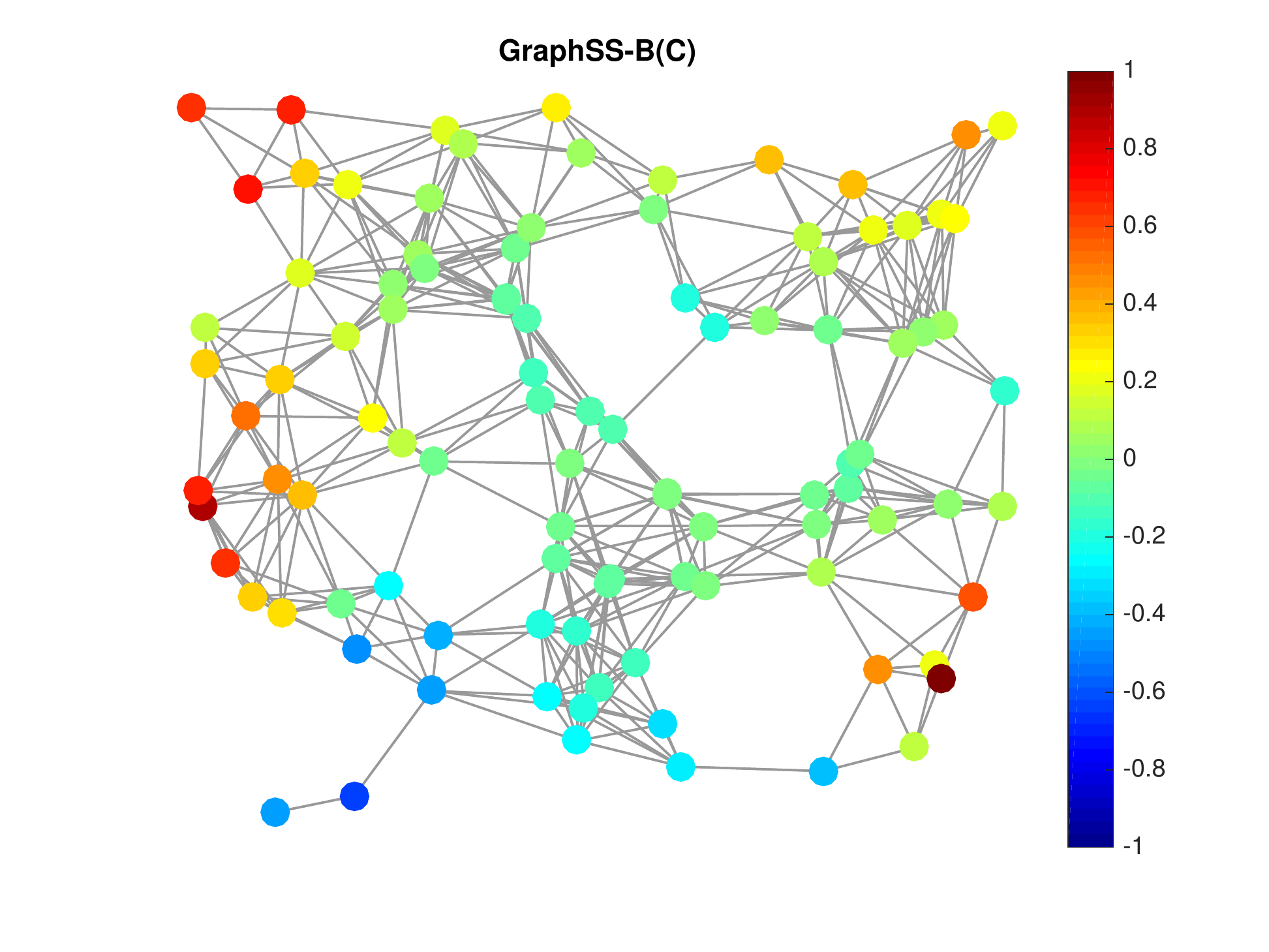}}
\caption{Denoising results.
Some signal values have large errors after denoising using graph filter banks with vertex domain sampling.  In contrast, errors when using the proposed method are small for almost all vertices.
}
\label{fig:denoising}
\end{figure*}

\subsection{Denoising}
Hard thresholding with a threshold of $3\sigma$ was performed on each subband except the lowest to remove additive white Gaussian noise with variance $\sigma^2$. The SNRs after denoising are compared in Table \ref{table:snr}. The proposed methods with the combinatorial graph Laplacian outperformed the conventional methods in most cases, except SubGFB for the vertex-localized signals. Because SubGFB is a Haar-like transform in the vertex domain, it works well for this case. In contrast, GraphSSs present better performance than SubGFB (and other transforms) for signals with spectrally-smooth components.
In addition, due to the small number of scales, DiffWav would have the worst performance in this case.
Similar to the nonlinear approximation, GraphSSs with the combinatorial graph Laplacian are better than those with the normalized one.

The denoised signals for the sensor graph are shown in Fig.  \ref{fig:denoising} where $\sigma = 1/4$. As in the numerical comparison, GraphSS-B(C) presented cleaner signals than others.

\subsection{Discussion: Ideal vs. Non-Ideal Filters}\label{sec:idealvsnonideal}
An interesting and somewhat counter-intuitive result of the experiments was that, in some cases, non-ideal filters (GraphSS-O and GraphSS-B) outperformed the ideal filter (GraphSS-I). To understand why this happened, we should note that our approach involves spectral folding of frequencies based on their indices, i.e., their ordering in the frequency domain. Thus, each ideal filter removes exactly half the frequencies, but this may correspond to very different ranges of variation, as will be seen next. Also, non-ideal filters are more localized in the vertex domain. We discuss these two advantages of non-ideal filters in what follows.

\subsubsection{Passband Widths}
In the case of bipartite graphs whose variation operator is the symmetric normalized Laplacian, the frequencies are naturally symmetric around $\lambda = 1$, which is not only the center of symmetry for the frequencies but also the middle point of the range of frequency variation.

Instead, in our setting, our spectral folding imposes symmetry. For an arbitrary graph, the exact distribution of frequencies does not exhibit symmetry and there could be more low (high) variation eigenvectors than low (high) variation ones. Thus, the ``ideal'' low-pass filter, i.e., the one preserving the first $N/2$ frequencies, can represent very different variation ranges for different graphs, and thus, it is not guaranteed to be always the best choice of low-pass filter.

As an example, let us assume that the frequency range $[0,\lambda_{\max}/2]$ contains more than half the frequencies (eigenvalues). In this case, an ideal low-pass filter passing through the frequency range  $[0,\lambda_{N/2}]$ would eliminate some of the frequencies in that range, since its cutoff frequency will be at $\lambda_{N/2} < \lambda_{\max}/2$. Non-ideal filters, in contrast, can use a more natural range of frequencies in the decomposition because of the overlapping of the frequency responses for the low-pass and high-pass filters. This allows the low-pass channel to include more of the ``natural" low frequencies, i.e., those having lower variation.

To validate the above discussion numerically, the low-pass filtered signals obtained from GraphSS were compared with those given by the ideal low-pass filter based on the frequency values, i.e.,
\begin{equation}
H_{\text{value}}(\bm{\Lambda}) = \text{diag}(\underbrace{1, \ldots, 1}_{\#(\lambda_{\max}/2)}, 0, \ldots, 0),
\end{equation}
where $\#(\lambda_{\max}/2)$ represents the number of eigenvalues smaller than $\lambda_{\max}/2$. In the following, the signals that were low-pass filtered by $H_0(\bm{\Lambda})$ of GraphSS-X ($\text{X} \in \{\text{I}, \text{O}, \text{B}\}$) are specified as $\widetilde{\mathbf{f}}_{0, \text{X}}$, while those filtered with $H_{\text{value}}(\bm{\Lambda})$ are represented as $\widetilde{\mathbf{f}}_{\text{value}}$.

Signals on two sensor networks with $N = 100$ were used in the experiments. These networks had different sensor distributions, as shown in Fig. \ref{fig:filt_comp}(a) and (b). That is, the vertices of the graph shown in Fig. \ref{fig:filt_comp}(a) are randomly distributed, whereas those in Fig. \ref{fig:filt_comp}(b) have a concentrated region at the bottom left. Hence, the eigenvalue distributions of their combinatorial graph Laplacian are different. Specifically, the graph shown in Fig. \ref{fig:filt_comp}(a) has $\lambda_{N/2}=7.89$ and $\lambda_{\max}/2=7.50$, while that in Fig. \ref{fig:filt_comp}(b) has $\lambda_{N/2}=5.73$ and $\lambda_{\max}/2=15.40$, i.e., $\#(\lambda_{\max}/2) > N/2$.

Despite the difference between the eigenvalue distributions, the spectra of both signals are defined similarly on the basis of the frequency value:
\begin{equation}
\widetilde{f}[i] = \exp(\lambda_i/4) + \epsilon,
\end{equation}
where $\epsilon$ is zero-mean i.i.d. Gaussian noise with standard deviation $\sigma = 0.05$.

The original and filtered spectra are shown in Figs. \ref{fig:filt_comp}(c) and (d). The squared differences between $\widetilde{\mathbf{f}}_{\text{value}}$ and $\widetilde{\mathbf{f}}_{0, \text{X}}$, i.e., $E_{\text{X}}[i] := (\widetilde{f}_{\text{value}}[i] - \widetilde{f}_{0, \text{X}}[i])^2$, are also shown in Figs. \ref{fig:filt_comp}(e) and (f) for a clear visualization.

For the signal on the regularly distributed sensor network, all GraphSSs present similar results in their passband, but the non-ideal filters have sidelobes in the transition bands. In contrast, $\widetilde{\mathbf{f}}_{0, \text{I}}$ is far from $\widetilde{\mathbf{f}}_{\text{value}}$ for the graphs with irregularly distributed sensors, whereas the non-ideal filters have smaller maximum errors than those of GraphSS-I in that case. The errors are numerically compared in Table \ref{tb:comp_pw}. For the signal of Fig. \ref{fig:filt_comp}(a), all low-pass filters of GraphSS show comparable errors with respect to $\widetilde{\mathbf{f}}_{\text{value}}$, whereas GraphSS-I has a larger error than those of GraphSS-O and GraphSS-O for the signal of Fig.~\ref{fig:filt_comp}(b).

\begin{table}[tp]
\centering
\caption{Differences between $\widetilde{\mathbf{f}}_{\text{value}}$ and $\widetilde{\mathbf{f}}_{0, \text{X}}$ (Average of 100 Independent Runs)}
\label{tb:comp_pw}
\begin{tabular}{c||c|c} \hline
Graph signal & Fig. \ref{fig:filt_comp}(a) & Fig. \ref{fig:filt_comp}(b) \\\hline
$||\widetilde{\mathbf{f}}_{\text{value}} - \widetilde{\mathbf{f}}_{0, \text{I}}||_2$ & 0.42 & 0.54\\
$||\widetilde{\mathbf{f}}_{\text{value}} - \widetilde{\mathbf{f}}_{0, \text{O}}||_2$ & 0.42 & 0.27\\
$||\widetilde{\mathbf{f}}_{\text{value}} - \widetilde{\mathbf{f}}_{0, \text{B}}||_2$ & 0.49 & 0.23\\
\hline
\end{tabular}
\end{table}

\begin{figure}[t]
\centering
\subfigure[][Signal on sensor network graph]{\includegraphics[width = .48 \linewidth]{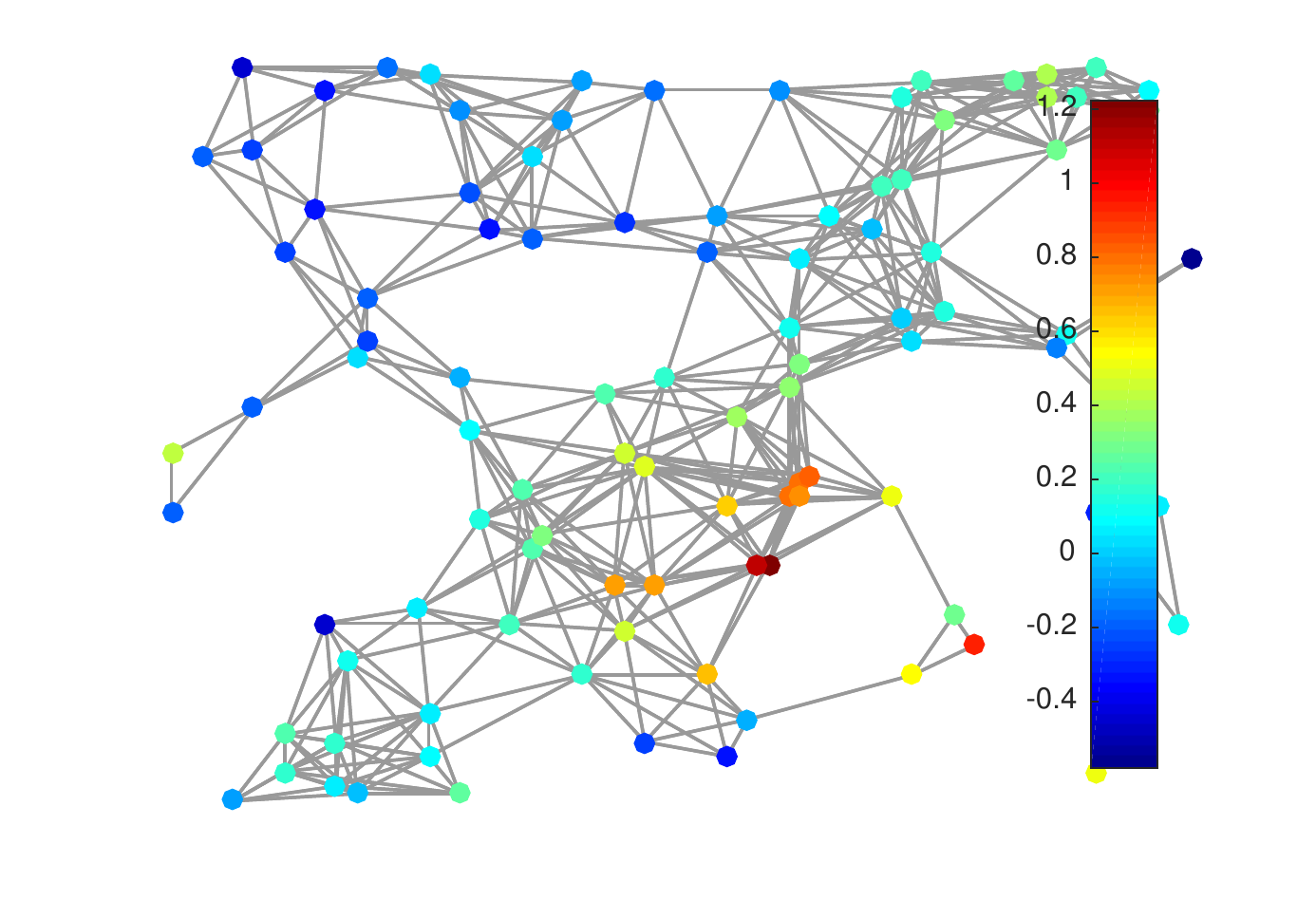}}\ 
\subfigure[][Signal on sensor network graph with concentrated region]{\includegraphics[width = .48 \linewidth]{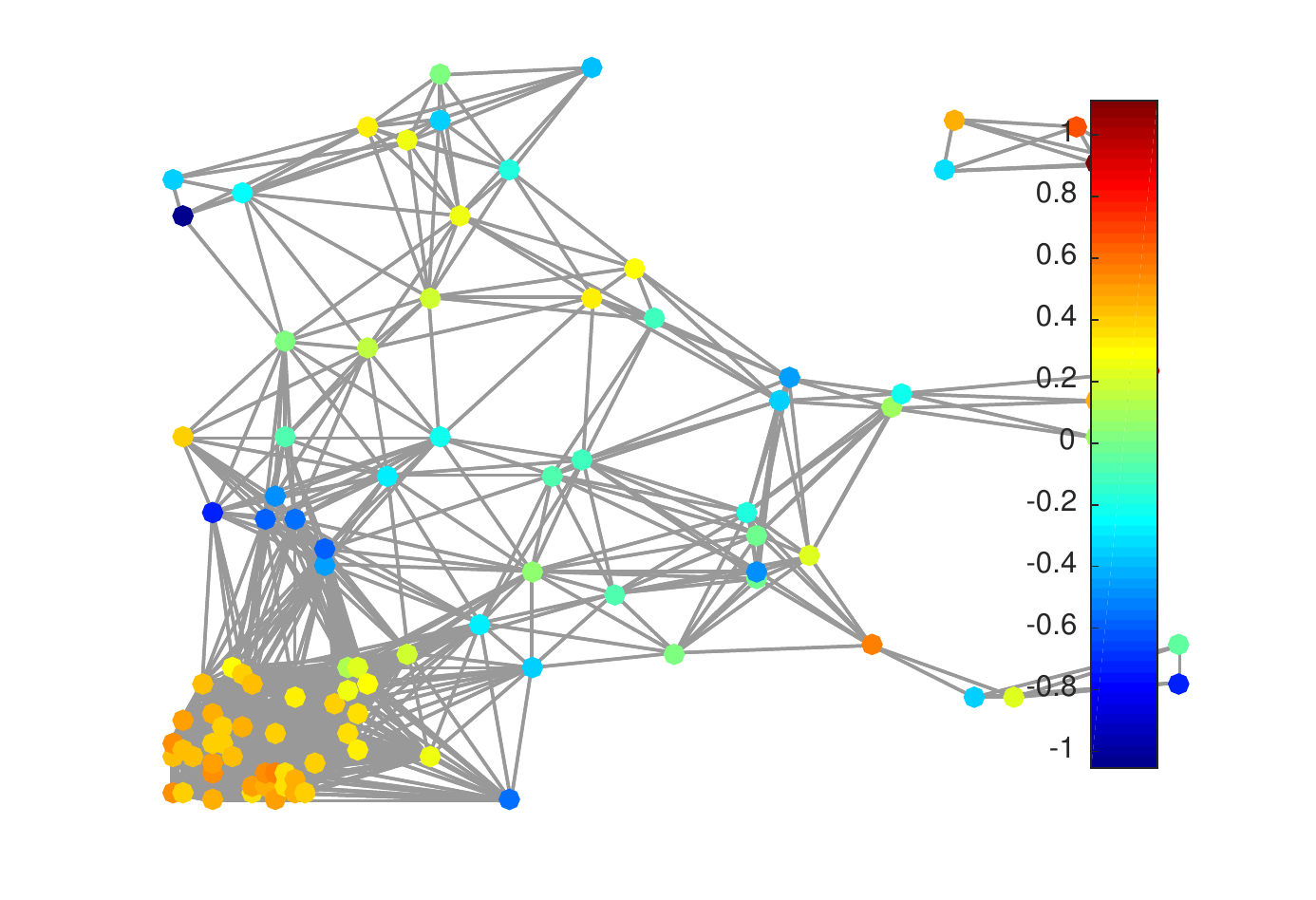}}\\
\subfigure[][Original and filtered spectra of (a)]{\includegraphics[width = .48 \linewidth]{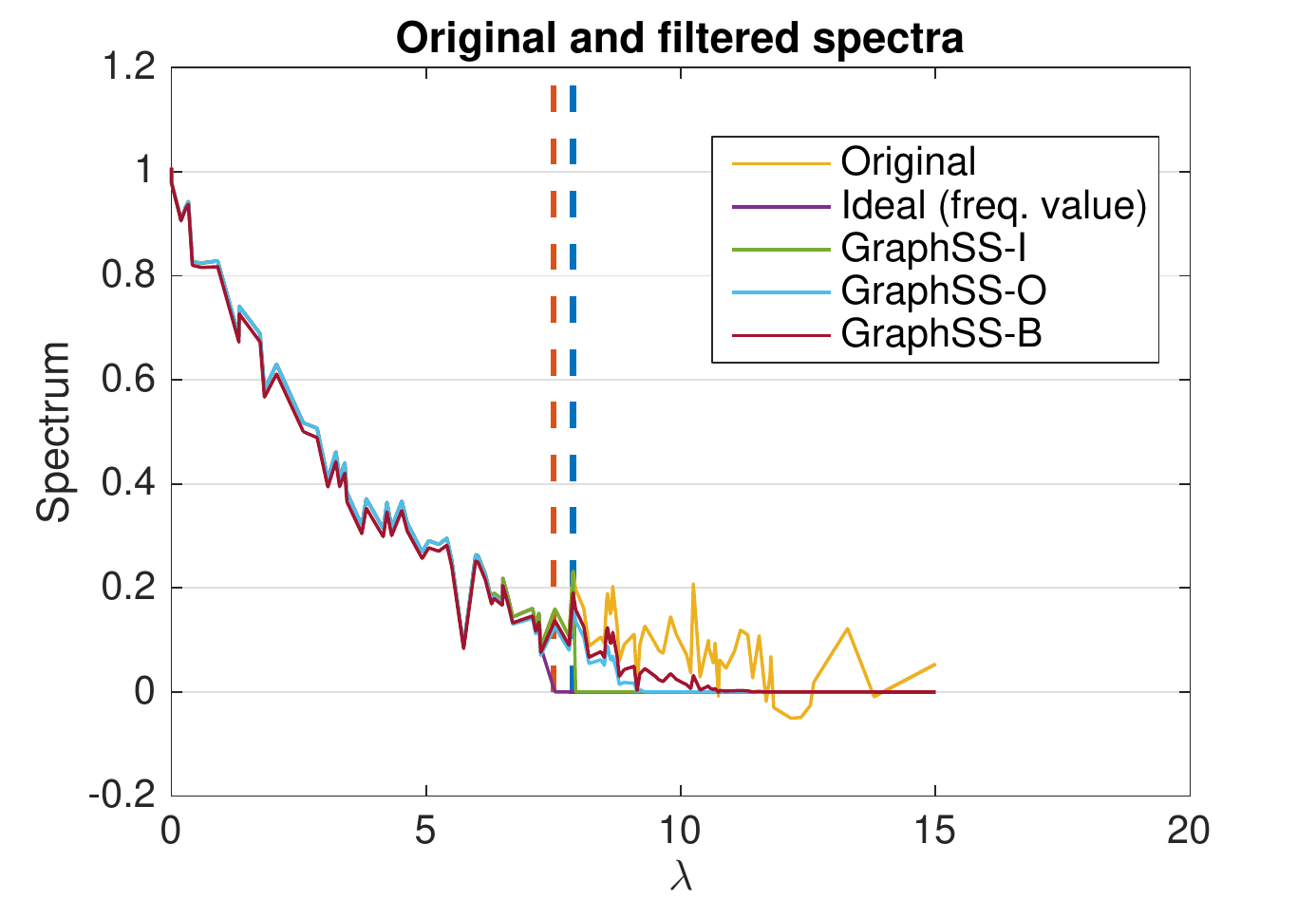}}\ 
\subfigure[][Original and filtered spectra of (b)]{\includegraphics[width = .48 \linewidth]{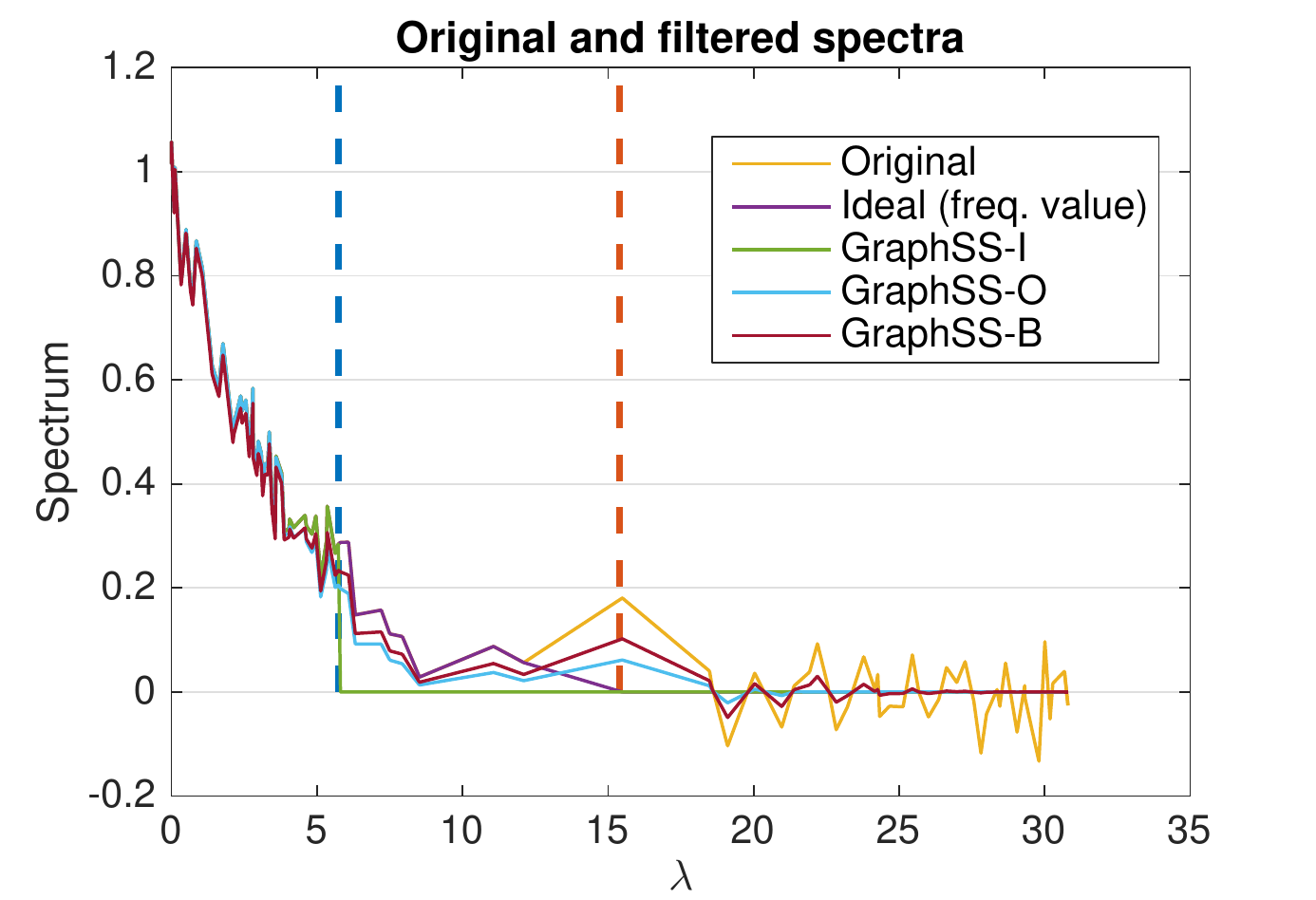}}\\
\subfigure[][Difference between $\widetilde{\mathbf{f}}_{\text{value}}$ and $\widetilde{\mathbf{f}}_{0, \text{X}}$ of (c)]{\includegraphics[width = .48 \linewidth]{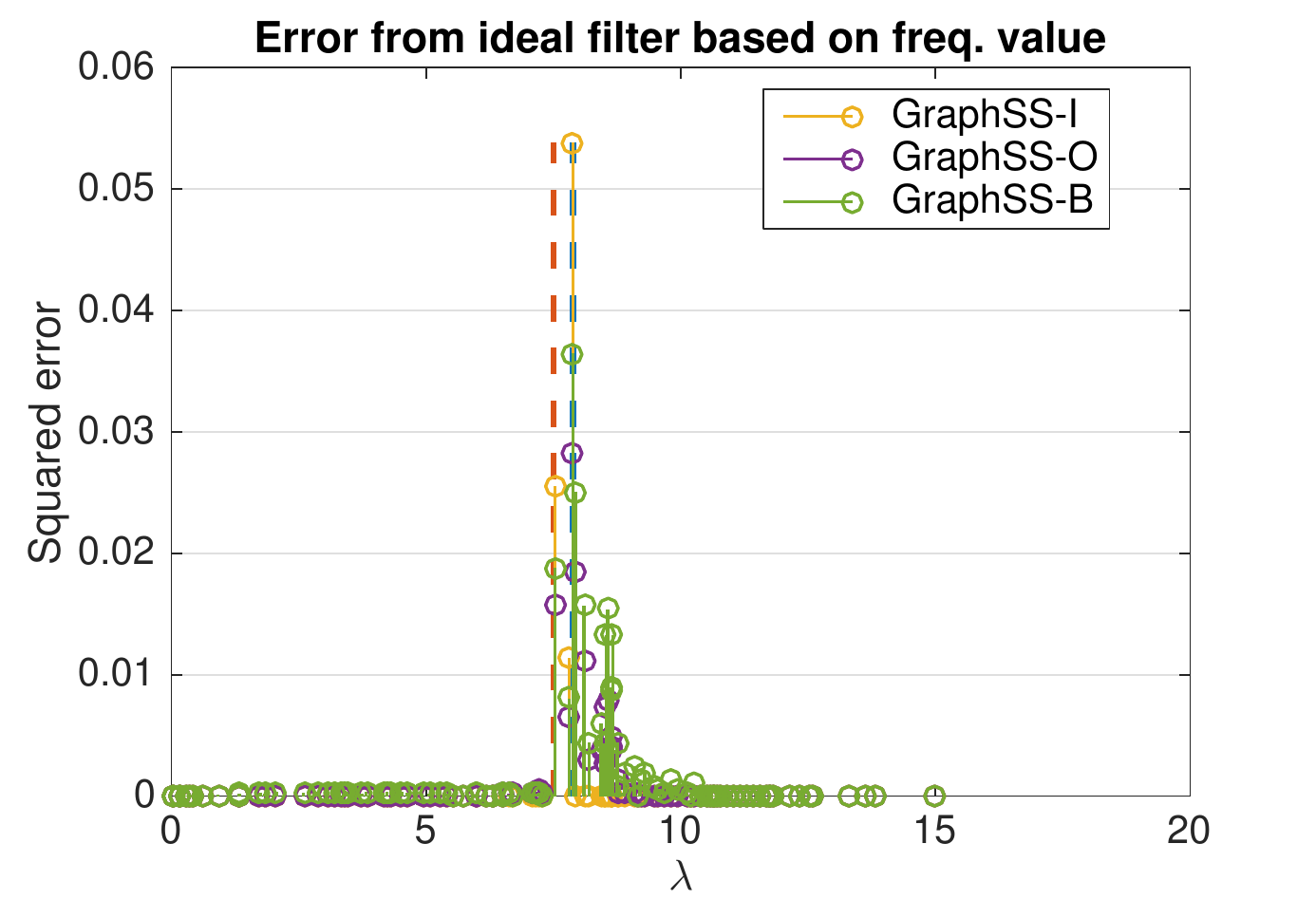}}\ 
\subfigure[][Difference between $\widetilde{\mathbf{f}}_{\text{value}}$ and $\widetilde{\mathbf{f}}_{0, \text{X}}$ of (d)]{\includegraphics[width = .48 \linewidth]{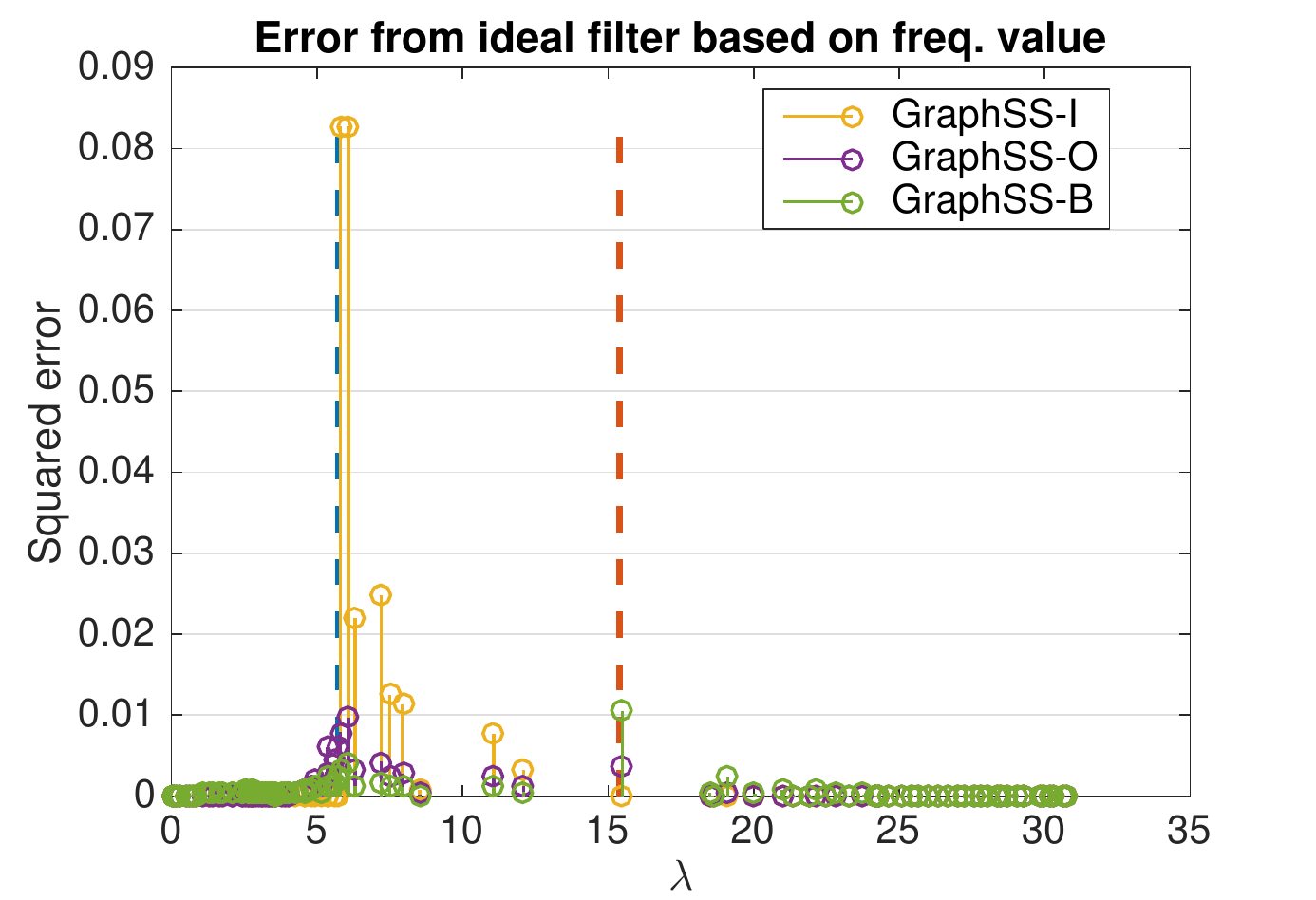}}
\caption{Comparison of filters for GraphSS. The squared error represents $E_{\text{X}}[i] = (\widetilde{f}_{\text{value}}[i] - \widetilde{f}_{0, \text{X}}[i])^2$. Blue and red dashed lines represent $\lambda_{N/2}$ and $\lambda_{\max}/2$, respectively.}
\label{fig:filt_comp}
\end{figure}

\begin{figure}[t]
\centering
\subfigure[][GraphSS-O, $L=1$]{\includegraphics[width = .4 \linewidth]{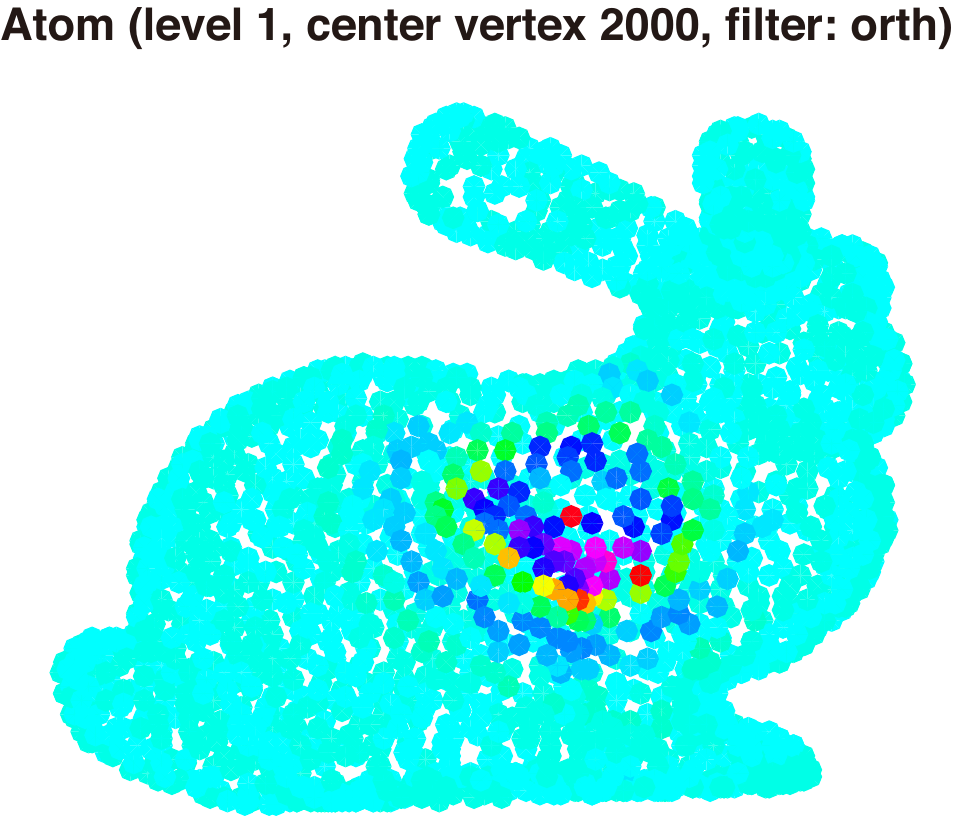}}\quad
\subfigure[][GraphSS-O, $L=3$]{\includegraphics[width = .4 \linewidth]{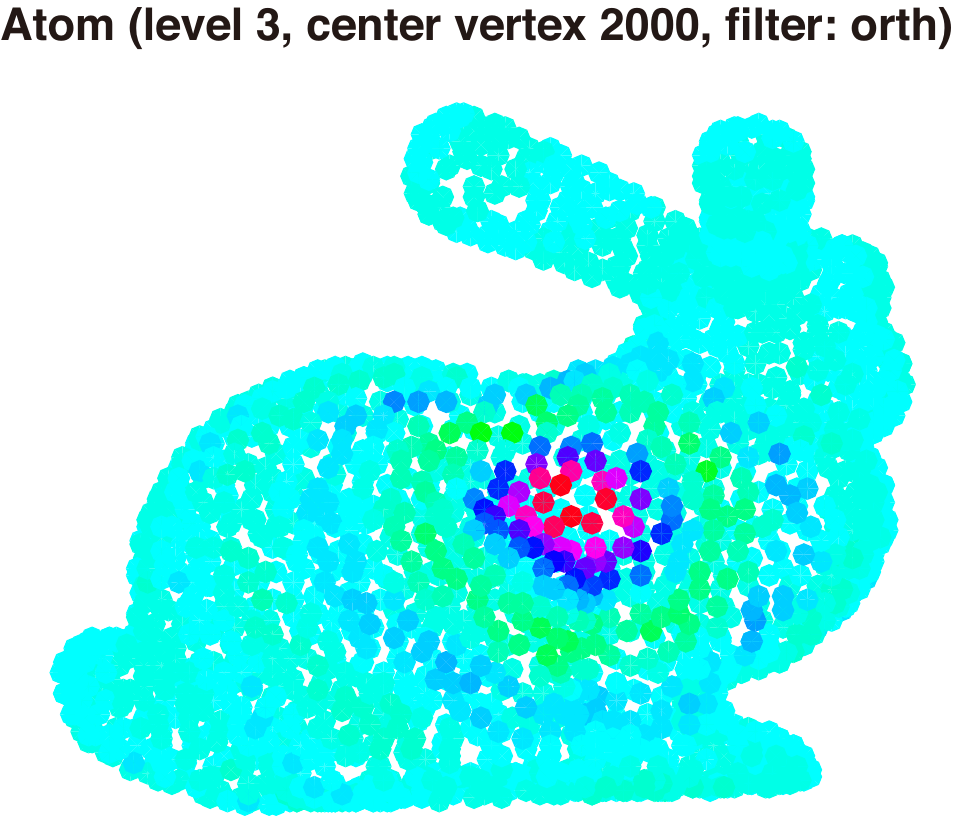}}\\
\subfigure[][GraphSS-B, $L=1$]{\includegraphics[width = .4 \linewidth]{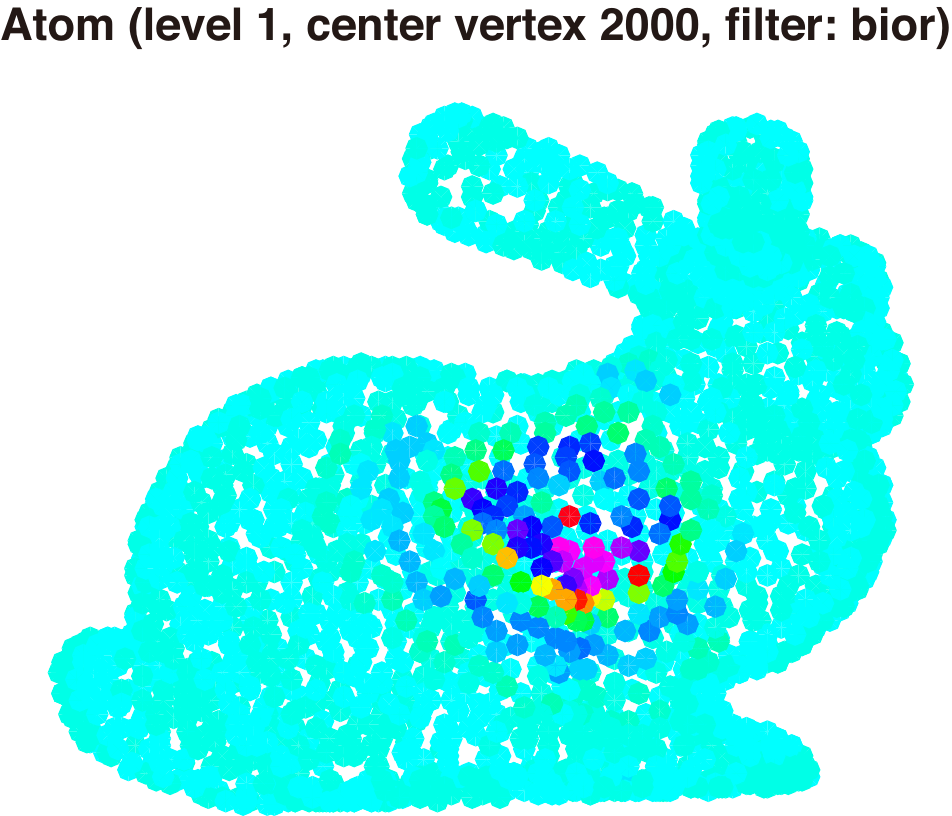}}\quad
\subfigure[][GraphSS-B, $L=3$]{\includegraphics[width = .4 \linewidth]{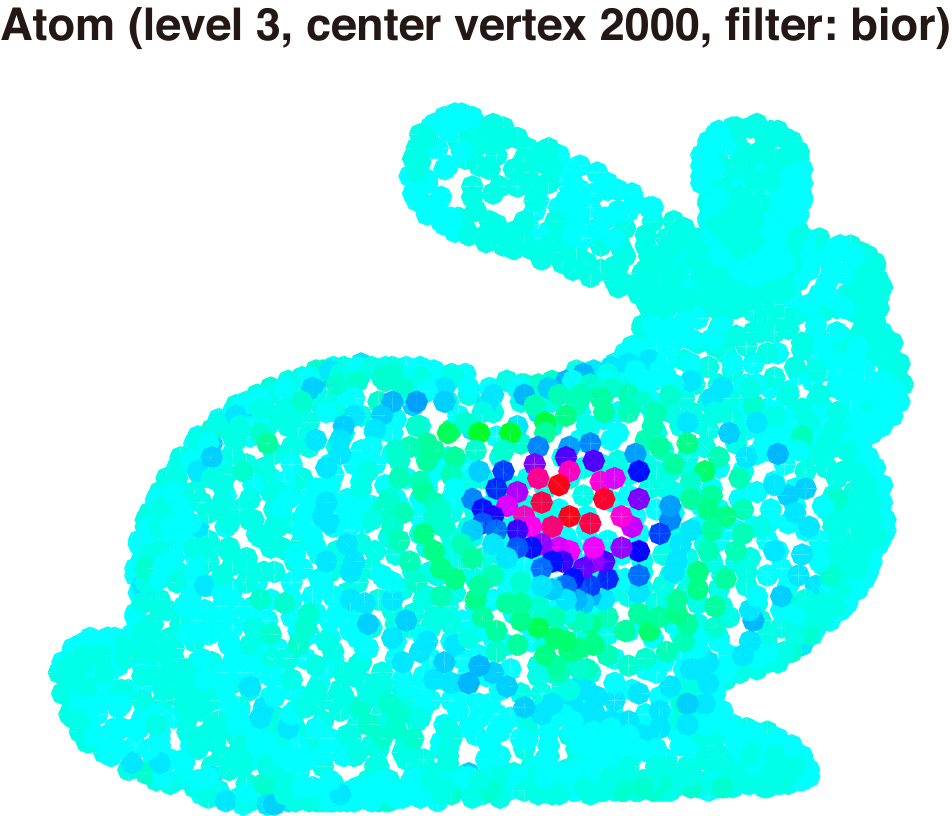}}\\
\subfigure[][GraphSS-I, $L=1$]{\includegraphics[width = .4 \linewidth]{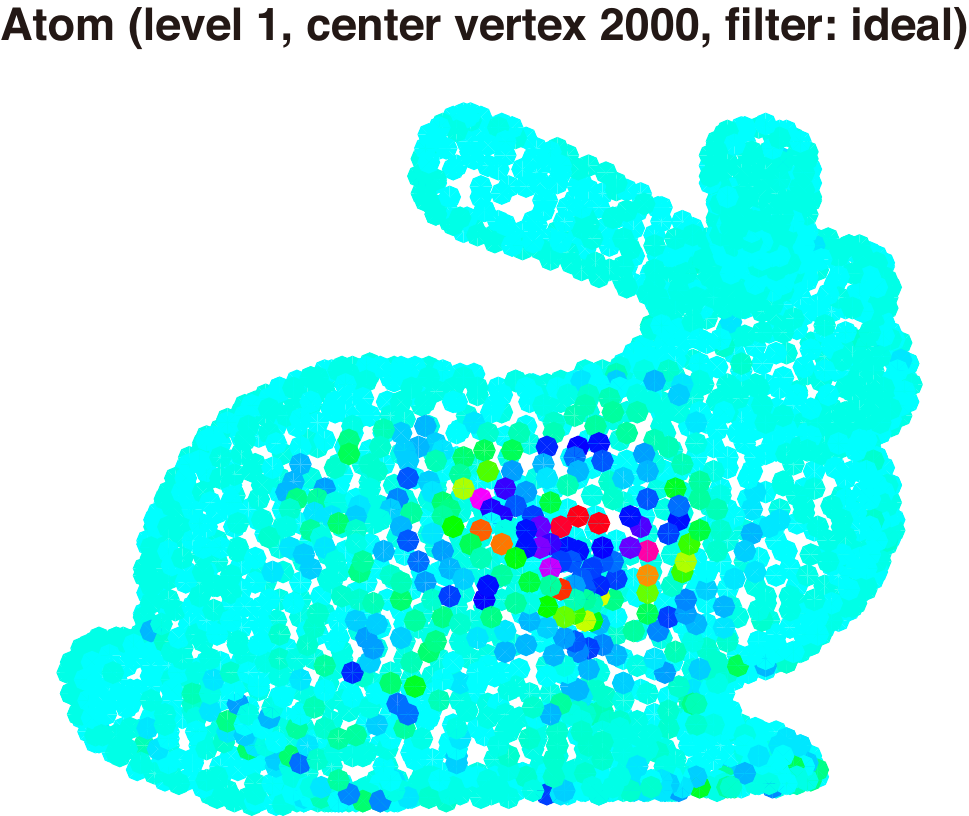}}\quad
\subfigure[][GraphSS-I, $L=3$]{\includegraphics[width = .4 \linewidth]{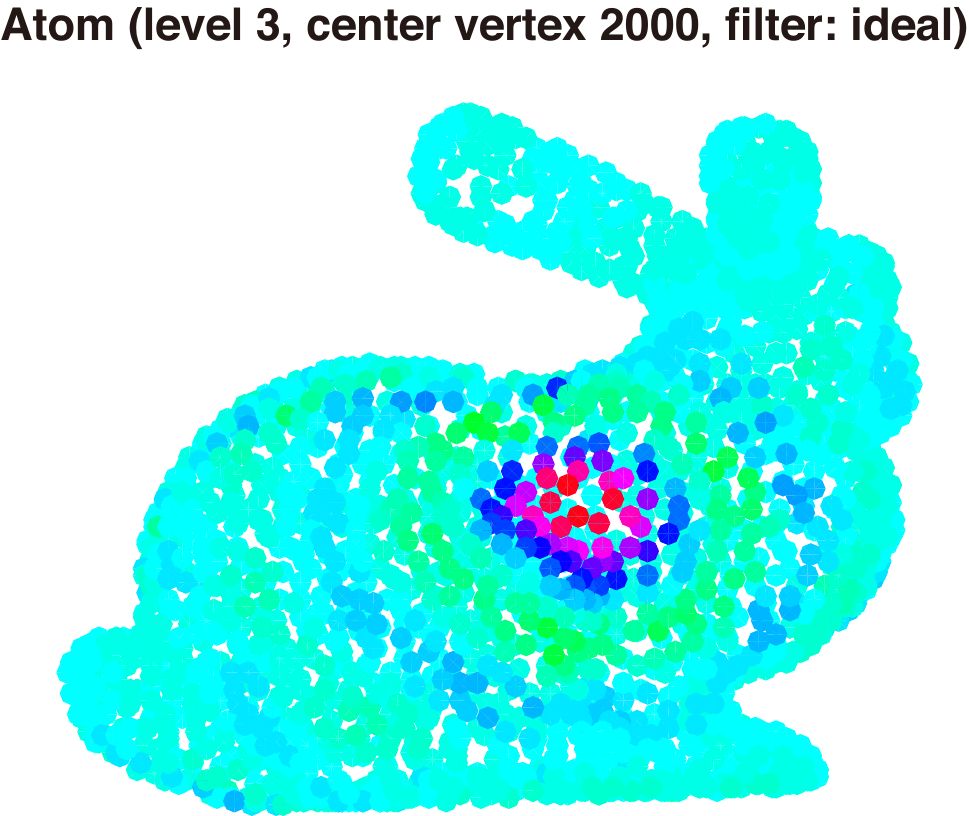}}
\caption{Atoms of low-pass filters for GraphSS on Bunny graph. The center vertex is at $k = 2000$ ($N = 2503$). Left: Atoms for the one-level transforms. Right: Atoms for the three-level transforms. From top to bottom: GraphSS-O, GraphSS-B, GraphSS-I.}
\label{fig:atoms}
\end{figure}

\subsubsection{Filter Localization}
Additionally, we compared the filter localizations in the vertex domain. As described in in Section \ref{subsec:designmethods}, the filters for GraphSS are designed on the basis of frequency values. In general, the filters are global operators in the vertex domain, but their spreads differ depending on the prototype filters.

Fig. \ref{fig:atoms} shows examples of the low-pass filter atoms, in two different scales, at the center vertex $k = 2000$ for the Bunny graph. In the one-level transforms, the atom of GraphSS-I spreads out widely around the center vertex, while those of GraphSS-O and GraphSS-B are more localized. After the three-level transform, all filters have similar spreads in the vertex domain, however, GraphSS-I still has a slightly larger spread than the other two.

The one-level GraphSS-I seems to have a larger spread than the three-level one; This could be due to localizations of the eigenvectors. Different from classical signal processing, the eigenvectors of the variation operator in graph signal processing are sometimes highly localized in the vertex domain \cite{Perrau2018}. In this case, the eigenvectors corresponding to higher graph frequencies (for the one-level transform) would contain large oscillations at vertices apart from the center vertex.

There may be different design methods to accomplish the (approximate) vertex localization for GraphSS; such an investigation would be an interesting topic of study.

\section{Conclusions}\label{sec:VII}
We proposed a new structure of CS GFBs with spectral domain sampling and clarified the perfect reconstruction condition. The structure is a symmetric one like wavelets and filter banks in classical signal processing, but unlike existing GFBs, it enables perfect reconstruction for any graph Laplacians. We also showed the theoretical relationship between the conventional graph wavelets with vertex domain sampling and the proposed ones. In experiments on nonlinear approximation and denoising, our CS GFBs outperformed several other methods. Future work will include devising a fast computation method, $M$-channel filter bank design, and dictionary learning with spectral domain sampling.

\section*{Acknowledgments}
MATLAB code examples are available at http://tanaka.msp-lab.org/software.



\begin{thebibliography}{10}
\providecommand{\url}[1]{#1}
\csname url@samestyle\endcsname
\providecommand{\newblock}{\relax}
\providecommand{\bibinfo}[2]{#2}
\providecommand{\BIBentrySTDinterwordspacing}{\spaceskip=0pt\relax}
\providecommand{\BIBentryALTinterwordstretchfactor}{4}
\providecommand{\BIBentryALTinterwordspacing}{\spaceskip=\fontdimen2\font plus
\BIBentryALTinterwordstretchfactor\fontdimen3\font minus
  \fontdimen4\font\relax}
\providecommand{\BIBforeignlanguage}[2]{{%
\expandafter\ifx\csname l@#1\endcsname\relax
\typeout{** WARNING: IEEEtran.bst: No hyphenation pattern has been}%
\typeout{** loaded for the language `#1'. Using the pattern for}%
\typeout{** the default language instead.}%
\else
\language=\csname l@#1\endcsname
\fi
#2}}
\providecommand{\BIBdecl}{\relax}
\BIBdecl

\bibitem{Shuman2013}
D.~I. Shuman, S.~K. Narang, P.~Frossard, A.~Ortega, and P.~Vandergheynst, ``The
  emerging field of signal processing on graphs: Extending high-dimensional
  data analysis to networks and other irregular domains,'' \emph{{IEEE} Signal
  Process. Mag.}, vol.~30, no.~3, pp. 83--98, Oct. 2013.

\bibitem{Sandry2013}
A.~Sandryhaila and J.~M.~F. Moura, ``Discrete signal processing on graphs,''
  \emph{{IEEE} Trans. Signal Process.}, vol.~61, no.~12, pp. 1644--1656, Apr.
  2013.

\bibitem{Ortega2018}
A.~Ortega, P.~Frossard, J.~Kova{\v{c}}evi{\'c}, J.~M.~F. Moura, and
  P.~Vandergheynst, ``Graph signal processing: {Overview}, challenges, and
  applications,'' \emph{Proc. {IEEE}}, vol. 106, no.~5, pp. 808--828, May 2018.

\bibitem{Leonar2013}
N.~Leonardi and D.~Van De~Ville, ``Tight wavelet frames on multislice graphs,''
  \emph{{IEEE} Trans. Signal Process.}, vol.~16, no.~13, pp. 3357--3367, Jul.
  2013.

\bibitem{Miller2015}
B.~A. Miller, M.~S. Beard, P.~J. Wolfe, and N.~T. Bliss, ``A spectral framework
  for anomalous subgraph detection,'' \emph{{IEEE} Trans. Signal Process.},
  vol.~63, no.~16, pp. 4191--4206, Aug. 2015.

\bibitem{Shahid2016}
N.~Shahid, N.~Perraudin, V.~Kalofolias, G.~Puy, and P.~Vandergheynst, ``Fast
  robust {PCA} on graphs,'' \emph{{IEEE} J. Sel. Topics Signal Process.},
  vol.~10, no.~4, pp. 740--756, Jun. 2016.

\bibitem{Onuki2016}
M.~Onuki, S.~Ono, M.~Yamagishi, and Y.~Tanaka, ``Graph signal denoising via
  trilateral filter on graph spectral domain,'' \emph{{IEEE} Trans. Signal Inf.
  Process. Netw.}, vol.~2, no.~2, pp. 137--148, Jun. 2016.

\bibitem{Pang2017}
J.~Pang and G.~Cheung, ``Graph {Laplacian} regularization for image denoising:
  Analysis in the continuous domain,'' \emph{{IEEE} Trans. Image Process.},
  vol.~26, no.~4, pp. 1770--1785, Apr. 2017.

\bibitem{Yamamo2016}
K.~Yamamoto, M.~Onuki, and Y.~Tanaka, ``Deblurring of point cloud attributes in
  graph spectral domain,'' in \emph{Proc. Int. Conf. Image Process.}, 2016, pp.
  1559--1563.

\bibitem{Cheung2018}
G.~Cheung, E.~Magli, Y.~Tanaka, and M.~Ng, ``Graph spectral image processing,''
  \emph{Proc. {IEEE}}, vol. 106, no.~5, pp. 907--930, May 2018.

\bibitem{Shuman2016}
D.~I. Shuman, M.~J. Faraji, and P.~Vandergheynst, ``A multiscale pyramid
  transform for graph signals,'' \emph{{IEEE} Trans. Signal Process.}, vol.~64,
  no.~8, pp. 2119--2134, Apr. 2016.

\bibitem{Hu2015}
W.~Hu, G.~Cheung, A.~Ortega, and O.~C. Au, ``Multiresolution graph {Fourier}
  transform for compression of piecewise smooth images,'' \emph{{IEEE} Trans.
  Image Process.}, vol.~24, no.~1, pp. 419--433, Jan. 2015.

\bibitem{Liu2017}
X.~Liu, G.~Cheung, X.~Wu, and D.~Zhao, ``Random walk graph {Laplacian}-based
  smoothness prior for soft decoding of {JPEG} images,'' \emph{{IEEE} Trans.
  Image Process.}, vol.~26, no.~2, pp. 509--524, Feb. 2017.

\bibitem{Zhang2014}
J.~Zhang and J.~M.~F. Moura, ``Diffusion in social networks as {SIS} epidemics:
  Beyond full mixing and complete graphs,'' \emph{{IEEE} J. Sel. Topics Signal
  Process.}, vol.~8, no.~4, pp. 537--551, Aug. 2014.

\bibitem{Ono2015}
S.~Ono, I.~Yamada, and I.~Kumazawa, ``Total generalized variation for graph
  signals,'' in \emph{Proc. IEEE Int. Conf. Acoust. Speech, Signal Process.},
  2015, pp. 5456--5460.

\bibitem{Higash2016}
H.~Higashi, T.~M. Rutkowski, T.~Tanaka, and Y.~Tanaka, ``Multilinear
  discriminant analysis with subspace constraints for single-trial
  classification of event-related potentials,'' \emph{{IEEE} J. Sel. Topics
  Signal Process.}, vol.~10, no.~7, pp. 1295--1305, Oct. 2016.

\bibitem{Bronst2017}
M.~M. Bronstein, J.~Bruna, Y.~LeCun, A.~Szlam, and P.~Vandergheynst,
  ``Geometric deep learning: {Going} beyond {euclidean} data,'' \emph{{IEEE}
  Signal Process. Mag.}, vol.~34, no.~4, pp. 18--42, Jul. 2017.

\bibitem{Segarr2016}
S.~Segarra, G.~Mateos, A.~G. Marques, and A.~Ribeiro, ``Blind identification of
  graph filters,'' \emph{{IEEE} Trans. Signal Process.}, vol.~65, no.~5, pp.
  1146--1159, Mar. 2016.

\bibitem{Rustam2013}
R.~Rustamov and L.~Guibas, ``Wavelets on graphs via deep learning,'' in
  \emph{Proc. Adv. Neural Inf. Process. Syst.}, 2013, pp. 998--1006.

\bibitem{Hammon2011}
\BIBentryALTinterwordspacing
D.~K. Hammond, P.~Vandergheynst, and R.~Gribonval, ``Wavelets on graphs via
  spectral graph theory,'' \emph{Applied and Computational Harmonic Analysis},
  vol.~30, no.~2, pp. 129--150, Mar. 2011. [Online]. Available:
  \url{http://wiki.epfl.ch/sgwt}
\BIBentrySTDinterwordspacing

\bibitem{Sakiya2016a}
A.~Sakiyama, K.~Watanabe, and Y.~Tanaka, ``Spectral graph wavelets and filter
  banks with low approximation error,'' \emph{{IEEE} Trans. Signal Inf.
  Process. Netw.}, vol.~2, no.~3, pp. 230--245, Sep. 2016.

\bibitem{Shuman2015}
\BIBentryALTinterwordspacing
D.~I. Shuman, C.~Wiesmeyr, N.~Holighaus, and P.~Vandergheynst,
  ``Spectrum-adapted tight graph wavelet and vertex-frequency frames,''
  \emph{{IEEE} Trans. Signal Process.}, vol.~63, no.~16, pp. 4223--4235, Aug.
  2015. [Online]. Available:
  \url{http://documents.epfl.ch/users/s/sh/shuman/www/publications.html}
\BIBentrySTDinterwordspacing

\bibitem{Tay2015}
D.~B.~H. Tay, Y.~Tanaka, and A.~Sakiyama, ``Near orthogonal oversampled graph
  filter banks,'' \emph{{IEEE} Signal Process. Lett.}, vol.~23, no.~2, pp.
  277--281, Feb. 2015.

\bibitem{Sakiya2014a}
A.~Sakiyama and Y.~Tanaka, ``Oversampled graph {Laplacian} matrix for graph
  filter banks,'' \emph{{IEEE} Trans. Signal Process.}, vol.~62, no.~24, pp.
  6425--6437, Dec. 2014.

\bibitem{Tanaka2014a}
Y.~Tanaka and A.~Sakiyama, ``{$M$}-channel oversampled graph filter banks,''
  \emph{{IEEE} Trans. Signal Process.}, vol.~62, no.~14, pp. 3578--3590, Jul.
  2014.

\bibitem{Narang2012}
\BIBentryALTinterwordspacing
S.~K. Narang and A.~Ortega, ``Perfect reconstruction two-channel wavelet filter
  banks for graph structured data,'' \emph{{IEEE} Trans. Signal Process.},
  vol.~60, no.~6, pp. 2786--2799, Jun. 2012. [Online]. Available:
  \url{http://biron.usc.edu/wiki/index.php/Graph\_Filterbanks}
\BIBentrySTDinterwordspacing

\bibitem{Narang2013}
\BIBentryALTinterwordspacing
------, ``Compact support biorthogonal wavelet filterbanks for arbitrary
  undirected graphs,'' \emph{{IEEE} Trans. Signal Process.}, vol.~61, no.~19,
  pp. 4673--4685, Oct. 2013. [Online]. Available:
  \url{http://biron.usc.edu/wiki/index.php/Graph\_Filterbanks}
\BIBentrySTDinterwordspacing

\bibitem{Tay2017}
D.~B.~H. Tay, Y.~Tanaka, and A.~Sakiyama, ``Almost tight spectral graph
  wavelets with polynomial filters,'' \emph{{IEEE} J. Sel. Topics Signal
  Process.}, vol.~11, no.~6, pp. 812--824, Sep. 2017.

\bibitem{Tay2017a}
------, ``Critically sampled graph filter banks with polynomial filters from
  regular domain filter banks,'' \emph{Signal Processing}, vol. 131, pp.
  66--72, Feb. 2017.

\bibitem{Trembl2016}
N.~Tremblay and P.~Borgnat, ``Subgraph-based filterbanks for graph signals,''
  \emph{{IEEE} Trans. Signal Process.}, vol.~64, no.~15, pp. 3827--3840, Aug.
  2016.

\bibitem{Jin2017}
Y.~Jin and D.~I. Shuman, ``An {$M$}-channel critically sampled filter bank for
  graph signals,'' in \emph{Proc. IEEE Int. Conf. Acoust. Speech, Signal
  Process.}, 2017, pp. 3909--3913.

\bibitem{Ekamba2015}
V.~N. Ekambaram, G.~C. Fanti, B.~Ayazifar, and K.~Ramchandran, ``Spline-like
  wavelet filterbanks for multiresolution analysis of graph-structured data,''
  \emph{{IEEE} Trans. Signal Inf. Process. Netw.}, vol.~1, no.~4, pp. 268--278,
  Dec. 2015.

\bibitem{Teke2016b}
O.~Teke and P.~P. Vaidyanathan, ``Extending classical multirate signal
  processing theory to graphs---{Part II}: {$M$}-channel filter banks,''
  \emph{{IEEE} Trans. Signal Process.}, vol.~65, no.~2, pp. 423--437, Jan.
  2016.

\bibitem{Sakiya2016b}
A.~Sakiyama and Y.~Tanaka, ``Construction of undersampled graph filter banks
  via row subset selection,'' in \emph{Proc. IEEE Global Conf. Signal Inf.
  Process.}, 2016, pp. 322--326.

\bibitem{Vaidya1993}
P.~P. Vaidyanathan, \emph{{Multirate Systems and Filter Banks}}.\hskip 1em plus
  0.5em minus 0.4em\relax NJ: Prentice-Hall, 1993.

\bibitem{Oppenh2009}
A.~V. Oppenheim and R.~W. Schafer, \emph{Discrete-Time Signal Processing},
  3rd~ed.\hskip 1em plus 0.5em minus 0.4em\relax Pearson, 2009.

\bibitem{Vetter2014}
M.~Vetterli, J.~Kova{\v{c}}evi{\'c}, and V.~K. Goyal, \emph{Foundations of
  Signal Processing}.\hskip 1em plus 0.5em minus 0.4em\relax Cambridge
  University Press, 2014.

\bibitem{Pesens2008}
I.~Pesenson, ``Sampling in {Paley--Wiener} spaces on combinatorial graphs,''
  \emph{Transactions of the American Mathematical Society}, vol. 360, no.~10,
  pp. 5603--5627, 2008.

\bibitem{Chen2015}
S.~Chen, R.~Varma, A.~Sandryhaila, and J.~Kova\v{c}evi\'c, ``Discrete signal
  processing on graphs: Sampling theory,'' \emph{{IEEE} Trans. Signal
  Process.}, vol.~63, no.~24, pp. 6510--6523, Dec. 2015.

\bibitem{Wang2015}
X.~Wang, P.~Liu, and Y.~Gu, ``Local-set-based graph signal reconstruction,''
  \emph{{IEEE} Trans. Signal Process.}, vol.~63, no.~9, pp. 2432--2444, May
  2015.

\bibitem{Anis2016}
A.~Anis, A.~Gadde, and A.~Ortega, ``Efficient sampling set selection for
  bandlimited graph signals using graph spectral proxies,'' \emph{{IEEE} Trans.
  Signal Process.}, vol.~64, no.~14, pp. 3775--3789, Jul. 2016.

\bibitem{Tsitsv2016}
M.~Tsitsvero, S.~Barbarossa, and P.~Di~Lorenzo, ``Signals on graphs:
  Uncertainty principle and sampling,'' \emph{{IEEE} Trans. Signal Process.},
  vol.~64, no.~18, pp. 4845--4860, Sep. 2016.

\bibitem{Tanaka2018}
Y.~Tanaka, ``Spectral domain sampling of graph signals,'' \emph{{IEEE} Trans.
  Signal Process.}, vol.~66, no.~14, pp. 3752--3767, Jul. 2018.

\bibitem{Watana2018}
K.~Watanabe, A.~Sakiyama, Y.~Tanaka, and A.~Ortega, ``Critically-sampled graph
  filter banks with spectral domain sampling,'' in \emph{Proc. IEEE Int. Conf.
  Acoust. Speech, Signal Process.}, 2018, pp. 4054--4058.

\bibitem{Chung1997}
F.~R.~K. Chung, \emph{Spectral Graph Theory (CBMS Regional Conference Series in
  Mathematics, No. 92)}.\hskip 1em plus 0.5em minus 0.4em\relax American
  Mathematical Society, 1997.

\bibitem{Deri2017}
J.~A. Deri and J.~M.~F. Moura, ``Spectral projector-based graph {Fourier}
  transforms,'' \emph{{IEEE} J. Sel. Topics Signal Process.}, vol.~11, no.~6,
  pp. 785--795, Sep. 2017.

\bibitem{Giraul2018}
B.~Girault, A.~Ortega, and S.~Narayanan, ``Irregularity-aware graph {Fourier}
  transforms,'' \emph{arXiv preprint arXiv:1802.10220}, 2018.

\bibitem{Narang2009}
S.~K. Narang and A.~Ortega, ``Lifting based wavelet transforms on graphs,'' in
  \emph{Proc. Asia-Pacific Signal Inf. Process. Assoc. Annual Summit Conf.},
  2009, pp. 441--444.

\bibitem{Gavish2010}
M.~Gavish, B.~Nadler, and R.~R. Coifman, ``Multiscale wavelets on trees, graphs
  and high dimensional data: Theory and applications to semi supervised
  learning,'' in \emph{Proc. Int. Conf. Mach. Learn.}, 2010, pp. 367--374.

\bibitem{Anis2017}
A.~Anis and A.~Ortega, ``Critical sampling for wavelet filterbanks on arbitrary
  graphs,'' in \emph{Proc. Int. Conf. Acoust. Speech, Signal Process.}, 2017,
  pp. 3889--3893.

\bibitem{Aspval1984}
B.~Aspvall and J.~R. Gilbert, ``Graph coloring using eigenvalue
  decomposition,'' \emph{SIAM Journal on Algebraic Discrete Methods}, vol.~5,
  no.~4, pp. 526--538, 1984.

\bibitem{Harary1977}
F.~Harary, D.~Hsu, and Z.~Miller, ``The biparticity of a graph,'' \emph{J.
  Graph Theory}, vol.~1, no.~2, pp. 131--133, 1977.

\bibitem{Dorfle2013}
F.~Dorfler and F.~Bullo, ``Kron reduction of graphs with applications to
  electrical networks,'' \emph{{IEEE} Trans. Circuits Syst. {I}}, vol.~60,
  no.~1, pp. 150--163, Jan. 2013.

\bibitem{Nguyen2015}
H.~Q. Nguyen and M.~N. Do, ``Downsampling of signals on graphs via maximum
  spanning trees,'' \emph{{IEEE} Trans. Signal Process.}, vol.~63, no.~1, pp.
  182--191, Jan. 2015.

\bibitem{Narang2010}
S.~K. Narang and A.~Ortega, ``Local two-channel critically sampled filter-banks
  on graphs,'' in \emph{Proc. Int. Conf. Image Process.}, 2010, pp. 333--336.

\bibitem{Ron2011}
D.~Ron, I.~Safro, and A.~Brandt, ``Relaxation-based coarsening and multiscale
  graph organization,'' \emph{Multiscale Modeling \& Simulation}, vol.~9,
  no.~1, pp. 407--423, Sep. 2011.

\bibitem{Cohen1992}
A.~Cohen, I.~Daubechies, and J.-C. Feauveau, ``Biorthogonal bases of compactly
  supported wavelets,'' \emph{Communications on pure and applied mathematics},
  vol.~45, no.~5, pp. 485--560, Jun. 1992.

\bibitem{Strang1996}
G.~Strang and T.~Q. Nguyen, \emph{Wavelets and Filter Banks}.\hskip 1em plus
  0.5em minus 0.4em\relax MA: Wellesley-Cambridge, 1996.

\bibitem{Bellan1976}
M.~Bellanger, G.~Bonnerot, and M.~Coudreuse, ``{Digital filtering by polyphase
  network: Application to sample-rate alteration and filter banks},''
  \emph{{IEEE} Trans. Acoust., Speech, Signal Process.}, vol.~24, no.~2, pp.
  109--114, Apr. 1976.

\bibitem{Tay2017b}
D.~B.~H. Tay and A.~Ortega, ``Bipartite graph filter banks: Polyphase analysis
  and generalization,'' \emph{{IEEE} Trans. Signal Process.}, vol.~65, no.~18,
  pp. 4833--4846, Sep. 2017.

\bibitem{Tanaka2017a}
Y.~Tanaka and A.~Sakiyama, ``A `polyphase' structure of two-channel spectral
  graph wavelets and filter banks,'' in \emph{Proc. IEEE Int. Conf. Acoust.
  Speech, Signal Process.}, 2017, pp. 4144--4148.

\bibitem{LeMag2018}
L.~Le~Magoarou, R.~Gribonval, and N.~Tremblay, ``Approximate fast graph fourier
  transforms via multilayer sparse approximations,'' \emph{{IEEE} Trans. Signal
  Inf. Process. Netw.}, vol.~4, no.~2, pp. 407--420, Jun. 2018.

\bibitem{Lu2017}
K.-S. Lu and A.~Ortega, ``A graph {Laplacian} matrix learning method for fast
  implementation of graph fourier transform,'' in \emph{Proc. IEEE Int. Conf.
  Image Process.}, 2017, pp. 1677--1681.

\bibitem{Perrau2014}
N.~Perraudin, J.~Paratte, D.~I. Shuman, L.~Martin, V.~Kalofolias,
  P.~Vandergheynst, and D.~K. Hammond, ``Gspbox: A toolbox for signal
  processing on graphs,'' \emph{arXiv preprint arXiv:1408.5781}, 2014.

\bibitem{Coifma2006}
R.~R. Coifman and M.~Maggioni, ``Diffusion wavelets,'' \emph{Applied and
  Computational Harmonic Analysis}, vol.~21, no.~1, pp. 53--94, 2006.

\bibitem{Perrau2018}
N.~Perraudin, B.~Ricaud, D.~I. Shuman, and P.~Vandergheynst, ``Global and local
  uncertainty principles for signals on graphs,'' \emph{APSIPA Transactions on
  Signal and Information Processing}, vol.~7, p.~e3, 2018.

\end{thebibliography}

\end{document}